\pdfoutput=1 
\documentclass[11pt]{article}

\usepackage{amssymb,amsmath,amsthm,bbm}
\usepackage{thmtools} % Fix cleverref referring to propositions as theorems

\usepackage[T1]{fontenc}
\usepackage[tt=false, type1=true]{libertine}

\usepackage[libertine]{newtxmath}

\usepackage[margin=1in]{geometry}
\patchcmd{\abstract}{\small}{}{}{}

\usepackage{enumitem}
\setlist{nosep,topsep=0pt,leftmargin=*}

\usepackage{hyperref}

\usepackage[sortcites,sorting=nyt,style=alphabetic,backend=bibtex]{biblatex}
\usepackage{booktabs} % For formal tables
\usepackage[ruled]{algorithm2e} % For algorithms

\usepackage{float}

\usepackage{xcolor,nicefrac,physics,tcolorbox,enumitem}

\usepackage[capitalize]{cleveref}
\usepackage[suppress]{color-edits}
\definecolor{@gray}{HTML}{edc3c5}
\addauthor[Eva]{et}{cyan}
\addauthor[Sid]{sb}{orange}
\addauthor[Giannis]{gf}{teal}
\addauthor[David]{dl}{violet}
\addauthor[\textbf{TODO}]{todo}{purple}

\usepackage{tikz}
\usetikzlibrary{fadings,patterns,shadows.blur,shapes}

\SetAlgorithmName{\textbf{MECHANISM}}{\textbf{MECHANISM}}{List of Mechanisms}
\crefname{algorithm}{Mechanism}{Mechanisms} % Lowercase form
\Crefname{algorithm}{Mechanism}{Mechanisms} % Capitalized form

\DeclareMathOperator*{\E}{\mathbb{E}}
\DeclareMathOperator*{\argmin}{arg\,min}
\DeclareMathOperator*{\argmax}{arg\,max}

\let\hat\widehat

\newcommand{\Bern}{\mathrm{Bernoulli}}
\newcommand{\Binom}{\mathrm{Binomial}}

\renewcommand{\pmb}{\mathbbm}

\newcommand{\lnash}{\lambda_{\mathtt{NASH}}}
\newcommand{\lrob}{\lambda_{\mathtt{ROB}}}
\newcommand{\allPayMechanism}{\hyperref[alg:principal_bids_for_you_all_pay]{\color{black}Competitive Subsidy Mechanism}\xspace}
\newcommand{\allPayMechanismGeneralCost}{{\color{black}General Cost Mechanism}\xspace}
\newcommand{\allPayMechanismWithBiddingMinimum}{\hyperref[alg:principal_bids_for_you_all_pay_with_bidding_minimum]{\color{black}Competitive Subsidy Mechanism with Bidding Minimum}\xspace}
\newcommand{\allPayMechanismWithasymmetricFairShares}{\hyperref[alg:principal_bids_for_you_all_pay_with_asymmetric_fair_shares]{\color{black}Asymmetric Fair Share Mechanism}\xspace}
\newcommand{\paymentConstant}{\bar b}
\newcommand{\oneOverNAggressiveStrategy}{$1/n$-aggressive strategy}
\newcommand{\asymmetricRate}[1]{1-\left(1-\frac1m\right)^{k_{#1}}}
\newcommand{\asymmetricStrategy}[1]{$\left(\asymmetricRate{#1}\right)$-aggressive strategy}

\newtheorem{theorem}{Theorem}[section]

\newtheorem{lemma}[theorem]{Lemma}

\newtheorem{remark}[theorem]{Remark}
\newtheorem{proposition}[theorem]{Proposition}
\newtheorem{numericalresult}[theorem]{Numerical Result}
\newtheorem{fact}[theorem]{Fact}

\theoremstyle{definition}

\newtheorem{definition}{Definition}[section]

\newcommand{\citet}[1]{\textcite{#1}}

\title{
Robust Resource Allocation via Competitive Subsidies
}

\newcommand{\email}[1]{\protect\href{mailto:#1}{#1}}

\author{David X. Lin\thanks{Cornell University (\email{dxl2@cornell.edu}, \email{gfikioris@cs.cornell.edu}, \email{sbanerjee@cornell.edu}, \email{eva.tardos@cornell.edu}).}
    \and Giannis Fikioris\footnotemark[1]
    \and Siddhartha Banerjee\footnotemark[1]    
    \and \'Eva Tardos\footnotemark[1]}
\date{\vspace{-25pt}}

\begin{document}
% \maketitle

% % Title page for title and abstract only.
\begin{titlepage}

\pagenumbering{gobble}  % suppress page number
\maketitle

% Abstract. Note that this must come before \maketitle.
\begin{abstract}
    A canonical setting for non-monetary online resource allocation is one where agents compete over multiple rounds for a single item per round, with i.i.d. valuations and additive utilities across rounds. With $n$ symmetric agents, a natural benchmark for each agent is the utility realized by her favorite $1/n$-fraction of rounds; a line of work has demonstrated one can \emph{robustly} guarantee each agent a constant fraction of this ideal utility, irrespective of how other agents behave. In particular, several mechanisms have been shown to be $1/2$-robust, and recent work established that repeated first-price auctions based on artificial credits have a robustness factor of $0.59$, which cannot be improved beyond $0.6$ using first-price and simple strategies. In contrast, even without strategic considerations, the best achievable factor is $1-1/e\approx 0.63$.

In this work, we break the $0.6$ first-price barrier to get a new $0.625$-robust mechanism, which almost closes the gap to the non-strategic robustness bound. Surprisingly, we do so via a simple auction, where in each round, bidders decide if they ask for the item, and we allocate uniformly at random among those who ask.
The main new ingredient is the idea of \emph{competitive subsidies}, wherein we charge the winning agent an amount in artificial credits that decreases when fewer agents are bidding (specifically, when $k$ agents bid, then the winner pays proportional to $k/(k+1)$, varying the payment by a factor of 2 depending on the competition). Moreover, we show how it can be modified to get an equilibrium strategy with a slightly weaker robust guarantee of $5/(3e) \approx 0.61$ (and the optimal $1-1/e$ factor at equilibrium). 
Finally, we show that our mechanism gives the best possible bound under a wide class of auction-based mechanisms. 

\end{abstract}

% % Optionally include a table of contents
% \vspace{1cm}
\setcounter{tocdepth}{1} % adjust to 1 if desired
% %\gfcomment{Changed TOC counter to fit in first page}
% \tableofcontents

\end{titlepage}

\clearpage
\pagenumbering{arabic}  % start normal page numbers
\setcounter{page}{1}    % ensure next page is page 1

\section{Introduction}

Consider an indivisible resource shared between multiple selfish agents over time, e.g., a telescope shared by different research labs.
Over multiple rounds, a principal needs to decide which agent gets allocated, aiming to allocate to each agent when their value is the highest.
The mechanism used by the principal should be:
$(i)$ non-monetary, as the resource is being shared, and no one agent owns it,
$(ii)$ fair, in that each agent is individually satisfied with the utility they get,
$(iii)$ simple and understandable, so that participating in the mechanism is straightforward,
$(iv)$ robust, i.e., each individual has utility guarantees that hold independent of the behavior of others.
Under these properties, any agent following a simple strategy is able to guarantee high utility regardless of how the other agents behave.
Standard techniques in mechanism design, like the VCG mechanism, cannot be applied due to the lack of money.
In addition, such techniques often ignore fairness in allocation and instead focus on efficiency, i.e., maximizing the total allocated value, which in this setting is ill-defined: the lack of money makes agents' values incomparable.

We study this problem via a canonical model, first introduced by~\citet{guo2010}: $n$ agents compete over $T$ rounds for a single indivisible resource in each round. Agent $i$ has (random) private values $V_i[t]$ for the resource in round $t$; these values are i.i.d. across rounds and independent across agents. The principal can award the item in each round to a single agent, or not award at all.
\sbedit{Early work on this model looked at welfare approximations, with symmetric agents or known value distributions.}
\sbreplace{More recent work on this setting has focused on the design of \emph{robust mechanisms}. In particular,~\citet{gorokh2021remarkable} suggested that each agent first be endowed with a \emph{fair share} indicating their nominal share of the resource.
They then used a repeated first-price auction with an artificial currency to show that any agent can guarantee $1/2$ fraction of her ideal utility (i.e., her maximum utility under her nominal share; see \cref{def:ideal_utility}) in a \textit{robust} way, i.e., under arbitrary behavior by the other agents.}{More recent work on this setting has focused on the design of \emph{share-based} mechanisms, where each agent $i$ is endowed with a `fair share' $\alpha_i$ (with $\sum_i \alpha_i = 1$) indicating the nominal fraction of items they should be allocated. This was first suggested by~\citet{gorokh2021remarkable}, who showed that under a repeated first-price auction with artificial credits, each agent can \emph{robustly} guarantee at least $1/2$ of her ideal utility   (i.e., her maximum utility under her nominal share of the resource -- see \cref{def:ideal_utility}), under arbitrary behavior by other agents.}
Recently, using the same mechanism, \cite{lin2025online} improved this to $2-\sqrt 2 \approx 0.59$ fraction of ideal utility, using a more complicated strategy where the agent has to use a randomized bid; they also show that for the first-price auction, no static bidding strategy can be better than $0.6$ robust.
\sbedit{In contrast, a natural upper bound on the robustness factor is $1 - 1/e$, which follows from ignoring strategic considerations -- consider $n$ agents with equal fair shares $1/n$ and $\Bern(1/n)$ values, wherein each agent should ideally get $T/n$ rounds, but there are at most $(1-(1-1/n)^n)T$ rounds in which any agent has non-$0$ utility. Thus, it appears fundamentally new ideas are required to make progress towards this bound.}

\subsection{Our Results}

Our main result presents a new mechanism where any agent can robustly guarantee at least a $5/8 = 0.625$ fraction of her ideal utility\sbedit{, thus almost closing the gap to the upper bound of $1 - 1/e \approx 0.63$.}
\sbdelete{This bound is very close to the upper bound of $1 - 1/e \approx 0.63$ \sbreplace{that \cite{fikioris2025beyond} proved based on conflicts between the desires of the agents.}{that holds \emph{without} strategic considerations -- consider $n$ agents with equal fair shares $1/n$ and $\Bern(1/n)$ \gfedit{ values}, wherein each agent ideally wants $T/n$ rounds, but can at most get a $0.63T/n$ of these rounds \gfedit{ due to conflicts with others}.}}
In addition, our \sbreplace{mechanism is}{mechanism and strategies turn out to be} simpler than the ones used in prior work on first-price mechanisms~\cite{gorokh2021remarkable,lin2025online}, where each agent needs to bid in a carefully chosen (and potentially randomized) way, so their spending is not too low (to be aggressive enough) and not too high (to conserve budget).
Instead, our proposed \allPayMechanism, abstracts this complexity away from the agent.
As in earlier works, each agent is first endowed with a budget of artificial credits proportional to their fair share.
However, in each round, instead of allowing arbitrary non-negative bids, we require each agent to either request the item or not, and allocate uniformly at random among requesting agents.
Finally, the mechanism charges the winning agent an amount that is based on the competition in that round, subsidizing when fewer agents request.
Specifically, if $k$ agents request, the winner is charged proportionally to $\frac{k}{k+1}$. 
This is intuitive, as winning when more other agents request causes more externalities, and hence higher payments.
Our $0.625$ ideal-utility guarantee is now achieved when an agent with fair share $\alpha$ requests whenever her value is in the top $\alpha$-quantile.

A natural question is whether the payment scheme $\frac{k}{k+1}$ is the optimal one.
In \cref{sec:robust_upper_bound}, assuming static bidding policies, we show how to bound the robustness factor of any payment scheme \sbedit{(i.e., any function $p_k$ charged to the winner when $k$ agents request)} via an optimization problem.
\dlreplace{Solving this numerically, we get that we cannot improve over our $0.625$ result, showing that our simple payment rule is in fact optimal over any comparable payment function in such a mechanism}{Our numerical results indicate that we cannot improve over our $0.625$ result, strongly suggesting that our simple payment rule is in fact optimal over any comparable payment function in such a mechanism}.

A notable benefit of the simplicity of our mechanism is that it is easy to modify to address metrics beyond robustness.
To demonstrate this, we next consider a question raised in~\cite{onyeze2025allocating,lin2025robust} as to whether robust strategies can also be realized as equilibrium actions.
Our proposed strategy of requesting when agent $i$'s value is in her top $\alpha_i$-quantile turns out not to be a best response if all other agents use this strategy (instead of acting adversarially, as we consider in \cref{sec:robust_mechanism}).
In fact, when all agents use this strategy, they do not utilize their entire budget.
However, this is easy to fix by modifying our mechanism to have higher prices.
Specifically, in \cref{sec:equilibrium_mechanism}, we show that when agents have equal fair shares, then increasing our payments by a factor of approximately $\frac38e\approx 1.02$ makes the earlier strategy an equilibrium, albeit with a slightly worse robustness guarantee. In particular, every agent now enjoys a $(1-1/e)\approx 0.63$ fraction of her ideal utility at equilibrium and $\frac{5}{3e} \approx 0.61$ robustly.
%In  we develop a mechanism with good equilibrium properties when the agents' fair shares are equal.
In \cref{sec:asymmetric_fair_shares}, we modify this to realize an equilibrium strategy with the same robustness guarantee for arbitrary fair shares.

The innovations in our mechanism are twofold:
First, our congestion-aware pricing is key to getting robustness better than $1/2$.
In fact, \cite{lin2025robust}, whose pricing is invariant of the number of bidders, show that one cannot do better than $1/2$ with static bidding under this rule.
Secondly, simplifying the action space of the agents is key to our improved guarantees.
Allowing agents to bid any real number in a first-price auction gives an adversary too much freedom and leads to a deterioration of robustness results, as shown by the robustness upper bound of $0.6$ in \cite{lin2025online}.
Even for our equilibrium guarantees, ensuring that a strategy profile results in an equilibrium is very complicated when the action space is too large.

\subsection{Related work}

There is a long line of work on online resource allocation without money, building on the model of~\cite{guo2010}.
Earlier work focused on emulating outcomes of monetary mechanisms without using money~\cite{guo2010,cavallo2014incentive,gorokh2021monetary}. These culminated in the black-box reduction of \citet{gorokh2021monetary}, which, building on the “linking decisions” idea of \citet{jackson}, showed how repeated all-pay auctions can emulate the equilibrium outcome of any monetary mechanism with vanishing efficiency loss. However, these approaches assume full knowledge of value distributions and provide no guarantees under off-equilibrium behavior.

Our work builds on a more recent line~\cite{gorokh2021remarkable,banerjee2023robust,fikioris2025beyond,onyeze2025allocating,lin2025online,lin2025robust}, that considers the same model, but shifts the focus to distribution agnostic mechanisms, both achieving robust, individual-level guarantees, as well as analyzing equilibrium outcomes. \citet{gorokh2021remarkable} first showed that in a first-price auction with artificial currency, each agent has a $1/2$-robust strategy when using an appropriate fixed bid.
\cite{banerjee2023robust} extend their results to reusable resources (where an agent might require the resource for multiple consecutive rounds) using the same mechanism with a reserve price; they also show that in this setting, the $1/2$ factor is tight.
\cite{fikioris2025beyond} use the Dynamic Max-min Fair (DMMF) mechanism (that allocates to the bidding agents with the least (normalized) number of wins) to get $1/2$-robustness in the worst-case; they also get $(1-o(1))$-robustness under assumptions on the agent's value distribution.
\cite{lin2025online} improve the robustness of the first-price mechanism to $2 - \sqrt 2\approx 0.59$ by using randomization, where the agent's bid is uniformly distributed.
In fact, our payment scheme and the uniform distribution used by \cite{lin2025online} share a connection: if $k$ bidders bid using a uniform $[0, 1]$ bid, then the expected payment of the winner is $\frac{k}{k+1}$.
However, \cite{lin2025online} allow arbitrary bidding by the other agents, giving an adversary too much freedom: they prove with this freedom of the adversary and static bidding cannot do better than a $0.6$ fraction of ideal utility.

The issue with the above works is that the robust strategies considered do not result in an equilibrium if used by all players.
In fact, \cite{onyeze2025allocating} proves that even in very simple scenarios, the suggested strategy in the DMMF mechanism does not result in an equilibrium, in fact, DMMF does not have equilibria with static player strategies.
\cite{lin2025robust} design a mechanism that remedies this issue.
They limit each agent $i$ with fair share $\alpha_i$ to at most $\alpha_i T$ requests and, using a complicated randomized allocation scheme, prove that requesting with probability $\alpha_i$ results in $1/2$-robustness and a $\big(1 - \prod_i(1- \alpha_i) \big)$-good equilibrium.
In comparison to our work, their robustness factor $\lrob\etedit{=1/2}$ is lower but their equilibrium factor $\lnash\dledit{=1-\prod_i(1-\alpha_i)}$ can be greater for asymmetric fair shares $(\alpha_i)_i$. They prove that $\lnash = 1-\prod_j (1-\alpha_j)$ is optimal in that no mechanism, even if the principal can see all the values upfront, can guarantee each agent a greater fraction of their ideal utility with the worst-case value distributions. In contrast, our mechanism's $\lnash \geq 1-1/e$ is the best factor that does not depend on the fair shares $\alpha_i$, but also is always worse than the factor of \cite{lin2025robust}, since $\inf_{(\alpha_i)_i}(1-\prod_i(1-\alpha_i))=1-1/e$.

The ideal utility benchmark falls within a wider class of \emph{share-based} approaches to fair allocation problems. A notable parallel notion is that of the AnyPrice Share (APS)~\cite{babaioff2024fair,babaioff2025share}, defined as the value an agent with a given budget can guarantee under \sbedit{any choice of \emph{normalized} item prices (i.e., where the sum of budgets equals the sum of prices)}.  Constant-factor approximations of the APS have been characterized for a variety of full-information one-shot allocation problems; in contrast, we focus on repeated allocation with stochastic private valuations and strategic bidding. The worst-case nature of the APS also makes it much weaker than the ideal utility; for example, an agent with value $1$ for a $1/n$-fraction of items can get a $\Theta(1/n)$-fraction of these in the worst case as all other agents may desire the same items; in contrast, under i.i.d. Bernoulli$(1/n)$ values, she can get \dlreplace{$\approx0.625$}{$\approx 0.63$} fraction of these items.

\section{Preliminaries}
We consider the following canonical setting for repeated online allocation, introduced by \citet{guo2010}. There are $T$ rounds, and in each round $t$, a single indivisible item is available for allocation among $n$ agents. At the start of round $t$, each agent $i$ realizes a private value $V_i[t]$ for the item, drawn from a fixed distribution $\mathcal F_i$, independently across both agents and time. The value $V_i[t]$ becomes known to agent $i$ at the start of round $t$, but is unknown to the principal and other agents. We assume the value distribution $\mathcal F_i$ is nonnegative and bounded by some constant not depending on $T$ (which we take to be 1 without loss of generality). 

At the end of each round $t$, following some mechanism, the principal chooses an agent to allocate the item to, or to not allocate. Let $W_i[t]$ be an indicator of whether agent $i$ wins the round-$t$ item or not, resulting in utility $U_i[t] = V_i[t]W_i[t]$. Each agent seeks to maximize their average per-round utility, $\frac1T\sum_{t=1}^T U_i[t]$.

With symmetric agents (i.e., with identical distributions $\mathcal F_i$ and equal importance), a natural aim for the principal is to allocate to the agent with the highest value. To extend this to heterogeneous agents in the absence of money, we use the benchmark of \textit{ideal utility} introduced by~\citet{gorokh2021remarkable}. Roughly speaking, the ideal utility is the highest expected per-round utility an agent could obtain if they are restricted to winning the item only for a pre-specified fraction of the rounds. 
Formally, we assume each agent has an exogenously given fair share $\alpha_i>0$, where $\sum_i\alpha_i=1$, and define their ideal utility as:
\begin{definition}[Ideal Utility]\label{def:ideal_utility}
The \textit{ideal utility} $v_i^*$ of agent $i$ is the value of the following maximization problem over measurable $\rho:[0,\infty)\to[0,1]$.
\begin{equation} \label{eq:ideal_utility_maximization_problem}
    \max_{\rho}\E_{V_i\sim\mathcal F_i}[V_i\rho(V_i)] \quad\text{subject to}\quad \E_{V_i\sim\mathcal F_i}[\rho(V_i)]\leq\alpha_i
\end{equation}
\end{definition}

An agent's fair share measures an exogenously defined importance of this agent; symmetric fair shares $\alpha_i=1/n$ mean each agent is equally important.
If $\mathcal F_i$ has an absolutely continuous CDF $F_i$, then $v_i^*$ is just the portion of the expectation of $V_i\sim\mathcal F_i$ that comes from the top $\alpha_i$-quantile of $\mathcal F_i$, i.e., $v_i^* = \E_{V_i\sim\mathcal F_i}[V_i\pmb1_{V_i\geq F_i^{-1}(1-\alpha_i)}]$. This is thus a natural benchmark for what an agent can hope to receive. Moreover, with identical $\mathcal F_i$ and equal shares, summing ideal utilities gives the so-called \emph{ex-ante welfare}~\cite{alaei2012bayesian}, which is an upper bound for overall welfare widely used in approximate mechanism design.

Fix a mechanism used to allocate the item. As in prior work, \cite{gorokh2021remarkable,banerjee2023robust,fikioris2025beyond,lin2025online,lin2025robust}, we are interested in \textit{robust} strategies, strategies that guarantee a certain fraction of the agent's ideal utility, regardless of the other agents' strategies (even if the other agents adversarially collude).
\begin{definition}[\bf $\lambda$-robust]
Fix a mechanism and an agent $i$. A strategy $\pi_i$ used by agent $i$ in the mechanism is \textit{$\lambda$-robust} if regardless of the strategies of agents $j\neq i$,
\begin{equation*}
    \frac1T\sum_{t=1}^T \E[U_i[t]] \geq \lambda v_i^*.
\end{equation*}
\end{definition}
Robust strategies are nice in that they guarantee utility for the agent without assumptions the behavior of other agents. In addition, if all agents have a $\lambda$-robust strategy, then at any equilibrium, they must obtain at least a $\lambda$-fraction of their ideal utility; if not, they can deviate to the $\lambda$-robust strategy to gain more utility. In particular, for identical $\mathcal F_i$ and equal shares $\alpha_i=1/n$, if each agent has a $\lambda$-robust strategy, then the ratio of the optimal social walfare and the achieved welfare (price of anarchy) is at most $1/\lambda$.

Our primary focus in this work is on getting mechanisms with good robustness bounds. A secondary goal is to realize such robustness bounds under equilibrium strategies.
Unfortunately, robust strategies are not guaranteed to form an equilibrium, as it is possible that agents may want to deviate to an even higher payoff strategy. Indeed, for natural mechanisms, known robust strategies do not admit any equilibrium~\cite{onyeze2025allocating}. In such cases, a $\lambda$-robust strategy may not help in understanding a mechanism's performance under real agent behavior. To this end, \cite{lin2025robust} offers a mechanism where $\frac12$-robust strategies do form an equilibrium, which motivates the following definition.

\begin{definition}[\bf $\lrob$-robust $\lnash$-good approximate-equilibrium] \label{def:robust_good_equilibrium}
A profile of strategies $(\pi_1,\dots,\pi_n)$ is a \textit{$\lrob$-robust $\lnash$-good approximate-equilibrium} if the following hold.
\begin{enumerate}
\item Each strategy $\pi_i$ is $\lrob$-robust.
\item For some $\epsilon(T) = o(1)$, the profile of strategies $(\pi_1,\dots,\pi_n)$ is an $\epsilon(T)$-equilibrium: no agent can deviate from the strategy profile and gain more than an additive $\epsilon(T)$ in expected per-round utility.\footnote{\dledit{We allow this additive $o(1)$ to account for stochastic deviation.}}
\item In the strategy profile $(\pi_1,\dots,\pi_n)$, each agent obtains at least a $\lnash$-fraction of their ideal utility in expectation.
\end{enumerate}

\end{definition}
By definition, a $\lrob$-robust $\lnash$-good approximate-equilibrium has $\lnash \geq \lrob - \dlreplace{\epsilon}{o(1)}$, but $\lnash$ could potentially be substantially higher. In \cref{sec:equilibrium_mechanism,sec:asymmetric_fair_shares} we improve the result of \cite{lin2025robust} by offering a mechanism where $0.61$-robust strategies form a $(1-1/e)$-good equilibrium.

\section{\texorpdfstring{$0.625$}{0.625}-robustness} \label{sec:robust_mechanism}
In our mechanism, \allPayMechanism, each agent $i$ is endowed with $\alpha_i T$ tokens of artificial credits. At each time $t$, each agent must either \textit{bid} or not.
As we mention in the introduction, having a binary action space limits the possible adversarial behavior of other agents, which is crucial for the improved guarantee over \cite{lin2025online}.
The item is allocated uniformly at random among bidding agents. The winning agent must pay $\paymentConstant k/(1+k)$ artificial tokens where $k$ is the number of bidding agents, and $\paymentConstant$ is a scale factor, chosen later.
As we show in \cref{sec:robust_upper_bound}, this payment scheme is optimal, in that we cannot get better robustness guarantees.
When an agent's budget of artificial tokens becomes non-positive, she is barred from bidding in future rounds.
We formally specify our mechanism in \cref{alg:principal_bids_for_you_all_pay}.

\begin{algorithm}
	\SetAlgoNoLine
	\KwIn{Fair shares $(\alpha_i)_{i\in[n]}$, number of rounds $T$, payment constant $\paymentConstant$}
	Endow each agent with a budget $B_i[1] = \alpha_iT$ of bidding credits\;
	\For{$t=1,2,\dots,T$}{
		Agents either request to bid or not (let $r_i[t]$ be the indicator that agent $i$ requests to bid)\;
		Enforce budgets: $r_i[t]\gets 0$ for each $i$ such that $B_i[t]\leq 0$\;
		Define $S[t] = \{i:r_i[t]=1\}$ to be the set of bidding agents\;
	 	Select a winner uniformly at random from $S[t]$ 
        %(if $S[t]=\emptyset$, the item is not allocated)\;
	 	(let $W_i[t]$ be the indicator agent $i$ wins)\;
		Set payments $P_i[t] = \paymentConstant\cdot\frac{|S[t]|}{1+|S[t]|}\cdot W_i[t]$  (note only the winner pays)\;
		Update budgets: $B_i[t+1]\gets B_i[t]-P_i[t]$
	 }
    \caption{Competitive Subsidy Mechanism}
    \label{alg:principal_bids_for_you_all_pay}
\end{algorithm}

Next, we describe our proposed robust strategy. We adopt the notion of an $\alpha$-aggressive strategy from \cite{fikioris2025beyond}, whereby an agent bids only for the values that realize her ideal utility.

\begin{definition}[$\alpha$-aggressive strategy]
Agent $i$ follows a \textit{$\alpha_i$-aggressive strategy} if she bids whenever her budget is positive and her value $V_i[t]$ is in the top $\alpha_i$-quantile of her value distribution.\footnote{If the CDF $\mathcal{F}_i$ has jumps and the top $\alpha_i$-quantile is not uniquely-defined, then the agent randomizes appropriately at the cutoff to bid with probability exactly $\alpha_i$. Formally, the agent bids at time $t$ with probability $\rho(V_i[t])$ where $\rho$ solves Eqn. \eqref{eq:ideal_utility_maximization_problem}.}
\end{definition}

Our main result of this section is as follows, which says that an $\alpha_i$-aggressive strategy is $(0.625-o(1))$-robust, i.e., each agent $i$ guarantees that fraction of ideal utility regardless of other agents' behavior.

\begin{restatable}{theorem}{robustnessTheorem} \label{thm:robustness_theorem}
When running \allPayMechanism with $\paymentConstant=8/3$, an $\alpha_i$-aggressive strategy is $\lambda_i$-robust for some $\lambda_i\geq \frac{5}{8}-O\left(\sqrt{\frac{\log T}{T}}\right)$.
\end{restatable}
We prove the theorem formally in \cref{sec:app:robust_mechanism_proofs} and give a proof sketch here.

\begin{proof}[Proof Sketch]
When using an $\alpha_i$-aggressive strategy, the expected value of $V_i[t]$ when conditioned on requesting is the expected value of $V_i[t]$ conditioned on $V_i[t]$ being in the top $\alpha_i$-quantile. This is exactly $v_i^*/\alpha_i$ by the definition of the ideal utility $v_i^*$. Therefore, to show $\lambda_i$-robustness, we must show that agent $i$ wins at least an $\approx\lambda_i\alpha_i$ fraction of the rounds.
% Using the same construction as \cite{lin2025online}, it suffices to consider the case where agent $i$'s value distribution is $\Bern(\alpha_i)$. In this case, her ideal utility is $\alpha_i$ and since she bids only when her value is $1$, to show $\lambda_i$-robustness, we must show that she wins at least a $\lambda_i\alpha_i$ fraction of the rounds.

For this sketch, we will only work with expectations, the full proof, that uses standard probability concentration bounds to convert these to high probability statements, is given in \cref{sec:app:robust_mechanism_proofs}.
At a high level, we need to consider two cases: either agent $i$ does not use up all her budget (i.e., $\sum_{t=1}^T P_i[t] < \alpha_iT$), or she runs out of budget (i.e., $\sum_{t=1}^T P_i[t]\geq \alpha_i T$). In the first case, we argue that she gets utility $\frac1T\sum_{t=1}^T W_i[t] \geq \alpha_i(1 - 1/\paymentConstant)$; on the other hand, if she does use up all of her budget, we argue that she gets utility $\frac1T\sum_{t=1}^T W_i[t] \geq 5\alpha_i/(3\paymentConstant)$. By setting $\paymentConstant = 8/3$, we get that the fraction of rounds won is at least $5\alpha_i/8$ in either case.

First, let us consider the case that agent $i$ does not run out of budget. At each time $t$, some number $k$ of agents  $j\neq i$ bid. Let $x_k$ be the fraction of times $t$ that there are $k$ agents $j\neq i$ bidding,
\begin{equation*}
    x_k = \frac1T\cdot\#\big\{t: \#\{j\neq i:r_j[t]=1\}=k\big\}.
\end{equation*}
By the independence of agent $i$'s bid $r_i[t]$ from the others' bids $r_j[t]$ for $j\neq i$, %\dlreplace{when conditioned on the history,}{conditioned on agent $i$ having budget remaining}\etcomment{I am not sure what you mean by "when conditioned on history"}\dlcomment{Yes, it might be better to delete ``when conditioned on the history''. It is not technically true that $r_i[t]$ and $r_j[t]$ are independent: if agent $j$'s strategy is ``bid if and only if agent $i$ has budget left'', then $r_i[t]$ and $r_j[t]$ are not independent: $r_j[t]$ is more likely to be $1$ if $r_i[t]$ is $1$ because $r_i[t] = r_j[t] = 0$ on the whole event that agent $i$ has no budget left at time $t$. But for a proof sketch, maybe this doesn't matter}\etcomment{you are making a good point. Maybe we should keep "conditioned on history the determines if they still have budget to bid"}\dlcomment{Yes, it seems that it suffices to condition on agent $i$ having budget remaining.}
if we restrict to only times $t$ in which agent $i$ bids, $x_k$ is also the fraction of these times $t$ that both agent $i$ bids and there are $k$ agents $j\neq i$ bidding at time $t$. Then focusing on expectations\footnote{Formally, we write  $f(T)\approx g(T)$ if $|f(T)-g(T)| \leq o(1)$ with probability at least $1-o(1)$. Similarly, $f(T)\lessapprox g(T)$ if $f(T) \leq g(T) + o(1)$ with probability at least $1-o(1)$.} we have,
\begin{equation}
\label{eq:handwave_agent_utility_not_running_out_of_budget}
    \frac1T\sum_{t=1}^T W_i[t] \approx\alpha_i\sum_{k=0}^{n-1} \frac{x_k}{1+k},
\end{equation}
since at each time $t$, agent $i$ bids with probability $\alpha_i$, and if $k$ others bid, there will be $1+k$ total bidders, and agent $i$ will win with probability $1/(1+k)$ since the mechanism uniformly at random allocates among bidding agents. Agents' $j\neq i$ budget constraint imply (again focusing on expectations and ignoring $o(1)$ terms)
\begin{equation}
\label{eq:handwave_adversary_budget_constraint_not_running_out_of_budget}
    \sum_{k=1}^{n-1}\left((1-\alpha_i)\cdot\frac{\paymentConstant k}{1+k} + \alpha_i\cdot\frac{\paymentConstant k}{2+k}\right)x_k \lessapprox 1-\alpha_i.
\end{equation}
This is because at any time, with probability $1-\alpha_i$, if $k$ other agents $j\neq i$ bid, they pay in total $\frac{\paymentConstant k}{1+k}$ by the allocation rule, and with probability $\alpha_i$, there are $1+k$ total bidding agents, so some agent $j\neq i$ wins with probability $\frac{k}{1+k}$ in which case they pay $\paymentConstant\cdot\frac{1+k}{2+k}$. Agents' $j\neq i$ total budget is $\sum_{j\neq i}\alpha_jT = (1-\alpha_i)T$, so their total expected per-round payment has to be at most $1-\alpha_i$. %\etcomment{I am a bit lost here. the other agents can also run out of money, and the last round they dont pay the correct amount. Maybe we should just say this more explicitly. Is this why in the inequality about you have $\lessapprox$? } \dlcomment{I've made the mechanism such that agents always pay $\paymentConstant/(1+k)$, even if they don't have that much money left. Then, their budget may go negative. Once they have nonpositive budget, they cannot bid anymore. The $\lessapprox$ is just the $O(\sqrt{T\log T})$ error in the Azuma-Hoeffding inequality when relating the actual payments to expected payments.}\etcomment{good point, maybe that is worth saying? I did put now on top that here we focus on expectations only, but maybe say that we use $\approx$ to talk about the expected value?} \dlcomment{I put some footnotes under $\approx$ and $\lessapprox$ to clarify this.}
Using \eqref{eq:handwave_agent_utility_not_running_out_of_budget} and \eqref{eq:handwave_adversary_budget_constraint_not_running_out_of_budget}, it suffices to lower bound the value of the following linear program.
\begin{align}
    \min_{(x_k)_{k=0}^{n-1}} \quad & \alpha_i\sum_{k=0}^{n-1} \frac{x_k}{1+k} \nonumber \\
    \text{s.t.} \quad & \paymentConstant\sum_{k=1}^{n-1}\left((1-\alpha_i)\cdot\frac{k}{1+k} + \alpha_i\cdot\frac{k}{2+k}\right)x_k \leq 1-\alpha_i \label{eq:handwave_adversary_budget_constraint_agent_does_not_run_out_of_money_in_optimization_problem} \\
    & \sum_{k=0}^{n-1} x_k = 1 \nonumber \\
    & x_k \geq 0 & \forall k \nonumber
\end{align}
Since this is a linear program with two constraints, its minimum is achieved at some $(x_k^*)$ that only has two nonzero coordinates. It is not hard to show that of these nonzero coordinates must be $x_0^*$ if $\paymentConstant\geq 2$ for \eqref{eq:handwave_adversary_budget_constraint_agent_does_not_run_out_of_money_in_optimization_problem} to be satisfied. Letting $x_{k^*}^*$ be the other nonzero coordinate, one can work out that the objective value is at least
\begin{equation*}
    \alpha_i\sum_{k=0}^{n-1}\frac{x_k^*}{1+k} \geq \alpha_i\left(1 - \frac{(2+k^*)(1-\alpha_i)}{\paymentConstant(2+k^*-\alpha_i)}\right) \geq \alpha_i\left(1 - \frac{3(1-\alpha_i)}{3\paymentConstant - \alpha_i\paymentConstant}\right) \geq \alpha_i\left(1-\frac1\paymentConstant\right).
\end{equation*}

Now let us handle the case where agent $i$ does run out of budget. The argument is similar. Let $\tau$ be the time at which the agent uses up all of her budget. Let $x_k$ be the fraction of times $t\leq \tau$ that there are $k$ agents $j\neq i$ bidding. Agent $i$'s number of wins is as before, but we multiply by $\tau/T$ to account for the fact that agent $i$ runs out of budget early:
\begin{equation} \label{eq:handwave_agent_wins_running_out_of_budget}
    \frac1T\sum_{t=1}^T W_i[t] \approx \frac{\tau}{T}\cdot\alpha_i\sum_{k=0}^{n-1} \frac{x_k}{1+k}.
\end{equation}
The budget constraint on agents $j\neq i$ works similarly:
\begin{equation} \label{eq:handwave_adversary_budget_constraint_running_out_of_budget}
    \frac{\tau}{T}\sum_{k=1}^{n-1} \left((1-\alpha_i)\cdot \frac{\paymentConstant k}{1+k} + \alpha_i\cdot \frac{\paymentConstant k}{2+k}\right)x_k \lessapprox 1-\alpha_i.
\end{equation}
To calculate $\tau$, we calculate that
\begin{equation*}
    \frac1T\sum_{t=1}^\tau P_i[t] \approx \alpha_i\paymentConstant\sum_{k=0}^{n-1}\frac{x_k}{2+k},
\end{equation*}
using the payment rule, since at each time $t$, agent $i$ bids with probability $\alpha_i$, and if $k$ other agents $j\neq i$ bid, there are $1+k$ bidding agents total, so agent $i$ wins with probability $\frac{1}{1+k}$ in which case they pay $\paymentConstant\cdot \frac{1+k}{2+k}$. Since agent $i$ has $\alpha_iT$ tokens total,
\begin{equation} \label{eq:handwave_agent_stopping_time}
    \frac\tau T \approx \frac{\alpha_i T}{\sum_{t=1}^\tau P_i[t]} \approx \frac{1}{\paymentConstant\sum_{k=0}^{n-1}\frac{x_k}{2+k}}.
\end{equation}
Substitute \eqref{eq:handwave_agent_stopping_time} into \eqref{eq:handwave_agent_wins_running_out_of_budget} and \eqref{eq:handwave_adversary_budget_constraint_running_out_of_budget} to see that we now must lower bound the value of the following minimization problem.
\begin{align}
    \min_{(x_k)_{k=0}^n} \quad & \alpha_i\cdot\frac{\sum_{k=0}^{n-1} \frac{x_k}{1+k}}{\paymentConstant\sum_{k=0}^{n-1} \frac{x_k}{2+k}}. \nonumber \\
    \text{s.t.} \quad & (1-\alpha_i)\sum_{k=0}^{n-1} \frac{x_k}{2+k} \geq \sum_{k=1}^{n-1}\left((1-\alpha_i)\cdot\frac{k}{1+k} + \alpha_i\cdot\frac{k}{2+k}\right)x_k \nonumber\\
    & \sum_{k=0}^{n-1} x_k = 1 \nonumber \\
    & x_k \geq 0 & \forall k \nonumber
\end{align}
The rest of the argument is similar to the previous case. We argue that this is the minimization of a linear-fractional function with positive denominator, and thus quasi-concave function, over a polytope, so there is an optimal solution $x_k^*$ with only two nonzero coordinates. We then argue that one of these nonzero coordinates must be $x_0^*$, and letting $x_{k^*}^*$ be the other nonzero coordinates, we show the value of this minimization problem is lower bounded by
\begin{equation*}
    \alpha_i\cdot\frac{3+2k^*-\alpha_i}{\paymentConstant(2+k^*-\alpha_i)} \geq \alpha_i\cdot \frac{5-\alpha_i}{\paymentConstant(3-\alpha_i)} \geq \frac{5\alpha_i}{3\paymentConstant},
\end{equation*}
as desired.
\end{proof}
\section{Upper Bound on the Robustness under an Arbitrary Payment Rule} \label{sec:robust_upper_bound}

In this section, we consider payment schemes different than the ones in \cref{sec:robust_mechanism}, and show that the one used in our \allPayMechanism achieves optimal robustness under static strategies. We come up with an optimization problem that bounds this robustness factor and then solve it numerically to get the desired upper bound.

Consider the generalization of the \allPayMechanism where agents pay some $p_k \ge 0$ when there are $k$ bidding agents (with our mechanism being a special case where $p_k = \paymentConstant \cdot k/(1+k)$). We call this generalization \allPayMechanismGeneralCost{} and bound its performance under any payment scheme.

We consider an agent $i$ following the $\alpha_i$-aggressive strategy and all other agents (with a combined budget $(1 - \alpha_i) T$) coordinating so that at each time, exactly $k$ agents other than $i$ bid.
Under such simple strategies, we can explicitly calculate how many rounds agent $i$ wins in expectation, which is equal to the fraction of ideal utility she receives by the definitions of ideal utility and $\alpha_i$-aggressive strategy.
Given the commitment to these strategies, agents $j \ne i$ are picking $k$ to minimize $i$'s utility, while the mechanism picks the payment scheme to maximize it.
The harder part is limiting the search space of this optimization problem.
We consider agents $j \ne i$ only bidding $1$ or $2$ at a time, making only the payments $p_1, p_2, p_3$ relevant for the mechanism designer.
While in principle considering more agents bidding at the same time could harm agent $i$ more, we show below that even with this strategy by the other players, we can argue that the $0.625$ bound is best possible.
Next, we prove upper bounds on these three payments, given that having them too high can make agents run out of budget even without adversarial competition.

This process is shown in the next theorem, where $\mu(p_1, p_2, p_3, k)$ is the number of rounds agent $i$ wins, which is minimized over $k$ and maximized over $p_1, p_2, p_3$.
In the computation of $\mu(\cdot)$, $\gamma$ is the fraction of rounds where agents $j \ne i$ still have budget remaining out of all the rounds that agent $i$ still has budget remaining.
In \cref{sec:app:robust_upper_bound_proofs}, we work out and prove the above computation in detail, which considers agents with equal fair shares $1/n$ and $n \to \infty$.

\begin{restatable}{theorem}{robustnessUpperBound} \label{thm:robustness_upper_bound}
Suppose $\lambda$ is greater than the value of the following optimization problem.
\begin{equation*}
    \max_{\substack{0< p_1\leq 2e\\0\leq p_2\leq 4e\\0\leq p_3\leq 12e}} \min_{k\in\{1,2\}}\mu(p_1,p_2,p_3,k)
\end{equation*}
where
\begin{align*}
    \mu(p_1,p_2,p_3,k) := \max_{\gamma} & \left(\left(\frac{\gamma}{1+k} + (1-\gamma)\right)\min\left\{1, \frac{1}{\frac{p_{k+1}}{k+1}\cdot\gamma + p_1(1-\gamma)}\right\}\right)\\
    \mathrm{s.t.}&  \quad 1\geq \gamma \geq \min\left\{1, \max\left\{\frac{1}{p_k},\frac{p_1}{p_k - \frac{p_{k+1}}{k+1} + p_1} \right\}\right\}\\
	& \quad \gamma\in\left\{\min\left\{1, \max\left\{\frac{1}{p_k},\frac{p_1}{p_k - \frac{p_{k+1}}{k+1} + p_1} \right\}\right\}, \frac{p_1}{\frac{p_{k+1}}{k+1}-p_1}\right\}.
\end{align*}
Then, there exists a number of players $n$ such that with equal fair shares $\alpha_i=1/n$, an $1/n$-aggressive strategy has a robustness factor at most $\lambda + O\left(\sqrt{\nicefrac{\log T}{T}}\right)$ in \allPayMechanismGeneralCost{} no matter the choice of the payment scheme $(p_k)_k$.
\end{restatable}

% \dldelete{The upper bound comes from upper bounding agent $i$'s utility when agents  $j\neq i$ use the following strategy. At the beginning of the mechanism, they fix a $k\in\{1,2\}$. At each time $t$, $k$ agents other than $i$ bid, until all agents $j\neq i$ are out of budget. Due to the simple strategy of both agent $i$ and the agents $j\neq i$, analyzing the utility of agent $i$ is relatively simple. The various maximums and minimums in the theorem statement come from casework: the agent can either run out of budget or not at some time, the adversary can either run out of budget or not at some time, if both the agent and adversary run out of budget, either the agent runs out of budget first or the adversary runs out of budget first.}

We numerically brute-force optimized the optimization problem in \cref{thm:robustness_upper_bound} by discretizing the space of $(p_1,p_2,p_3)$ and evaluating the objective function at each $(p_1,p_2,p_3)$ in the discretized space. In our code, we found no $(p_1,p_2,p_3)$ such that the objective is more than $0.625$, giving numerical evidence that the value of the optimization problem is $0.625$.\footnote{Our code can be found at \url{https://github.com/davidxlin/optimization-problem-for-robustness}}

\begin{numericalresult}
By numeric calculations, the expression of \cref{thm:robustness_upper_bound} is at most $0.625$, which indicates that there is no payment rule $(p_k)_k$ that makes the $\alpha_i$-aggressive strategy better than $\left(0.625 + \omega\left(\sqrt{\nicefrac{\log T}{T}}\right)\right)$-robust for all $n$.
\end{numericalresult}

\section{Making the Robust Strategies Form an Equilibrium} \label{sec:equilibrium_mechanism}

While robust strategies like in the previous section are nice, they do not always form an equilibrium. Previous work \cite{gorokh2021remarkable,fikioris2025beyond,banerjee2023robust,lin2025online} proposed robust strategies in different mechanisms for this setting that do not form an equilibrium.
On the other hand, \cite{onyeze2025allocating} designed strategies that form an equilibrium but offer no robustness.
\citet{lin2025robust} achieve both types of results by designing a mechanism that has a $1/2$-robust $\left(1 - \prod_{j\in[n]}(1-\alpha_j)\right)$-good equilibrium (see \cref{def:robust_good_equilibrium}). They also note that the $1 - \prod_j(1-\alpha_j)$ factor is optimal in that even if the principal is able to see all the values, they cannot guarantee a higher fraction of ideal utility than $1 - \prod_j(1-\alpha_j)$ to every agent in expectation. 

In this section, we explore if our mechanism achieves similar results.
First, we shall argue that in our mechanism, if we choose $\paymentConstant$ as suggested by the previous section, each player using an $\alpha_i$-aggressive strategy does not form an equilibrium. However, if we choose $\paymentConstant$ slightly differently and make a slight modification of the mechanism by forcing players to bid at least an $\alpha_i$ fraction of the time, we can make the $\alpha_i$-aggressive strategies form an equilibrium at the expense of sacrificing the robustness factor a little bit. Here, we demonstrate how to make this modification to $\paymentConstant$ in the symmetric case where $\alpha_i=1/n$. In \cref{sec:asymmetric_fair_shares}, we give a reduction from the general $\alpha_i$ case to the symmetric case.

First, let us argue why when choosing the payment constant $\paymentConstant = 8/3$ as in the previous section, everyone playing the robust strategy is not necessarily an equilibrium. Assume every agent is playing the robust strategy, so they each bid with probability $1/n$ independently across agents and time, so long as they have budget. Then, at any given time that everyone has budget remaining, for any agent $i$, the number of other agents $j\neq i$ that bid is distributed as $X\sim\Binom(n-1,1/n)$. Therefore, conditioned on agent $i$ bidding, agent $i$'s expected payment is $\E\left[\frac{\paymentConstant}{2+X}\right]$ since conditioned on the number of other agents bidding $X$, agent $i$ wins with probability $\frac{1}{1+X}$, in which case they pay $\paymentConstant\cdot\frac{1+X}{2+X}$. By substituting in $\paymentConstant = 8/3$ and computing $\E[1/(2+X)]$, which we do in \cref{prop:expectation_of_functions_of_binom}, we can show that this expected payment is less than $1$ for $n\geq 3$. Since this is less than $1$, agents have $T/n$ budget, and are bidding with probability $1/n$, this means that agents will have $\Omega(T)$ budget remaining at the end of the mechanism with high probability. Clearly, for value distributions that are nonzero with probability greater than $1/n$, the agents are not best responding: they should bid more to use more of their budget.

However, the above calculation suggests the following idea. We should set $\paymentConstant = 1/\E[1/(2+X)]$, and then agents will be spending their budget exactly in expectation. We still obtain high robustness beating all previous work with this choice of $\paymentConstant = 1/\E[1/(2+X)]  = (n+1)/(1 + n(1-1/n)^{n+1})\approx e$: the robustness factor (using the calculations in the proof of \cref{thm:robustness_theorem}) becomes
\begin{equation*}
    \min\left\{1 - \frac{3(1-1/n)}{3\paymentConstant - \paymentConstant/n}, \frac{5-1/n}{\paymentConstant(3-1/n)}\right\} \geq \frac{5}{3e} \approx 0.61.
\end{equation*}

Notice that the mechanism allocates the item so long as there is at least one bidder. Each bidder playing a \oneOverNAggressiveStrategy{} bids with probability $1/n$ as long as they have budget remaining. If everyone plays a \oneOverNAggressiveStrategy{}, by the fact that agents spend their budget exactly in expectation and by concentration inequalities, no agent runs out of budget too early with high probability. Then, there are $1-(1-1/n)^n$ rounds in which the item is allocated in expectation. %\etcomment{I think there is something wrong in the phrasing here. the mechanism is preventing agents to request when out of money, right?}\dlcomment{Yes, hopefully my fix makes it better.}\etcomment{Thank you} 
By the strategies and symmetry, this will be the fraction of ideal utility each agent achieves, matching the optimal as shown in \cite{lin2025robust}.

What is left to do is to argue that choosing $\paymentConstant$ in this way is indeed an equilibrium. Intuitively, for a fixed agent $i$, they have two potential deviations: they could bid more often, or they could bid less often. By bidding more often, agent $i$ spends more artificial currency. We can also see that agent $i$ bidding more often decreases the spending of other agents as follows. In a given round, conditioned on there being $k$ bidders, the expected payment of a bidder is $\paymentConstant/(k+1)$: each bidder wins with probability $1/k$ in which case they pay $\paymentConstant k/(k+1)$. Hence, fixing the strategies of agents $j\neq i$, agent $i$ bidding more often increases the number of bidders $k$ at each timestep, lowering the other agents' payments. Therefore, agent $i$ will run out of money quicker with the same probability of winning each round (conditioned on bidding), so they will obtain more lower-valued rounds at the beginning instead of spacing their wins.

Whether or not agent $i$ should bid less is slightly more nuanced. In fact, if we retain the same mechanism as in \allPayMechanism, agent $i$ may want to bid less. For example, if agent $i$'s value distribution were $\mathcal F_i = \Bern(1/(2n))$, then a \oneOverNAggressiveStrategy{} would imply that agent $i$ is bidding sometimes when they have $0$ value. Instead, they should not bid when they have $0$ value, which would cause others' payments to go up (using the previous reasoning about the other agents' spending), and then when other agents run out of budget, agent $i$ will have a higher probability of winning on rounds that they actually have value $1$.

To accommodate this, we modify the mechanism to enforce that at each time $t$, each agent must have bid at least $t/n - o(T)$ times; otherwise, the principal will force them to bid. Then, attempting to underbid is not a helpful strategy. We describe the mechanism formally in \cref{alg:principal_bids_for_you_all_pay_with_bidding_minimum}, where the only difference from \cref{alg:principal_bids_for_you_all_pay} is \hyperref[line:minimum_bidding_requirement]{the enforcement of the minimum bidding requirement}.

\begin{algorithm}
\SetAlgoNoLine
\KwIn{Number of rounds $T$, payment constant $\paymentConstant$, Underbidding allowance $\epsilon = o(T)$ }
Endow each agent with a budget $B_i[1] = T/n$ of artificial credits\;
\For{$t=1,2,\dots,T$}{
    %Agents can either request to bid or not\;
    %Let $r_i[t]$ be the indicator that agent $i$ requests to bid\;
    Endow each agent with a budget $B_i[1] = \alpha_iT$ of bidding credits\;
    \For{$t=1,2,\dots,T$}
    {
    Agents either request to bid or not (let $r_i[t]$ be the indicator that agent $i$ requests to bid)\;	    
    %\textbf{Bidding minimums are enforced: $r_i[t]\gets 1$ for each $i$ such that $\sum_{s=1}^t r_i[s] \leq \frac tn - \epsilon$} \label{line:minimum_bidding_requirement}\;
    \textbf{Enforce bidding minimums: $r_i[t]\gets 1$ for each $i$ such that $\sum_{s=1}^t r_i[s] \leq \frac tn - \epsilon$} \label{line:minimum_bidding_requirement}\;
    %Budgets are enforced: $r_i[t]\gets 0$ for each $i$ such that $B_i[t]\leq 0$\;
    %Let $S[t] = \{i:r_i[t]=1\}$ be the set of bidding agents (if $S[t]=\emptyset$, then the item is not allocated)\;
    %A winner is uniformly at random selected from $S[t]$\;
    %Let $W_i[t]$ be the indicator that agent $i$ wins\;
    %Set $P_i[t] = \paymentConstant\cdot\frac{|S[t]|}{1+|S[t]|}\cdot W_i[t]$ to be the payment for the winner\;
    %Budgets are updated: $B_i[t+1]\gets B_i[t]-P_i[t]$
    Enforce budgets: $r_i[t]\gets 0$ for each $i$ such that $B_i[t]\leq 0$\;
    Define $S[t] = \{i:r_i[t]=1\}$ to be the set of bidding agents\;
    Select a winner uniformly at random from $S[t]$ 
    (let $W_i[t]$ be the indicator agent $i$ wins)\;
    Set payments $P_i[t] = \paymentConstant\cdot\frac{|S[t]|}{1+|S[t]|}\cdot W_i[t]$  (note only the winner pays)\;
    Update budgets: $B_i[t+1]\gets B_i[t]-P_i[t]$}
 }
    \caption{Competitive Subsidy Mechanism with Bidding Minimum}
    \label{alg:principal_bids_for_you_all_pay_with_bidding_minimum}
\end{algorithm}

Our main result of this section is as follows, proved formally in \cref{sec:app:equilibrium_mechanism_proofs}.
\begin{restatable}{theorem}{equilibriumTheorem} \label{thm:equilibrium_mechanism}
Consider \allPayMechanismWithBiddingMinimum{} with payment constant $\paymentConstant = (n+1)/(1+n(1-1/n)^{n+1})$ and underbidding allowance $\epsilon = \sqrt{T\log T}$. Then, when players have equal shares, every agent playing a \oneOverNAggressiveStrategy{} is a $\lrob$-robust $\lnash$-good approximate-equilibrium for some $\lrob$ and $\lnash$ satisfying
\begin{equation*}
    \lrob \geq \frac{5}{3e}-O\left(\sqrt{\frac{\log T}{T}}\right),\; \lnash \geq 1-\left(1-\frac1n\right)^n - O\left(\sqrt{\frac{\log T}{T}}\right).
\end{equation*}

% a \oneOverNAggressiveStrategy{} is $\lrob$-robust for
% \begin{equation*}
% 	\lrob \geq \frac{5}{3e} - O\left(\sqrt{\frac{\log T}{T}}\right).
% \end{equation*}
% Furthermore, each agent playing a \oneOverNAggressiveStrategy{} is an $O\left(\sqrt{\frac{\log T}{T}}\right)$-equilibrium\footnote{An $\epsilon$-equilibrium is a profile of strategies such that no agent can obtain more than an additive $\epsilon$ utility by deviating.} in which each agent obtains at least a $\lnash$-fraction of their ideal utility with probability at least $1-O(1/T^2)$, where
% \begin{equation*}
% 	\lnash \geq 1-\left(1-\frac1n\right)^n - O\left(\sqrt{\frac{\log T}{T}}\right).
% \end{equation*}
\end{restatable}

\section{Robust Equilibrium with Asymmetric Fair Shares} \label{sec:asymmetric_fair_shares}

In this section, we generalize \cref{sec:equilibrium_mechanism} to asymmetric fair shares. 
The issue with \allPayMechanism{} is that it allocates the item uniformly at random among the bidding agents.
While this is not an issue for robustness claims (all such results of previous sections hold for arbitrary fair shares), it is an issue for the equilibrium claim.
In fact, to get better than $1/2$ ideal utility guarantees in equilibrium, we have to allocate the item asymmetrically.
As before, to get equilibrium guarantees, we want a payment scheme such that agents use their budget exactly in expectation.
Consider such a scheme and the following simplified setting: two agents with fair shares $\alpha$ and $1-\alpha$ that request the item with probabilities $\alpha$ and $1 - \alpha$, respectively.
Allocating uniformly at random makes the agent with fair share $\alpha$ win $\alpha(\alpha+(1-\alpha)/2) = \alpha(1 + \alpha)/2$ fraction of the rounds.
This results in $1/2$ fraction of her ideal utility when $\alpha$ is small, which is smaller than the robustness guarantee when not setting $\paymentConstant$ to obtain an equilibrium.

To remedy this issue and extend our mechanism, we simulate agents with different fair shares with multiple ``small'' agents.
The basic idea is that if the fair shares $\alpha_i$ are all rational with common denominator $m$\footnote{If the fair shares $\alpha_i$ are irrational, or if the common denominator is large, we can approximate the fair shares with rational fair shares with small denominators and obtain approximate guarantees. \dledit{Specifically, we show in \cref{sec:approximate_fair_shares} that it suffices to approximate the fair shares with rational numbers with denominators at least $1/(2\min_i\alpha_i\epsilon)$ to obtain a $(1-\epsilon)$-fraction of our guarantees on the achieved fraction of ideal utility.}}, we run \allPayMechanismWithBiddingMinimum{} with $m$ agents where agent $i$ gets to control $k_i :=\alpha_i m$ simulated agents in \allPayMechanismWithBiddingMinimum{}.

Implementing this idea naively does not quite work to obtain the same equilibrium behavior, where each simulated agent is bidding independently across rounds with probability $1/m$. There is no particular reason why an agent $i$ controlling multiple simulated agents should have the simulated agents bid independently of each other. Instead, we shall have the principal force some level of independence by only allowing an agent to request whether they want at least one simulated agent bidding or not. If agent $i$ requests to bid at time $t$, then the principal shall sample requests to bid $(\hat r_{(i,1)}[t], \dots, \hat r_{(i,k_i)}[t])$ distributed as i.i.d. $\Bern(1/m)$ conditioned on at least one of them being nonzero to use as the requests in the simulated agents that $i$ controls.

If agent $i$ uses a \asymmetricStrategy{i}, where a $\beta$-aggressive strategy is as before, requesting whenever the value is in the top $\beta$-quantile of the value distribution, then the $\hat r_{(i,i')}[t]$ will be i.i.d. $\Bern(1/m)$. This emulates each simulated agent playing a $1/m$-aggressive strategy.
If each simulated agent wins at least $\lambda_i/m$ rounds, then agent $i$ will win $\lambda_ik_i/m$ rounds, which then implies robustness and equilibrium utility lower bounds. In particular, if we use $\paymentConstant = (m+1)/(1+m(1-1/m)^{n+1})$ as in \cref{sec:equilibrium_mechanism}, a \asymmetricStrategy{i} is $(5/(3e) - o(1))$-robust, and if each agent $i$ plays a \asymmetricStrategy{i}, each agent obtains a $1-\left(1-\frac1m\right)^m - o(1) \geq 1-1/e-o(1)$ fraction of their ideal utility.

% We can show that by the robustness property in \allPayMechanism, that this strategy is robust. As in \cref{sec:equilibrium_mechanism}, by setting $\paymentConstant$ such that agents do not run out of budget in expectation by following this strategy profile, each agent obtains a $1 - (1-1/m)^m \geq 1-1/e$ fraction of their ideal utility.
Using the same arguments as in \cref{sec:equilibrium_mechanism}, if we put a minimum bidding constraint, each agent playing this \asymmetricStrategy{i}{} is an equilibrium. The way we have set the mechanism up, agents can only control the frequency at which they bid. They can't underbid by enforcement, and overbidding does not help because this only decreases the payments of others while winning worse rounds for the overbidder.

Our full mechanism for asymmetric fair shares is as follows. Given the fair shares $\alpha_i = k_i/m$, create a set of simulated agents $\hat N = \{(i,i'):i\in[n],i'\in[k_i]\}$. Initialize an instance $\hat{\mathcal M}$ of \allPayMechanism{} with the agent set $\hat N$ and equal fair shares: $\hat{\alpha}_{(i,i')} = 1/m$ for each $(i,i')\in\hat N$. At each time $t$, each real agent $i$ can request to bid or not. Let $r_i[t]$ denote the indicator that agent $i$ requests to bid. We enforce a bidding minimum that $\sum_{s=1}^t r_i[s] \geq \left(\asymmetricRate{i}\right)t - o(T)$. Let $S[t] = \{i:r_i[t]=1\}$ be the set of bidding agents. For $X_1, \dots, X_k$ i.i.d. $\Bern(1/m)$ random variables, let $\mathcal D_{k,m}$ be the distribution of $(X_1, \dots, X_k)$ conditioned on the event that at least one of the $X_i=1$. For each bidding agent $i\in S[t]$, we sample $(\hat V_{(i,1)}[t], \dots, \hat V_{(i,k_i)}[t])\sim \mathcal D_{k_i,m}$. Set $\hat S[t] = \{(i,i'):i\in S[t],\hat V_{(i,i')}=1\}$, which we set as the set of requesting simulated agents at time $t$, and we simulate $\hat{\mathcal M}$ at time with bidding agents $\hat S[t]$. From $\hat{\mathcal M}$, there is simulated agent $(i,i')$ who won the item, and we give the item to the real agent $i$. (We fully simulate $\hat{\mathcal M}$, so the budgets of the simulated agents also get updated within $\hat{\mathcal M}$, and they are enforced in $\hat{\mathcal M}$ when the simulated agents request to bid.)

We formally describe this mechanism in \cref{alg:principal_bids_for_you_all_pay_with_asymmetric_fair_shares} and prove the following theorem in \cref{sec:app:asymmetric_fair_shares_proofs}.

\begin{algorithm}
	\SetAlgoNoLine
	\KwIn{Fair shares $(\alpha_i)_{i\in[n]}$ where $\alpha_i = k_i/m$, number of rounds $T$, payment constant $\paymentConstant$, underbidding allowance $\epsilon = o(T)$}
    Let $\hat N = \{(i,i'):i\in[n], i'\in [k_i]\}$\;
    Initialize an instance $\hat{\mathcal M}$ of \allPayMechanism{} with the agent set $\hat N$, equal fair shares $\hat {\alpha}_{(i,i')} = 1/m$, and payment constant $\paymentConstant$\;
    \For{$t=1,2,\dots,T$}{
        Agents either request to bid or not (let $r_i[t]$ be the indicator that agent $i$ requests to bid)\;	    
        Enforce bidding minimums: $r_i[t]\gets 1$ for each $i$ such that $\sum_{s=1}^t r_i[s]\leq \left(\asymmetricRate{i}\right)t - \epsilon$\;
        Let $S[t] = \{i:r_i[t]=1\}$ be the set of bidding agents\;
        For each $i\in S[t]$, sample $(\hat V_{(i,1)}[t], \dots, \hat V_{(i,k_i)}[t])\sim \mathcal D_{k_i,m}$ (where $\mathcal D_{k,m}$ is the distribution of $(X_1, \dots, X_k)$ conditioned on at least one  $X_i=1$ where the $X_i$ are i.i.d. $\Bern(1/m)$) \;
        Let $\hat S[t] = \left\{(i,i'):i\in S[t], \hat V_{(i,i')}=1\right\}$\;
        Simulate $\hat{\mathcal M}$ at time $t$ with requesting agents $\hat S[t]$ to determine a simulated winner in $\hat N$\;
        For the simulated agent $(i,i')$ that won the item in $\hat{\mathcal M}$, give the item to agent $i$.
    }
    \caption{Asymmetric Fair Share Mechanism}
    \label{alg:principal_bids_for_you_all_pay_with_asymmetric_fair_shares}
\end{algorithm}

\begin{restatable}{theorem}{asymmetricFairSharesTheorem} \label{thm:asymmetric_fair_shares}
Consider \allPayMechanismWithasymmetricFairShares with $\paymentConstant = (m+1)/(1+m(1-1/m)^{m+1})$ and $\epsilon = \sqrt{T\ln T}$. Then, each agent $i$ playing a \asymmetricStrategy{i} is a $\lrob$-robust $\lnash$-good approximate-equilibrium for some $\lrob$ and $\lnash$ satisfying
\begin{equation*}
    \lrob \geq \frac{5}{3e} - O\left(\sqrt{\frac{\log T}{T}}\right), \lnash \geq 1-\left(1-\frac1m\right)^m - O\left(\sqrt{\frac{\log T}{T}}\right).
\end{equation*}
% Furthermore, each agent playing a \asymmetricStrategy{i} is an $O\left(\sqrt{\frac{\log T}{T}}\right)$-equilibrium in which each agent obtains at least a $\lnash$-fraction of their ideal utility with probability at least $1-O(1/T^2)$, where
% \begin{equation*}
%     \lnash \geq 1-\left(1-\frac1m\right)^m - O\left(\sqrt{\frac{\log T}{T}}\right).
% \end{equation*}
\end{restatable}

\begin{remark}
At a high level, our techniques are similar to that of \cite{lin2025robust}: in the equal fair share case, their mechanism is just \allPayMechanism, but instead of agents paying an amount dependent on the number of bidders, each bidder always pays $1$ each round they request. Our works also differ in how we handle the asymmetric fair share case. Our approach in reducing asymmetric fair shares to symmetric fair shares does not seem to be able to get the optimal $\lnash = 1-\prod_j (1-\alpha_j)$, since if there are $m$ agents with equal fair shares, the best possible is $\lnash = 1-(1-1/m)^m$, and if $1/m \leq \alpha_j$ for each $j$ with the inequality strict for at least one $j$, then $1-(1-1/m)^m < 1-\prod_j(1-\alpha_j)$. In contrast, \citet{lin2025robust} are able to obtain $\lnash = 1-\prod_j(1-\alpha_j)$. They do this by changing the uniformly random allocation rule. Specifically, if at some time a set of agents $S$ bid, they allocate the item to an $i\in S$ according a probability distribution $(p_i^S)_{i\in S}$ that depends on $S$. These probability distributions $(p_i^S)_{i\in S}$ are extremely complicated. Much of their work is dedicated to only proving the existence of $(p_i^S)_{i\in S}$ that guarantee their robustness and equilibrium guarantees, and they do not have any formula for the $p_i^S$. In contrast, our mechanism, \allPayMechanismWithasymmetricFairShares, is much simpler.
\end{remark}

\section*{Acknowledgements}

David X. Lin and Siddhartha Banerjee were supported by AFOSR grant FA9550-231-0068, NSF grant ECCS-1847393 and through a SPROUT Award from Cornell Engineering. 
Giannis Fikioris was supported
in part by the Google PhD Fellowship, the Onassis Foundation – Scholarship ID: F ZS 068-1/2022-2023, and ONR MURI grant
N000142412742. 
Éva Tardos was supported in part by AFOSR grant FA9550-23-1-0410, AFOSR grant FA9550-231-0068, ONR MURI grant N000142412742, and through a SPROUT Award from Cornell Engineering.

% Bibliography
% \bibliographystyle{ACM-Reference-Format}
% \bibliography{bibliography}
\printbibliography

% Appendix
\appendix
\crefalias{section}{appendix}
\section{Deferred Proofs from Section \ref{sec:robust_mechanism}}
\label{sec:app:robust_mechanism_proofs}

Our goal of this section is to provide the full proof of \cref{thm:robustness_theorem}, which we restate below.

\robustnessTheorem*

We first reduce the problem to lower bounding the amount the number of wins the agent has. 
% First, we use a similar construction as \cite{lin2025online}, to reduce the argue that we can assume without loss of generality that agent $i$'s value distribution is $\mathcal F_i = \Bern(\alpha_i$).

If agent $i$ has budget at time $t$, let $\hat V_i[t]$ be $1$ if agent $i$ bids at time $t$ and $0$ otherwise. If agent $i$ does not have budget at time $t$, let $\hat V_i[t]$ be sampled from $\Bern(\alpha_i)$, independently of everything else. Notice that the $\hat V_i[t]$ are i.i.d. $\Bern(\alpha_i)$ random variables across time such that $\hat V_i[t] = r_i[t]$ whenever agent $i$ has budget.

\begin{lemma} \label{lem:bernoulli_reduction}
Fix an agent $i$ and the strategies of other agents $j\neq i$. Assume agent $i$ is following an $\alpha_i$-aggressive strategy. Then,
\begin{equation*}
    \left|\frac1T\sum_{t=1}^T U_i[t] - \frac{v_i^*}{\alpha_i T}\sum_{t=1}^T W_i[t]\right| \leq \left(\sqrt{\frac{\log T}{T}}\right)
\end{equation*}
with probability at least $1-O(1/T^2)$.
\end{lemma}
\begin{proof}

Conditioned on $(W_i[s])_{s\in[T]}$, using the fact that conditioned on $\hat V_i[t]$, the $V_i[t]$ are independent of $(W_i[s])_{s\in[T]}$ and the fact that the values $V_i[t]$ are independent across time, the random variables $U_i[t] = V_i[t]W_i[t]$ are independent random variables with expectation
\begin{align*}
    \E[V_i[t]\mid (W_i[s])_{s\in [T]}]W_i[t] & = \E[V_i[t]\mid \hat V_i[t]]W_i[t]\\
    & = \E[V_i[t]\mid \hat V_i[t]]\hat V_i[t]W_i[t] & \text{[Bids if and only if $\hat V_i[t]=1$]}\\
    & = \E[V_i[t]\mid \hat V_i[t]=1]\hat V_i[t]W_i[t] = \frac{v_i^*}{\alpha_i}W_i[t]. & \text{[Definition of ideal utility]}
\end{align*}
Thus, by the Hoeffding bound, letting $\bar v$ bound the value distribution $\mathcal F_i$,
\begin{equation*} 
% \label{eq:high_probability_bound_conditioned_on_bernoulli_wins}
    \Pr\left(\left|\sum_{t=1}^T U_i[t] - \frac{v_i^*}{\alpha_i}\sum_{t=1}^T W_i[t]\right| \geq \bar v\sqrt{T\ln T}\,\middle|\,(W_i[s])_{s\in [T]} \right) \leq 2\exp\left(-\frac{2\bar v^2T\ln T}{\bar v^2T}\right) = \frac{2}{T^2}.
\end{equation*}
This implies the lemma statement by dividing by $T$ and using the unconditional probability:
\begin{equation*}
    \Pr\left(\left|\frac1T\sum_{t=1}^T U_i[t] - \frac{v_i^*}{\alpha_i T}\sum_{t=1}^T W_i[t]\right| \geq \bar v\sqrt{\frac{\ln T}{T}}\right) \leq \frac{2}{T^2}.
\end{equation*}

\end{proof}

The below lemma uses standard concentration bounds to relate the random quantities to their (conditional) expectations. This allows us to only reason about these expectations.

Let $\mathcal H_t$ denote the history up to and including time $t$. Let $\mathcal G_t$ be the $\sigma$-algebra generated by $\mathcal H_t$ and $r_j[t+1]$ for $j\neq i$.

\begin{lemma}
    \label{lem:high_probability}
    There is an event $E$ of probability at least $1-O(1/T^2)$ such that on $E$, the following hold.
    \begin{align}
        & \frac1T\sum_{t=1}^T \hat V_i[t] \leq \alpha_i + O\left(\sqrt{\frac{\log T}{T}}\right) \label{eq:high_probability_agent_values}\\
        & \frac1T\sum_{j\neq i}\sum_{t=1}^T \E[P_j[t]\mid \mathcal G_{t-1}] \leq 1-\alpha_i + O\left(\sqrt{\frac{\log T}{T}}\right) \label{eq:high_probability_adversary_payment}\\
        % & \frac1T\sum_{t=1}^T W_i[t]\geq \frac1T\sum_{t=1}^T \E[W_i[t]\mid \mathcal F_{t-1}] - O\left(\sqrt{\frac{\log T}{T}}\right) \label{eq:high_probability_allocation_agent_wins}\\
        & \frac1T\sum_{t=1}^T W_i[t]\geq \frac1T\sum_{t=1}^T \E[W_i[t]\mid \mathcal G_{t-1}] - O\left(\sqrt{\frac{\log T}{T}}\right) \label{eq:high_probability_value_and_allocation_agent_wins}\\
        & \frac1T\sum_{t=1}^T P_i[t]\leq \frac1T\sum_{t=1}^T \E[P_i[t]\mid \mathcal G_{t-1}] + O\left(\sqrt{\frac{\log T}{T}}\right) \label{eq:high_probability_agent_payment}
    \end{align}
    \end{lemma}
    \begin{proof}
    To prove \eqref{eq:high_probability_agent_values}, note that the $\hat V_i[t]$ are i.i.d. random variables bounded by $1$ with mean $\alpha_i$. By Hoeffding's inequality,
    \begin{equation*}
        \Pr\left(\sum_{i=1}^n \hat V_i[t] \leq \alpha_iT + \sqrt{T\ln T}\right) \leq \exp\left(-\frac{2T\ln T}{T}\right) \leq \frac{1}{T^2}.
    \end{equation*}
    \cref{eq:high_probability_adversary_payment,eq:high_probability_value_and_allocation_agent_wins,eq:high_probability_agent_payment} can be proven similarly by applying Azuma-Hoeffding to the $\mathcal G_t$-martingale $\sum_{j\neq i}\sum_{s=1}^t (P_j[s] - \E[P_j[s]\mid\mathcal G_{t-1}])$ along with the budget constraint $\sum_{j\neq i}\sum_{s=1}^T P_j[s] \leq \sum_{j\neq i}\alpha_jT = (1-\alpha_i)T$ (with increments bounded by $\paymentConstant$),
    Azuma-Hoeffding to the $\mathcal G_t$-martingale $\sum_{s=1}^t (W_i[s] - \E[W_i[s]\mid\mathcal G_{s-1}])$,
    and Azuma-Hoeffding to the $\mathcal G_t$-martingale $\sum_{s=1}^t (P_i[s] - \E[P_i[s]\mid\mathcal G_{s-1}])$.
\end{proof}

The below lemma helps us reason about the payments of agents.
\begin{lemma} \label{lem:expected_payment_of_agent}
The expected payment of an agent $i$ conditioned on the bids $r_k[t]$ for $k\in[n]$ is
\begin{equation*}
    \E[P_i[t]\mid (r_k[t])_{k\in[n]}] = \paymentConstant\cdot\frac{r_i[t]}{1 + \sum_{k\in[n]}r_k[t]}
\end{equation*}
\end{lemma}
\begin{proof}
Conditioned on the bids $r_k[t]$, agent $i$ wins with probability $\frac{r_i[t]}{\sum_kr_k[t]}$ by the uniform allocation rule in which case they pay $\paymentConstant\cdot\frac{1}{1+\sum_kr_k[t]}$ by the payment rule. The expected payment $\E[P_i[t]\mid (r_k[t])_{k\in[n]}]$ is just the product of these.
\end{proof}

% By the strategy of agent $i$, $U_i[t] = W_i[t]$. The ideal utility of agent $i$ with a $\Bern(\alpha_i)$ is exactly $\alpha_i$, so to show our proposed strategy is $\lambda_i$-robust, we must show that $\frac1{\alpha_i T}\sum_{t=1}^T \E[W_i[t]]\geq \lambda_i$.

Let $\tau_i$ be the time at which agent $i$ runs out of budget and $T$ if the agent never runs out of budget:
\begin{equation*}
    \tau_i = \begin{cases}\min\left\{t:\sum_{s=1}^tP_i[s] \geq \alpha_iT\right\} & \text{if $\sum_{s=1}^T P_i[s] \geq \alpha_iT$}\\ T & \text{otherwise}\end{cases}.
\end{equation*}
Notice that $\tau_i$ is a stopping time with respect to the filtration $\mathcal G_t$.
Let $k_i[t]$ be the number of agents $j\neq i$ that bid at time $t$ and have budget remaining, i.e., $k_i[t] = \#\{j\neq i:r_j[t]=1\}$.
Let $X_k$ be the fraction of times $t\leq\tau$ where $k_i[t]=k$, i.e., $X_k = \frac{1}{\tau_i}\cdot\#\{t\leq\tau_i: k_i[t]=k\}$. The below lemma expresses the agent's utility and adversary's budget constraint in terms of $\tau_i$ and the $X_k$.
\begin{lemma} \label{lem:agent_utility_adversary_budget_constraint_rewritten}
On the event $E$ as in \cref{lem:high_probability}, the following hold.
\begin{equation}
\label{eq:agent_adversary_budget_matching}
    \frac{\paymentConstant\tau_i}{T}\sum_{k=1}^{n-1} \left((1-\alpha_i)\cdot\frac{k}{1+k} + \alpha_i\cdot\frac{k}{2+k}\right)X_k \leq 1-\alpha_i + O\left(\sqrt{\frac{\log T}{T}}\right)
\end{equation}
\begin{equation} \label{eq:agent_wins_lower_bound}
    \frac{1}{\alpha_i T}\sum_{t=1}^T W_i[t] \geq \frac{\tau_i}{T}\sum_{k=0}^{n-1} \frac{X_k}{1+k} - O\left(\sqrt{\frac{\log T}{T}}\right)
\end{equation}
Also, if $\sum_{t=1}^T P_i[t]\geq \alpha_iT$, then on $E$,
\begin{equation}
\label{eq:agent_stopping_time_lower_bound}
    \tau_i \geq \frac{T}{\paymentConstant\sum_{k=0}^{n-1} \frac{X_k}{2+k}} - O\left(\sqrt{\frac{\log T}{T}}\right).
\end{equation}

\end{lemma}
\begin{proof}
For $t\leq\tau_i$, by \cref{lem:expected_payment_of_agent},
\begin{equation*}
    \E[P_i[t]\mid\mathcal G_{t-1}] = \E\left[\paymentConstant\cdot\frac{r_i[t]}{2+k_i[t]}\,\middle|\,\mathcal G_{t-1}\right] = \frac{\alpha_i\paymentConstant}{2+k_i[t]}.
\end{equation*}
since agent $i$ bids with probability $\alpha_i$ independent of $\mathcal G_{t-1}$ conditioned on $t\leq\tau_i$, and if agent $i$ bids, there are $1+k_i[t]$ agents total bidding.

Applying \eqref{eq:high_probability_agent_payment}, noting that $\sum_{t=1}^T P_i[t] = \sum_{t=1}^{\tau_i} P_i[t]$ by the budget constraint enforcement, on $E$,
\begin{equation*}
\begin{split}
    \frac1T\sum_{t=1}^{\tau_i} P_i[t] & \leq \frac{\alpha_i\paymentConstant}{T}\sum_{t=1}^{\tau_i}\frac{1}{2+k_i[t]} + O\left(\sqrt{\frac{\log T}{T}}\right)\\
    & = \frac{\alpha_i\paymentConstant\tau_i}{T}\sum_{k=0}^{n-1} \frac{X_k}{2+k} + O\left(\sqrt{\frac{\log T}{T}}\right).
\end{split}
\end{equation*}
If $\sum_{t=1}^{\tau_i} P_i[t] \geq \alpha_iT$, the above is at least $\alpha_i$, and we can rearrange to obtain \eqref{eq:agent_stopping_time_lower_bound}.

When $t\leq\tau_i$, for bidding agents $j\neq i$, by \cref{lem:expected_payment_of_agent},
\begin{equation*}
    \E\left[P_j[t]\mid\mathcal G_{t-1}\right] = \paymentConstant\left((1-\alpha_i)\cdot\frac{1}{1+k_i[t]} + \alpha_i\cdot\frac{1}{2+k_i[t]}\right)r_j[t]
\end{equation*}
since agent $i$ does not bid with probability $1-\alpha_i$, in which case there are $k_i[t]$ agents bidding, and the agent bids with probability $\alpha_i$, in which case there are $1+k_i[t]$ agents bidding. 
Then,
\begin{equation*}
\begin{split}
    \frac1T\sum_{t=1}^{\tau_i} \sum_{j\neq i}\E[P_j[t]\mid\mathcal G_{t-1}] & = \frac1T\sum_{t=1}^{\tau_i}\sum_{j\neq i}\paymentConstant\left((1-\alpha_i)\cdot\frac{1}{1+k_i[t]} + \alpha_i\cdot\frac{1}{2+k_i[t]}\right)r_j[t]\\
     & = \frac1T\sum_{t=1}^{\tau_i} \paymentConstant\left((1-\alpha_i)\cdot\frac{k_i[t]}{1+k_i[t]} + \alpha_i\cdot\frac{k_i[t]}{2+k_i[t]}\right)\\
    & = \frac{\paymentConstant \tau_i}{T}\sum_{k=1}^{n-1} \left((1-\alpha_i)\cdot\frac{k}{1+k} + \alpha_i\cdot\frac{k}{2+k}\right)X_k.
\end{split}
\end{equation*}
By \eqref{eq:high_probability_adversary_payment}, on $E$, this is at most $1-\alpha_i + O\left(\sqrt{\frac{\log T}{T}}\right)$, so we obtain \eqref{eq:agent_adversary_budget_matching}.

Agent $i$'s utility satisfies
\begin{equation*}
    \frac{1}{\alpha_i T}\sum_{t=1}^{\tau_i} \E[W_i[t]\mid\mathcal G_{t-1}] = \frac{1}{\alpha_iT}\sum_{t=1}^{\tau_i} \frac{\alpha_i}{1+k_i[t]} = \frac{\tau_i}{T}\sum_{k=0}^{n-1} \frac{X_k}{1+k}
\end{equation*}
since at each time $t\leq\tau_i$, agent $i$ bids with probability $\alpha_i$, and there are $1+k_i[t]$ total bidding agents.
Using \eqref{eq:high_probability_value_and_allocation_agent_wins}, we obtain \eqref{eq:agent_wins_lower_bound}.

\end{proof}

We now break into two cases and and lower bound agent $i$'s wins in each case. Either agent $i$ spends all their budget, or they don't. We first handle the case where they do not run out of budget.

\begin{lemma} \label{lem:robustness_when_not_running_out_of_budget}
Assume $\paymentConstant \geq 2$. On the event $E$ in \cref{lem:high_probability}, if $\sum_{t=1}^T P_i[t] < \alpha_iT$, then
\begin{equation}
\label{eq:expected_wins_agent_does_not_run_out_of_budget}
\frac1{\alpha_i T}\sum_{t=1}^T W_i[t]\geq 1 - \frac{3(1-\alpha_i)}{3\paymentConstant - \alpha_i\paymentConstant} - O\left(\sqrt{\frac{\log T}{T}}\right).
\end{equation}
\end{lemma}
\begin{proof}
The assumption that $\sum_{t=1}^T P_i[t] < \alpha_iT$ implies that $\tau_i = T$ by the definition of $\tau_i$. By \cref{lem:agent_utility_adversary_budget_constraint_rewritten}, specifically \eqref{eq:agent_adversary_budget_matching}, on $E$, we have
\begin{equation}
    \paymentConstant\sum_{k=1}^{n-1} \left((1-\alpha_i)\cdot\frac{k}{1+k} + \alpha_i\cdot\frac{k}{2+k}\right)X_k \leq 1-\alpha_i + O\left(\sqrt{\frac{\log T}{T}}\right)
\end{equation}
and
\begin{equation} 
    \frac{1}{\alpha_i T}\sum_{t=1}^T W_i[t] \geq \sum_{k=0}^{n-1} \frac{X_k}{1+k} - O\left(\sqrt{\frac{\log T}{T}}\right).
\end{equation}
Also, we clearly have $X_k\geq 0$ for each $k$ and $\sum_{k=0}^{n-1} X_k = 1$. Thus, $\frac{1}{\alpha_i T}\sum_{t=1}^T W_i[t]$ is lower bounded by the value of the following minimization problem, where $f(T)$ and $g(T)$ are functions in $ O\left(\sqrt{\frac{\log T}{T}}\right)$.
\begin{align}
    \min_{(x_k)_{k=0}^n} \quad & \sum_{k=0}^{n-1} \frac{x_k}{1+k} - f(T)\nonumber \\
    \text{s.t.} \quad & \paymentConstant\sum_{k=1}^{n-1}\left((1-\alpha_i)\cdot\frac{k}{1+k} + \alpha_i\cdot\frac{k}{2+k}\right)x_k \leq 1-\alpha_i + g(T)\label{eq:adversary_budget_constraint_agent_does_not_run_out_of_money_in_optimization_problem} \\
    & \sum_{k=0}^{n-1} x_k = 1 \nonumber \\
    & x_k \geq 0 & \forall k \nonumber
\end{align}
This is a linear program and so its minimum is achieved at an extreme point $(x_k^*)$ of its feasible polytope. This implies that there are only two nonzero coordinates $x_k^*$. At least one of these two nonzero coordinates must be $k=0$ because if $k>0$, then
\begin{equation*}
    \paymentConstant\left((1-\alpha_i)\cdot\frac{k}{1+k} + \alpha_i\cdot\frac{k}{2+k}\right) > \frac{\paymentConstant}{2}(1-\alpha_i) \geq 1-\alpha_i.
\end{equation*}
We use the assumption that $\paymentConstant \geq 2$ for the last inequality. Thus, if $\sum_{k=1}^{n-1} x_k = 1$, for $T$ sufficiently large, \eqref{eq:adversary_budget_constraint_agent_does_not_run_out_of_money_in_optimization_problem} cannot hold. It follows that any feasible $(x_k)$ has $x_0>0$.

Let $k^* > 0$ be the nonzero coordinate of $(x_k^*)$. The constraint \eqref{eq:adversary_budget_constraint_agent_does_not_run_out_of_money_in_optimization_problem} says
\begin{equation*}
    \paymentConstant\left((1-\alpha_i)\cdot \frac{k^*}{1+k^*} + \alpha_i\cdot \frac{k^*}{2+k^*}\right)x_{k^*}^* \leq 1-\alpha_i + g(T).
\end{equation*}
Solve for $x_{k^*}^*$ to obtain 
\begin{equation*}
    x_{k^*}^* \leq \frac{(1+k^*)(2+k^*)(1-\alpha_i)}{\paymentConstant k^*(2+k^*-\alpha_i)} + O\left(\sqrt{\frac{\log T}{T}}\right).
\end{equation*}
Then, the objective value that $(x_{k^*}^*)$ achieves is
\begin{equation*}
\begin{split}
    x_0^* + \frac{x_{k^*}^*}{1+k^*} - f(T) & = (1-x_{k^*}^*) + \frac{x_{k^*}^*}{1+k^*} - O\left(\sqrt{\frac{\log T}{T}}\right)\\
    & = 1 - \frac{(2+k^*)(1-\alpha_i)}{\paymentConstant(2+k^*-\alpha_i)} - O\left(\sqrt{\frac{\log T}{T}}\right).
\end{split}
\end{equation*}
The above is increasing in $k^*$, so it is minimized when $k^*=1$, in which case it is
\begin{equation*}
    1 - \frac{3(1-\alpha_i)}{3\paymentConstant - \alpha_i\paymentConstant} - O\left(\sqrt{\frac{\log T}{T}}\right),
\end{equation*}
as desired.

\end{proof}

Now we handle the case where the agent runs out of budget.
\begin{lemma} \label{lem:robustness_when_running_out_of_budget}
On the event $E$ in \cref{lem:high_probability}, if $\sum_{t=1}^T P_i[t] \geq \alpha_i T$, then
\begin{equation*}
    \frac1{\alpha_i T}\sum_{t=1}^T W_i[t] \geq \frac{5-\alpha_i}{\paymentConstant(3-\alpha_i)} - O\left(\sqrt{\frac{\log T}{T}}\right).
\end{equation*}
\end{lemma}
\begin{proof}
Substituting the stopping time bound \eqref{eq:agent_stopping_time_lower_bound} in \cref{lem:agent_utility_adversary_budget_constraint_rewritten} for when $\sum_{t=1}^T P_i[t] \geq \alpha_i T$ into \eqref{eq:agent_adversary_budget_matching} and \eqref{eq:agent_wins_lower_bound}, we obtain the following.
\begin{equation*}
    \frac{\sum_{k=1}^{n-1} \left((1-\alpha_i)\cdot\frac{k}{1+k} + \alpha_i\cdot\frac{k}{2+k}\right)X_k}{\sum_{k=0}^{n-1} \frac{X_k}{2+k}} \leq 1-\alpha_i + O\left(\sqrt{\frac{\log T}{T}}\right)
\end{equation*}
\begin{equation*}
    \frac{1}{\alpha_i T}\sum_{t=1}^T W_i[t] \geq \frac{\sum_{k=0}^{n-1} \frac{X_k}{1+k}}{\paymentConstant\sum_{k=0}^{n-1}\frac{X_k}{2+k}} - O\left(\sqrt{\frac{\log T}{T}}\right)
\end{equation*}

In addition, the $(X_k)_{k=0}^n$ satisfy $\sum_{k=0}^{n-1} X_k=1$ and $X_k\geq 0$ for all $k$. Therefore, it suffices to lower bound the value of the following minimization problem, where $f(T)$ and $g(T)$ are functions in $O\left(\sqrt{\frac{\log T}{T}}\right)$
\begin{align}
    \min_{(x_k)_{k=0}^n} \quad & \frac{\sum_{k=0}^{n-1} \frac{x_k}{1+k}}{\paymentConstant\sum_{k=0}^{n-1} \frac{x_k}{2+k}} - f(T). \nonumber \\
    \text{s.t.} \quad & (1-\alpha_i)\sum_{k=0}^{n-1} \frac{x_k}{2+k} \geq \sum_{k=1}^{n-1}\left((1-\alpha_i)\cdot\frac{k}{1+k} + \alpha_i\cdot\frac{k}{2+k}\right)x_k - g(T)\label{eq:agent_adversary_budget_matching_optimization_problem}\\
    & \sum_{k=0}^{n-1} x_k = 1 \nonumber \\
    & x_k \geq 0 & \forall k \nonumber
\end{align}
The objective function is a linear-fractional function where the denominator is always positive, so it is quasi-concave and achieves its minimum at a point $(x_k^*)$ at an extreme point of the feasible polytope. This must occur when there are only two coordinates $x_k^*$ that are nonzero. We now show that one of the nonzero coordinates must be $k=0$ by proving that $x_k>0$ for any feasible $(x_k)$. Since if $(x_k)$ has $x_0=0$, the left-hand side of \eqref{eq:agent_adversary_budget_matching_optimization_problem} is just $(1-\alpha_i)\sum_{k=1}^{n-1} \frac{x_k}{2+k}$ and the right-hand side is greater than $(1-\alpha_i)\sum_{k=1}^{n-1} \frac{k}{1+k}\cdot x_k - O\left(\sqrt{\frac{\log T}{T}}\right)$, which is greater for sufficiently large $T$, a contradiction.

Let $k^* > 0$ be the nonzero coordinate of $(x_k^*)$. The constraint \eqref{eq:agent_adversary_budget_matching_optimization_problem} says
\begin{equation*}
   (1-\alpha_i)\left(\frac{x_0^*}{2} + \frac{x_{k^*}^*}{2+k^*}\right) \geq  \left((1-\alpha_i)\cdot\frac{k^*}{1+k^*} + \alpha_i\cdot\frac{k^*}{2+k^*}\right)x_{k^*}^* - g(T).
\end{equation*}
Substitute $x_0^* = 1-x_{k^*}^*$ into the above and solve for $x_k^*$ to obtain
\begin{equation*}
    x_{k^*}^* \leq \frac{(1+k^*)(2+k^*)(1-\alpha_i)}{k^*(5 + 3k^* - 3\alpha_i - k^*\alpha_i)} + O\left(\sqrt{\frac{\log T}{T}}\right).
\end{equation*}
Then, using the above bound, the objective value that $(x_k^*)$ achieves is 
\begin{equation*}
\begin{split}
    \frac{\frac{x_0^*}{2} + \frac{x_{k^*}^*}{1+k^*}}{\paymentConstant\left(\frac{x_0^*}{2} + \frac{x_{k^*}^*}{2+k^*}\right)} - f(T) & = \frac{\frac{1-x_{k^*}^*}{2} + \frac{x_{k^*}^*}{1+k^*}}{\paymentConstant\left(\frac{1-x_{k^*}^*}{2} + \frac{x_{k^*}^*}{2+k^*}\right)} - f(T)\\
    & \geq \frac{3+2k^*-\alpha_i}{\paymentConstant(2+k^*-\alpha_i)} - O\left(\sqrt{\frac{\log T}{T}}\right).
\end{split}
\end{equation*}
The above is increasing in $k^*$, so it is minimized when $k^*=1$, in which case it is
\begin{equation*}
    \frac{5-\alpha_i}{\paymentConstant(3-\alpha_i)} - O\left(\sqrt{\frac{\log T}{T}}\right),
\end{equation*}
as desired.

\end{proof}

Now we prove our main robustness result from \cref{lem:robustness_when_not_running_out_of_budget,lem:robustness_when_running_out_of_budget}. We state the robustness in terms of $\paymentConstant$. The below theorem implies \cref{thm:robustness_theorem}; just note that substituting $\paymentConstant = 8/3$ guarantees
\begin{equation*}
    \lambda_i \geq \frac{3(5-\alpha_i)}{8(3-\alpha_i)} - O\left(\sqrt{\frac{\log T}{T}}\right) \geq \frac58 - O\left(\sqrt{\frac{\log T}{T}}\right).
\end{equation*}

\begin{theorem}
\label{thm:general_robustness_claim}
When running \allPayMechanism with $\paymentConstant\geq 2$, if agent $i$ uses an $\alpha_i$-aggressive strategy, regardless of the strategies of other agents,
\begin{equation*}
    \frac1T\sum_{t=1}^T U_i[t] \geq \lambda_i v_i^*
\end{equation*}
with probability at least $1-O(1/T^2)$ where 
\begin{equation*}
    \lambda_i \geq \min\left\{1 - \frac{3(1-\alpha_i)}{3\paymentConstant - \alpha_i\paymentConstant}, \frac{5-\alpha_i}{\paymentConstant(3-\alpha_i)}\right\} - O\left(\sqrt{\frac{\log T}{T}}\right).
\end{equation*}
In particular, an $\alpha_i$-aggressive strategy is $\lambda_i$-robust.
\end{theorem}
\begin{proof}
By \cref{lem:robustness_when_not_running_out_of_budget,lem:robustness_when_running_out_of_budget}, on all of $E$,
\begin{equation*}
    \frac{1}{\alpha_iT}\sum_{t=1}^T W_i[t] = \frac{1}{\alpha_iT}\sum_{t=1}^T W_i[t] \geq \min\left\{1 - \frac{3(1-\alpha_i)}{3\paymentConstant - \alpha_i\paymentConstant}, \frac{5-\alpha_i}{\paymentConstant(3-\alpha_i)}\right\} - O\left(\sqrt{\frac{\log T}{T}}\right).
\end{equation*}
where the first equality follows from the fact that the agent only bids when she has value $1$.

It follows from \cref{lem:bernoulli_reduction} that %with an arbitrary value distribution $\mathcal F_i$,
\begin{equation*}
    \frac1T\sum_{t=1}^T U_i[t] \geq \min\left\{1 - \frac{3(1-\alpha_i)}{3\paymentConstant - \alpha_i\paymentConstant}, \frac{5-\alpha_i}{\paymentConstant(3-\alpha_i)}\right\}v_i^* - O\left(\sqrt{\frac{\log T}{T}}\right).
\end{equation*}
with probability at least $\Pr(E) - O(1/T^2) \geq 1- O(1/T^2)$.

For the robustness claim, we just need to convert the high probability claim into one in expectation:
\begin{equation*}
\begin{split}
    \frac1T\sum_{t=1}^T\E[U_i[t]] \geq \frac1T\sum_{t=1}^T \E[U_i[t]\pmb1_E] \geq \left(\min\left\{1 - \frac{3(1-\alpha_i)}{3\paymentConstant - \alpha_i\paymentConstant}, \frac{5-\alpha_i}{\paymentConstant(3-\alpha_i)}\right\}v_i^* - O\left(\sqrt{\frac{\log T}{T}}\right)\right)\Pr(E)\\
    \geq \left(\min\left\{1 - \frac{3(1-\alpha_i)}{3\paymentConstant - \alpha_i\paymentConstant}, \frac{5-\alpha_i}{\paymentConstant(3-\alpha_i)}\right\}v_i^* - O\left(\sqrt{\frac{\log T}{T}}\right)\right)\left(1-\frac{1}{T^2}\right)\\
    \geq \min\left\{1 - \frac{3(1-\alpha_i)}{3\paymentConstant - \alpha_i\paymentConstant}, \frac{5-\alpha_i}{\paymentConstant(3-\alpha_i)}\right\}v_i^* - O\left(\sqrt{\frac{\log T}{T}}\right).
\end{split}
\end{equation*}
\end{proof}
\section{Expectations of Functions of Binomial Random Variables}
\begin{proposition}
\label{prop:expectation_of_functions_of_binom}
Let $X\sim\Binom(n-1,1/n)$ and $Y\sim\Binom(n-2,1/n)$. Then,
\begin{align*}
    \E\left[\frac{1}{1+X}\right] & = 1-\left(1-\frac1n\right)^n\\
    \E\left[\frac{1}{2+X}\right] & = \frac{1 + n(1-1/n)^{n+1}}{n+1}\\
    \E\left[\frac{1}{2+Y}\right] & = \left(1-\frac1n\right)^{n-1}\\
    \E\left[\frac{1}{3+Y}\right] & = \frac{n}{n+1}\left(1 - 2\left(1-\frac1n\right)^n\right)
\end{align*}
\end{proposition}
\begin{proof}
Write
\begin{equation*}
    \E\left[\frac{1}{1+X}\right] = \sum_{k=0}^{n-1}\frac{1}{1+k}\binom{n-1}{k}\left(\frac1n\right)^k\left(1-\frac1n\right)^{n-1-k}.
\end{equation*}
Now we write $\frac{1}{1+k} = \int_0^1 t^k\,dt$, swap the integral and summation, apply the binomial theorem, and integrate:
\begin{equation*}
\begin{split}
    \sum_{k=0}^{n-1}\frac{1}{1+k}\binom{n-1}{k}\left(\frac1n\right)^k\left(1-\frac1n\right)^{n-1-k} & = \sum_{k=0}^{n-1}\left(\int_0^1 t^k\,dt\right)\binom{n-1}{k}\left(\frac1n\right)^k\left(1-\frac1n\right)^{n-1-k}\\
    & = \int_0^1 \sum_{k=0}^{n-1}\binom{n-1}{k} \left(\frac tn\right)^k\left(1-\frac1n\right)^{n-1-k}\,dt\\
    & = \int_0^1 \left(1 + \frac{t-1}n\right)^{n-1}\,dt\\
    & = 1-\left(1-\frac1n\right)^n.
\end{split}
\end{equation*}

The others are computed similarly.
\begin{equation*}
\begin{split}
    \E\left[\frac{1}{2+X}\right] & = \sum_{k=0}^{n-1}\frac{1}{2+k}\binom{n-1}{k}\left(\frac1n\right)^k\left(1-\frac1n\right)^{n-1-k}\\
    & = \sum_{k=0}^{n-1}\left(\int_0^1 t^{1+k}\,dt\right)\binom{n-1}{k}\left(\frac1n\right)^k\left(1-\frac1n\right)^{n-1-k}\\
    & = \int_0^1 t\sum_{k=0}^{n-1}\binom{n-1}{k} \left(\frac tn\right)^k\left(1-\frac1n\right)^{n-1-k}\,dt\\
    & = \int_0^1 t\left(1 + \frac{t-1}n\right)^{n-1}\,dt\\
    & = \frac{1 + n(1-1/n)^{n+1}}{n+1}
\end{split}
\end{equation*}

\begin{equation*}
\begin{split}
    \E\left[\frac{1}{2+Y}\right] & = \sum_{k=0}^{n-2}\frac{1}{2+k}\binom{n-2}{k}\left(\frac1n\right)^k\left(1-\frac1n\right)^{n-2-k}\\
    & = \sum_{k=0}^{n-2}\left(\int_0^1 t^{1+k}\,dt\right)\binom{n-2}{k}\left(\frac1n\right)^k\left(1-\frac1n\right)^{n-2-k}\\
    & = \int_0^1 t\sum_{k=0}^{n-2}\binom{n-2}{k} \left(\frac tn\right)^k\left(1-\frac1n\right)^{n-2-k}\,dt\\
    & = \int_0^1 t\left(1 + \frac{t-1}n\right)^{n-2}\,dt\\
    & = \left(1-\frac1n\right)^{n-1}
\end{split}
\end{equation*}

\begin{equation*}
\begin{split}
    \E\left[\frac{1}{3+Y}\right] & = \sum_{k=0}^{n-2}\frac{1}{3+k}\binom{n-2}{k}\left(\frac1n\right)^k\left(1-\frac1n\right)^{n-2-k}\\
    & = \sum_{k=0}^{n-2}\left(\int_0^1 t^{2+k}\,dt\right)\binom{n-2}{k}\left(\frac1n\right)^k\left(1-\frac1n\right)^{n-2-k}\\
    & = \int_0^1 t^2\sum_{k=0}^{n-2}\binom{n-2}{k} \left(\frac tn\right)^k\left(1-\frac1n\right)^{n-2-k}\,dt\\
    & = \int_0^1 t^2\left(1 + \frac{t-1}n\right)^{n-2}\,dt\\
    & = \frac{n}{n+1}\left(1 - 2\left(1-\frac1n\right)^n\right)
\end{split}
\end{equation*}
\end{proof}

\begin{proposition} \label{prop:ratio_of_binomials}
Let $k\leq m$, and let $X\sim\Binom(k,1/m)$ and $Y\sim\Binom(m-k,1/m)$ be independent. Then,
\begin{equation*}
    \E\left[\frac{X}{X+Y}\pmb1_{X>0}\right] = \frac km\left(1-\left(1-\frac1m\right)^m\right)
\end{equation*}
and
\begin{equation*}
    \E\left[\frac{X}{1+X+Y}\right] = \frac{k}{m}\cdot \frac{1 + m(1-1/m)^{m+1}}{m+1}.
\end{equation*}
\end{proposition}
\begin{proof}
Let $Z_1, Z_2, \dots, Z_m$ be i.i.d. $\Bern(1/m)$ random variables. Notice that $(X,Y)\overset{d} = (Z_1 + \dots + Z_k, Z_{k+1} + \dots + Z_m)$. Then, 
\begin{equation*}
\begin{split}
    \E\left[\frac{X}{X+Y}\pmb1_{X>0}\right] = \E\left[\frac{Z_1 + \dots + Z_k}{Z_1 + \dots + Z_m}\pmb1_{Z_1 + \dots + Z_k>0}\right] = k\E\left[\frac{Z_1}{Z_1 + \dots + Z_m}\pmb1_{Z_1=1}\right]\\
    = k\E\left[\frac{1}{1 + Z_2 + Z_3 + \dots + Z_m}\pmb1_{Z_1=1}\right] = \frac km\E\left[\frac{1}{1 + Z_2 + \dots + Z_m}\right],
\end{split}
\end{equation*}
using symmetry for the second equality. Similarly,
\begin{equation*}
\begin{split}
    \E\left[\frac{X}{1+X+Y}\right] = \E\left[\frac{Z_1 + \dots + Z_k}{1 + Z_1 + \dots + Z_m}\right] = k\E\left[\frac{Z_1}{1 + Z_1 + \dots + Z_m}\right] = k\E\left[\frac{1}{2 + Z_2 + Z_3 + \dots + Z_m}\pmb1_{Z_1=1}\right]\\
    = \frac km\E\left[\frac{1}{2 + Z_2 + \dots + Z_m}\right].
\end{split}
\end{equation*}
Notice that $Z_2 + \dots + Z_m\sim\Binom(m-2,1/m)$, so by \cref{prop:expectation_of_functions_of_binom},
\begin{equation*}
    \E\left[\frac{1}{1 + Z_2 + \dots + Z_m}\right] = 1 - \left(1-\frac1m\right)^m.
\end{equation*}
and
\begin{equation*}
    \E\left[\frac{1}{2 + Z_2 + \dots + Z_m}\right] = \frac{1 + m(1-1/m)^{m+1}}{m+1}.
\end{equation*}
The result follows.

\end{proof}
\section{Deferred Proofs from Section \ref{sec:robust_upper_bound}} \label{sec:app:robust_upper_bound_proofs}

In this section, we prove \cref{thm:robustness_upper_bound}, which we restate below for convenience.

\robustnessUpperBound*

We start with a generalization of \cref{lem:expected_payment_of_agent} for \allPayMechanismGeneralCost.
\begin{lemma} \label{lem:expected_payment_of_agent_general_cost}
The expected payment of an agent $i$ conditioned on the bids $r_k[t]$ for $k\in[n]$ is
\begin{equation*}
    \E[P_i[t]\mid (r_k[t])_{k\in[n]}] = \frac{p_{1 + \sum_{j\neq i}r_j[t]}}{1 + \sum_{j\neq i}r_j[t]}\cdot r_i[t].
\end{equation*}
\end{lemma}
\begin{proof}
Conditioned on the bids $r_k[t]$, agent $i$ only pays a nonzero amount when $r_i[t]=1$. In this case, agent $i$ wins with probability $\frac{1}{\sum_kr_k[t]} = \frac{1}{1 + \sum_{j\neq i}r_j[t]}$ by the uniform allocation rule in which case they pay $p_{\sum_k r_k[t]} = p_{1 + \sum_{j\neq i}r_j[t]}$ by the payment rule. The expected payment $\E[P_i[t]\mid (r_k[t])_{k\in[n]}]$ is just the product of these and $r_i[t]$.
\end{proof}

Now let us establish some upper bounds on the $p_i$ that are required to have robust strategies. Intuitively, the $p_i$ cannot be too high, because if everyone uses the robust strategy and if the $p_i$ are too high, the agents will run out of budget quickly, contradicting robustness.

\begin{lemma} \label{lem:crude_upper_bounds}
Assume each fair share $\alpha_i=1/n$ with $n\geq 3$. If a \oneOverNAggressiveStrategy{} is $\lambda$-robust for $\lambda > 1/2$ not depending on $T$ for all $T$ sufficiently large, then
\begin{equation*}
    0 < p_1 \leq 2e, p_2 \leq 4e, p_3\leq 12e.
\end{equation*}
\end{lemma}
\begin{proof}
Let us first prove that $p_1 > 0$. By contradiction, assume $p_1 = 0$. Fix an agent $i$ that uses a \oneOverNAggressiveStrategy. Suppose the other agents $j\neq i$ use the strategy of having only one agent $j\neq i$ bid each round until they are all out of budget. Let $\tau_k$ be time at which agent $k$ runs out of budget and $T$ if this never happens:
\begin{equation*}
    \tau_k = \begin{cases}\min\left\{t:\sum_{s=1}^tP_k[s] \geq T/n\right\} & \text{if $\sum_{s=1}^T P_k[s] \geq T/n$}\\ T & \text{otherwise}\end{cases}.
\end{equation*}
Agent $i$ runs out of budget at time $\tau_i$ and the agents $j\neq i$ run out of budget at time $\tau_{-i} := \max_{j\neq i}\tau_j$. Notice that since $p_1=0$, the $P_i[t] = P_j[t] = 0$ for every $t$ in which $r_i[t]=0$ and $j\neq i$, and when $r_i[t]=1$ and $t\leq \min\{\tau_i,\tau_{-i}\}$, the payments $P_i[t]$ and $\sum_{j\neq i}P_j[t]$ are equally distributed and independent conditioned on the history and $r_i[t]$ (each of $P_i[t]$ and $P_j[t]$ is $p_2$ with probability $1/2$ and $0$ otherwise). It follows that 
\begin{equation*}
    \sum_{t=1}^{\min\{t,\tau_i,\tau_{-i}\}}\left(P_i[s] - \sum_{j\neq i}P_j[s]\right)
\end{equation*}
is an $\mathcal H_t$-martingale, where $\mathcal H_t$ denotes the history up to and including time $t$. By the Azuma-Hoeffding inequality, noting that the above martingale has increments bounded by $\max\{p_1,p_2\} = p_2$,
\begin{equation} \label{eq:spending_of_naive_adversary}
    \sum_{s=1}^{\min\{\tau_i,\tau_{-i}\}}P_i[t] \geq \sum_{s=1}^{\min\{\tau_i,\tau_{-i}\}}\sum_{j\neq i}P_j[t] - p_2\sqrt{T\ln T}
\end{equation}
with probability at least $1-1/T^2$. Since the total budget of agent $i$ is $T/n$ and the total budget of agents $j\neq i$ is $(n-1)T/n \geq T/n + \Omega(T)$, for sufficiently large $T$, \eqref{eq:spending_of_naive_adversary} implies that $\tau_i \leq \tau_i'$: if $\tau_i' > \tau_i$, then $\sum_{s=1}^{\min\{\tau_i,\tau_{-i}\}}P_i[t] < T/n$ but $\sum_{s=1}^{\min\{\tau_i,\tau_{-i}\}}\sum_{j\neq i}P_j[t] \geq (n-1)T/n > T/n + \Omega(T)$.

For each time $t\leq \min\{\tau_i,\tau_{-i}\}$, $\E[W_i[t]\mid r_i[t]=1] = 1/2$ since there are $2$ total agents bidding at time $t$ if agent $i$ bids. For each time $t\leq \min\{\tau_i,\tau_{-i}\}$, conditioned on the history $\mathcal H_{t-1}$, $\Pr(r_i[t]=1\mid\mathcal H_{t-1}) = 1/n$ since agent $i$ bids with probability $1/n$ as long as she has budget remaining. Then, $\sum_{s=1}^{\min\{t,\tau_i,\tau_{-i}\}} W_i[s] - \min\{t,\tau_i,\tau_{-i}\}/(2n)$ is a supermartingale, so by the Azuma-Hoeffding inequality,
\begin{equation*}
    \sum_{t=1}^{\min\{\tau_i,\tau_i'\}} W_i[t] \leq \frac{\min\{\tau_i,\tau_{-i}\}}{2n} + \sqrt{T\ln T}
\end{equation*}
with probability at least $1-1/T^2$. Since we established before that $\tau_i\leq \tau_i'$ with probability at least $1-1/T^2$, and the fact that $W_i[t]=0$ for $t>\tau_i$, by the union bound, the above display implies that
\begin{equation*}
    \sum_{t=1}^T W_i[t]\leq \frac{T}{n} + \sqrt{T\ln T}
\end{equation*}
with probability at least $1-2/T^2$.

By the above, \cref{lem:bernoulli_reduction} (which was proved for \allPayMechanism{} but the exact proof also works for \allPayMechanismGeneralCost since we did not use the payment structure in the proof), and the union bound,
\begin{equation*}
    \Pr\left(\frac1T\sum_{t=1}^T U_i[t] \leq \frac{v_i^*}{2} + O\left(\sqrt{\frac{\log T}{T}}\right)\right)\geq 1-O\left(\frac{1}{T^2}\right).
\end{equation*}
It follows that the \oneOverNAggressiveStrategy{} cannot be $\lambda$-robust for $\lambda > 1/2$ for $T$ sufficiently large, a contradiction.

We have now established that $p_1 > 0$. For the upper bounds on the $p_k$, suppose each agent uses a \oneOverNAggressiveStrategy. Let $\tau_i$ be the time at which agent $i$ runs out of budget, and $T$ if no agent ever runs out of budget:
\begin{equation*}
    \tau_i = \begin{cases}\min\left\{t:\sum_{s=1}^tP_i[s]\geq \frac Tn\right\} & \text{if $\sum_{s=1}^T P_i[s] \geq \frac Tn$}\\T & \text{otherwise}\end{cases}.
\end{equation*}
Let $\mathcal H_t$ denote the history up to and including time $t$. Notice that $\tau_i$ is a stopping time with respect to $\mathcal H_t$. Let $\tau = \min_i\tau_i$. At each time $t\leq \tau$, for each agent $i$, by \cref{lem:expected_payment_of_agent_general_cost},
\begin{equation*}
    \E[P_i[t]\mid\mathcal H_{t-1}] = \frac1n\cdot\E\left[\frac{p_{1+X}}{1+X}\right].
\end{equation*}
where $X\sim\Binom(n-1,1/n)$ since each agent is bidding with probability $1/n$. By the Azuma-Hoeffding inequality applied to the martingale $\sum_{s=1}^{\min\{t,\tau\}}P_i[s] - \sum_{s=1}^{\min\{t,\tau\}}\E[P_i[s]\mid \mathcal H_{s-1}]$ with increments bounded by $\max_k p_k$,
\begin{equation} \label{eq:everyone_using_robust_stopping_time}
    \Pr\left(\left|\sum_{t=1}^\tau P_i[t]  - \frac\tau n\cdot\E\left[\frac{p_{1+X}}{1+X}\right]\right| \geq \left(\max_k p_k\right)\sqrt{T\ln T}\right) \leq \frac{2}{T^2}.
\end{equation}
Since $p_1>0$, at each time $t\leq\tau_i$, $\E[P_i[t]\mid\mathcal H_{t-1}]\geq c > 0$ for some constant $c$ not depending on $t$ or $T$ since the probability that agent $i$ is the sole bidder is positive. Then, by the Azuma-Hoeffding inequality applied to the submartingale $\sum_{s=\tau+1}^{\min\{t,\tau_i\}}P_i[s] - c(\min\{t,\tau_i\} - \tau)^+$,
\begin{equation} \label{eq:spending_after_one_agent_ran_out}
    \Pr\left(\sum_{t=\tau+1}^{\tau_i}P_i[t] \leq c(\tau_i - \tau)^+ - \left(\max_k p_k\right)\sqrt{T\ln T}\right) \leq \frac{1}{T^2}.
\end{equation}
For each time $t\leq \tau$,
\begin{equation*}
    \E[W_i[t]\mid\mathcal H_{t-1}] = \E\left[\frac{r_i[t]}{r_i[t] + \sum_{j\neq i}r_j[t]}\,\middle|\,\mathcal H_{t-1}\right] = \frac1n\cdot\E_{X\sim\Binom(n-1,1/n)}\left[\frac{1}{1+X}\right] = \frac1n\left(1-\left(1-\frac1n\right)^n\right),
\end{equation*}
using \cref{prop:expectation_of_functions_of_binom} for the last equality.
By the Azuma-Hoeffding-inequality,
\begin{equation} \label{eq:upper_bound_symmetric_agent_wins}
    \Pr\left(\sum_{t=1}^\tau W_i[t] \geq \frac1n\left(1-\left(1-\frac1n\right)^n\right)\tau + \sqrt{T\ln T}\right) \leq \frac{1}{T^2}.
\end{equation}

Consider what happens on the event $E$ that \cref{eq:everyone_using_robust_stopping_time,eq:spending_after_one_agent_ran_out,eq:upper_bound_symmetric_agent_wins} do not happen for any agent $i$, which has probability $\Pr(E)\geq 1-O(1/T^2)$ by the union bound.

First, consider the case that $\tau = T$, so by the fact that the event in \eqref{eq:everyone_using_robust_stopping_time} does not occur and the fact that $\sum_{t=1}^T P_i[t] \leq T/n + \max_k p_k$ by the budget constraint,
\begin{equation*}
    \frac Tn + \max_k p_k \geq \sum_{t=1}^T P_i[t] \geq  \frac{T}{n}\cdot\E\left[\frac{p_{1+X}}{1+X}\right] - \left(\max_k p_k\right)\sqrt{T\ln T}.
\end{equation*}
It follows that
\begin{equation} \label{eq:expected_payment_loose_bound_robust}
    \E\left[\frac{p_{1+X}}{1+X}\right] \leq 1 + O\left(\sqrt{\frac{\log T}{T}}\right).
\end{equation}

Now consider the case when $\tau < T$. Then $\sum_{t=1}^\tau P_i[t] \geq T/n$ for some $i$ by definition of $\tau$. By the fact that the event in \eqref{eq:everyone_using_robust_stopping_time} does not happen,
\begin{equation} \label{eq:stopping_time_some_agent_runs_out}
    \frac Tn \leq \frac{\tau}{n}\cdot\E\left[\frac{p_{1+X}}{1+X}\right] + \left(\max_k p_k\right)\sqrt{T\ln T}.
\end{equation}

For each agent $j$, since $\sum_{t=1}^{\tau_j}P_j[t]\leq T/n + \max_k p_k$, by the fact that \eqref{eq:everyone_using_robust_stopping_time} and \eqref{eq:spending_after_one_agent_ran_out} do not happen, and then using the above display
\begin{equation*}
\begin{split}
    \frac Tn + \max_k p_k & \geq \sum_{t=1}^{\tau}P_j[t] + \sum_{t=\tau+1}^{\tau_j}P_j[t]\\
    & \geq \frac{\tau}{n}\cdot\E\left[\frac{p_{1+X}}{1+X}\right] + c(\tau_j-\tau)^+ - O\left(\sqrt{T\log T}\right)\\
    & \geq \frac Tn + c(\tau_j - \tau)^+ - O\left(\sqrt{T\log T}\right).
\end{split}
\end{equation*}
Therefore, for each agent $j$,
\begin{equation*}
    \tau_j - \tau \leq O\left(\sqrt{T\log T}\right).
\end{equation*}
Again by the fact that the event in \eqref{eq:everyone_using_robust_stopping_time} does not happen, 
\begin{equation*} 
    \frac Tn \geq \frac{\tau}{n}\cdot\E\left[\frac{p_{1+X}}{1+X}\right] - \left(\max_k p_k\right)\sqrt{T\ln T}.
\end{equation*}
By the fact that the event in \eqref{eq:upper_bound_symmetric_agent_wins} does not happen and the above,
\begin{equation*}
\begin{split}
    \sum_{t=1}^{\tau_j} W_j[t] & \leq \frac{\tau}{n}\left(1-\left(1-\frac1n\right)^n\right) + (\tau_j-\tau)^+ + O\left(\sqrt{T\log T}\right)\\
    & \leq \frac{\tau}{n}\left(1-\left(1-\frac1n\right)^n\right) + O\left(\sqrt{T\log T}\right)\\
    & \leq \frac{T}{n\E\left[\frac{p_{1+X}}{1+X}\right]}\left(1-\left(1-\frac1n\right)^n\right) + O\left(\sqrt{T\log T}\right).
\end{split}
\end{equation*}
This happens on the event $E$; in expectation, recalling that $\Pr(E)\leq O(1/T^2)$,
\begin{equation*}
\begin{split}
    \frac1T\sum_{t=1}^T \E[W_j[t]] & \leq \frac{1}{n\E\left[\frac{p_{1+X}}{1+X}\right]}\left(1-\left(1-\frac1n\right)^n\right) + T(1-\Pr(E)) + O\left(\sqrt{\frac{\log T}{T}}\right)\\
    & \leq \frac{1}{n\E\left[\frac{p_{1+X}}{1+X}\right]}\left(1-\left(1-\frac1n\right)^n\right) + O\left(\sqrt{\frac{\log T}{T}}\right)\\
    & \leq \frac{1}{n\E\left[\frac{p_{1+X}}{1+X}\right]} + O\left(\sqrt{\frac{\log T}{T}}\right)
\end{split}
\end{equation*}
By \cref{lem:bernoulli_reduction} and the hypothesis that a $1/n$-aggressive strategy is $\lambda$-robust for some $\lambda > 1/2$, the right-hand side is at least $1/(2n)$ for all $T$ sufficiently large, which implies
\begin{equation} \label{eq:expected_payment_not_too_high_robust}
    \E\left[\frac{p_{1+X}}{1+X}\right] \leq 2.
\end{equation}

Notice that \eqref{eq:expected_payment_not_too_high_robust} holds in either case; we proved it one case and \eqref{eq:expected_payment_loose_bound_robust} in the other. In particular, from \eqref{eq:expected_payment_not_too_high_robust}, all of $p_1\Pr(X=0)$, $\frac{p_2}{2}\cdot\Pr(X=1)$, and $\frac{p_3}{3}\cdot\Pr(X=2)$ are at most $2$. It can be computed that
\begin{equation*}
    \Pr(X=0) = \left(1-\frac1n\right)^{n-1} \geq \frac1e, \Pr(X=1) = \left(1-\frac1n\right)^{n-2}\geq \frac1e, \Pr(X=2) = \frac{n-1}{2n}\left(1-\frac1n\right)^{n-3} \geq \frac{1}{2e}.
\end{equation*}
Therefore,
\begin{equation*}
    \frac{p_1}{e}\leq 2, \frac{p_2}{2e}\leq 2, \frac{p_3}{6e}\leq 2.
\end{equation*}
The lemma statement follows.

\end{proof}

Fix any $(p_i)_{i=1}^n$. In the rest of this section, suppose agent $i$ uses an $\alpha_i$-aggressive strategy and the other agents $j\neq i$ use the following strategy. They choose a $k\in\{1,2,\dots,n-1\}$. At each time $t$, the $k$ agents $j\neq i$ with the highest budget remaining bid. If there are less than $k$ agents other than $i$ with remaining budget, no agent $j\neq i$ bids. We shall upper bound agent $i$'s utility in this scenario.

Let $\tau$ be the time at which agent $i$ runs out of budget and $T$ if they never run out of budget,
\begin{equation*}
    \tau = \begin{cases}\min\left\{t:\sum_{s=1}^tP_i[s] \geq \alpha_iT\right\} & \text{if $\sum_{s=1}^T P_i[s] \geq \alpha_iT$}\\ T & \text{otherwise}\end{cases}.
\end{equation*}
and let $\tau'$ be the last time that agents $j\neq i$ bid.

\begin{lemma} \label{lem:stopping_time_comparisons}
Let $\gamma = \min\{1, \tau'/\tau\}$. With probability at least $1-O(1/T^2)$, 
\begin{equation} \label{eq:upper_bound_agent_stopping_time}
    \frac{\tau}{T} \leq \min\left\{1, \frac{1}{\frac{p_{k+1}}{k+1}\cdot\gamma + p_1(1-\gamma)}\right\} + O\left(\sqrt{\frac{\log T}{T}}\right)
\end{equation}
and
\begin{equation} \label{eq:lower_bound_x_k}
    \gamma \geq \min\left\{1, \max\left\{\frac{1-\alpha_i}{\left((1-\alpha_i)\cdot\frac{p_k}{k} + \alpha_i\cdot\frac{p_{k+1}}{k+1}\right)k},\frac{p_1}{\frac{\left((1-\alpha_i)\cdot\frac{p_k}{k} + \alpha_i \cdot\frac{p_{k+1}}{k+1}\right)k}{(1-\alpha_i)} - \frac{p_{k+1}}{k+1} + p_1} \right\}\right\} - O\left(\sqrt{\frac{\log T}{T}}\right).
\end{equation}
\end{lemma}
\begin{proof}
Let $\mathcal H_t$ denote the history up to, and including, time $t$. Notice that $\tau$ and $\tau'$ are stopping times with respect to the filtration $\mathcal H_t$. For any time $t\leq\min\{\tau,\tau'\}$, by \cref{lem:expected_payment_of_agent_general_cost},
\begin{equation*}
    \E\left[\sum_{j\neq i}P_j[t]\,\middle|\,\mathcal H_{t-1}\right] = \left((1-\alpha_i)\cdot\frac{p_k}{k} + \alpha_i\cdot \frac{p_{k+1}}{k+1}\right)k
\end{equation*}
since $k$ agents $j\neq i$ are bidding, and with probability $1-\alpha_i$, agent $i$ does not bid, and with probability $\alpha_i$, agent $i$ bids so there are $k+1$ total bidders. By the Azuma-Hoeffding inequality applied to the martingale $\sum_{s=1}^{\min\{t,\tau,\tau'\}}\sum_{j\neq i}P_j[s] - \sum_{s=1}^{\min\{t,\tau,\tau'\}}\sum_{j\neq i}\E[P_j[s]\mid\mathcal H_{s-1}]$ that has increments bounded by $\max\{p_k, p_{k+1}\}$,
\begin{equation} \label{eq:upper_bound_payment_adversary}
    \left|\sum_{t=1}^{\min\{\tau,\tau'\}}\sum_{j\neq i}P_j[t] - \left((1-\alpha_i)\cdot\frac{p_k}{k} + \alpha_i\cdot \frac{p_{k+1}}{k+1}\right)k\min\{\tau,\tau'\}\right| \leq \max\{p_k, p_{k+1}\}\sqrt{T\ln T}
\end{equation}
with probability at least $1-2/T^2$. For $t\leq \min\{\tau,\tau'\}$, $\E[P_i[t]\mid\mathcal H_{t-1}] = \alpha_i\cdot \frac{p_{k+1}}{k+1}$ by \cref{lem:expected_payment_of_agent_general_cost} because agent $i$ bids with probability $\alpha_i$ and $k$ others also bid, and for $\tau'+1\leq t\leq \tau$, $\E[P_i[t]\mid\mathcal H_{t-1}] = \alpha_i p_1$ since agent $i$ bids with probability $\alpha_i$ and no others bid, so by the Azuma-Hoeffding inequality applied to the martingale $\sum_{s=1}^{\min\{t,\tau\}}P_i[s] - \sum_{s=1}^{\min\{t,\tau\}}\E[P_i[s]\mid\mathcal H_{s-1}]$ with increments bounded by $\max\{p_1, p_{k+1}\}$,
\begin{equation} \label{eq:lower_bound_payment_agent}
    \left|\sum_{t=1}^\tau P_i[t] - \alpha_i\left(\frac{p_{k+1}}{k+1}\min\{\tau,\tau'\} + p_1(\tau - \tau')^+\right)\right| \leq \max\{p_1, p_{k+1}\}\sqrt{T\ln T}
\end{equation}
with probability at least $1-2/T^2$.

Consider what happens on the event of probability at least $1-4/T^2$ that \eqref{eq:upper_bound_payment_adversary} and \eqref{eq:lower_bound_payment_agent} happen. 

We first prove \eqref{eq:upper_bound_agent_stopping_time}.
Notice that $\sum_{t=1}^\tau P_i[t] \leq \alpha_i T + \max\{p_1, p_{k+1}\}$ by the budget constraint; substituting into \eqref{eq:lower_bound_payment_agent},
\begin{equation} \label{eq:stopping_condition_hardness}
\begin{split}
    \alpha_i T + \max\{p_1, p_{k+1}\} & \geq \alpha_i\left(\frac{p_{k+1}}{k+1}\min\{\tau,\tau'\} + p_1(\tau - \tau')^+\right) - \max\{p_1, p_{k+1}\}\sqrt{T\ln T}.
\end{split}
\end{equation}
Notice that $\min\{\tau,\tau'\} = \min\{1,\tau'/\tau\}\tau = \gamma\tau$ and $(\tau - \tau')^+ = \tau - (\min\{\tau,\tau'\}) = (1 - \min\{1, \tau'/\tau\})\tau = (1-\gamma)\tau$. Substituting into \eqref{eq:stopping_condition_hardness}, and simplifying, we obtain
\begin{equation*}
    T \geq \left(\frac{p_{k+1}}{k+1}\cdot\gamma + p_1(1-\gamma)\right)\tau - O\left(\sqrt{T\log T}\right).
\end{equation*}
Rearranging,
\begin{equation} \label{eq:upper_bound_agent_stopping_time_general}
    \frac{\tau}{T} \leq \frac{1}{\frac{p_{k+1}}{k+1}\cdot\gamma + p_1(1-\gamma)} + O\left(\sqrt{\frac{\log T}{T}}\right).
\end{equation}
This is \eqref{eq:upper_bound_agent_stopping_time}.

Now, prove \eqref{eq:lower_bound_x_k}.
If $\tau \leq\tau'$, then $\gamma=1$, so \eqref{eq:lower_bound_x_k} holds trivially.

Assume $\tau' < \tau$. At time $\tau'$, by the greedy rule of the agents' $j\neq i$ strategy, each agent $j\neq i$ has budget at most $2\max\{p_k, p_{k+1}\}$. It follows that $\sum_{t=1}^\tau\sum_{j\neq i}P_j[t] \geq (1-\alpha_i)T - 2(n-1)\max\{p_k, p_{k+1}\}$. Then, by \eqref{eq:upper_bound_payment_adversary},
\begin{equation*}
    (1-\alpha_i)T - 2(n-1)\max\{p_k, p_{k+1}\} \leq \left((1-\alpha_i)\cdot\frac{p_k}{k} + \alpha_i\cdot\frac{p_{k+1}}{k+1}\right)k\tau'  + O\left(\sqrt{T\log T}\right).
\end{equation*}
Then,
\begin{equation} \label{eq:lower_bound_adversary_stopping_time_one}
    \tau' \geq \frac{(1-\alpha_i)T}{\left((1-\alpha_i)\cdot\frac{p_k}{k} + \alpha_i \cdot\frac{p_{k+1}}{k+1}\right)k} - O\left(\sqrt{T\log T}\right).
\end{equation}
Then,
\begin{equation} \label{eq:first_nontrivial_lower_bound_on_x_k}
    \frac{\tau'}{\tau} \geq \frac{\tau'}{T} \geq \frac{(1-\alpha_i)}{\left((1-\alpha_i)\cdot \frac{p_k}{k} + \alpha_i\cdot \frac{p_{k+1}}{k+1}\right)k} - O\left(\sqrt{\frac{\log T}{T}}\right).
\end{equation}

Also, solving for $T$ in \eqref{eq:lower_bound_adversary_stopping_time_one},
\begin{equation} \label{eq:lower_bound_adversary_stopping_time}
    T \leq \frac{\left((1-\alpha_i)\cdot \frac{p_k}{k} + \alpha_i \cdot\frac{p_{k+1}}{k+1}\right)k}{(1-\alpha_i)}\cdot\tau' + O\left(\sqrt{T\log T}\right).
\end{equation}
We have $\sum_{t=1}^\tau P_i[t] \leq \alpha_iT + \max\{p_1, p_{k+1}\}$, so by \eqref{eq:lower_bound_payment_agent},
\begin{equation*}
    \alpha_i T + \max\{p_1, p_{k+1}\} > \alpha_i\left(\frac{p_{k+1}}{k+1}\cdot\tau' + p_1(\tau - \tau')\right) - O\left(\sqrt{T\log T}\right).
\end{equation*}
Substituting \eqref{eq:lower_bound_adversary_stopping_time},
\begin{equation*}
    \alpha_i\cdot \frac{((1-\alpha_i)\cdot\frac{p_k}{k} + \alpha_i \cdot\frac{p_{k+1}}{k+1})k}{(1-\alpha_i)}\cdot\tau' > \alpha_i\left(\frac{p_{k+1}}{k+1}\cdot\tau' + p_1(\tau - \tau')\right) - O\left(\sqrt{T\log T}\right).
\end{equation*}
Rearranging, we obtain
\begin{equation*}
    \frac{\tau'}{\tau} \geq \frac{p_1}{\frac{\left((1-\alpha_i)\cdot\frac{p_k}{k} + \alpha_i \cdot\frac{p_{k+1}}{k+1}\right)k}{(1-\alpha_i)} - \frac{p_{k+1}}{k+1} + p_1} - O\left(\sqrt{T\log T}\right).
\end{equation*}
Notice that $\gamma = \tau'/\tau$ in this case where $\tau' \leq \tau$. The above display and \eqref{eq:first_nontrivial_lower_bound_on_x_k} establish \eqref{eq:lower_bound_x_k}.

\end{proof}

Now we upper bound the utility of the agent in terms of $\gamma$.
\begin{lemma} \label{lem:stopping_time_to_utility}
Let $\gamma = \min\{1,\tau'/\tau\}$. With probability at least $1-1/T^2$,
\begin{equation*}
    \frac1T\sum_{t=1}^T U_i[t] \leq \left(\frac{\gamma}{1+k} + (1-\gamma)\right)\min\left\{1, \frac{1}{\frac{p_{k+1}}{k+1}\cdot\gamma + p_1(1-\gamma)}\right\}v_i^* + O\left(\sqrt{\frac{\log T}{T}}\right).
\end{equation*}
\end{lemma}
\begin{proof}
As before, let $\mathcal H_t$ denote the history up to and including time $t$. For $t\leq\min\{\tau,\tau'\}$, $\E[W_i[t]\mid\mathcal H_{t-1}] = \alpha_i/(k+1)$, since agent $i$ bids with probability $\alpha_i$, in which case there are $k+1$ total bidding agents. For $\tau+1'\leq t\leq \tau$, $\E[W_i[t]\mid\mathcal H_{t-1}]=\alpha_i$, since agent $i$ bids with probability $\alpha_i$ and no other agent bids. Thus, by Azuma-Hoeffding applied to the martingale $\sum_{s=1}^t W_i[s] - \sum_{s=1}^t \E[W_i[s]\mid\mathcal H_{s-1}]$,
\begin{equation*} 
    \sum_{t=1}^T W_i[t] \leq \frac{\alpha_i}{1+k}\cdot\min\{\tau,\tau'\} + \alpha_i(\tau - \tau')^+ + \sqrt{T\ln T}
\end{equation*}
with probability at least $1-1/T^2$. Using the definition of $\gamma$, we can compute that $\min\{\tau,\tau'\} = \gamma\tau$ and $(\tau-\tau')^+ = (1-\gamma)\tau$, so the above display implies that
\begin{equation*} 
    \sum_{t=1}^T W_i[t] \leq \alpha_i\left(\frac{\gamma}{1+k} + (1-\gamma)\right)\tau.
\end{equation*}
Substituting \eqref{eq:upper_bound_agent_stopping_time} from \cref{lem:stopping_time_comparisons} in,
\begin{equation*}
    \frac1T\sum_{t=1}^T W_i[t] \leq \frac{\alpha_i}{T}\left(\frac{\gamma}{1+k} + (1-\gamma)\right)\min\left\{1, \frac{1}{\frac{p_{k+1}}{k+1}\cdot\gamma + p_1(1-\gamma)}\right\} + O\left(\sqrt{\frac{\log T}{T}}\right)
\end{equation*}
with probability at least $1-O(1/T^2)$. By \cref{lem:bernoulli_reduction},
\begin{equation*}
    \frac1T\sum_{t=1}^T U_i[t] \leq \left(\frac{\gamma}{1+k} + (1-\gamma)\right)\min\left\{1, \frac{1}{\frac{p_{k+1}}{k+1}\cdot\gamma + p_1(1-\gamma)}\right\}v_i^* + O\left(\sqrt{\frac{\log T}{T}}\right)
\end{equation*}
with probability at least $1-O(1/T^2)$.
\end{proof}

Now we prove \cref{thm:robustness_upper_bound}.
\begin{proof}[Proof of \cref{thm:robustness_upper_bound}]
Assume each $\alpha_i = 1/n$. It follows from \cref{lem:stopping_time_comparisons,lem:stopping_time_to_utility} that the maximum robustness of a $1/n$-aggressive strategy is no higher than
\begin{align*}
    \max_{\substack{(p_1, \dots, p_n)\in[0,\infty)^n}}\min_{k\in\{1,2,\dots,n-1\}}\left(\frac{\gamma(p,k)}{1+k} + (1-\gamma(p,k))\right)\min\left\{1, \frac{1}{\frac{p_{k+1}}{k+1}\cdot\gamma(p,k) + p_1(1-\gamma(p,k))}\right\} + O\left(\sqrt{\frac{\log T}{T}}\right)
\end{align*}
for some
\begin{equation*}
    1 \geq \gamma(p,k) \geq \min\left\{1, \max\left\{\frac{1-1/n}{\left((1-1/n)\frac{p_k}{k} + \frac{p_{k+1}}{(k+1)n}\right)k},\frac{p_1}{\frac{\left((1-1/n)\frac{p_k}{k} + \frac{p_{k+1}}{(k+1)n}\right)k}{(1-1/n)} - \frac{p_{k+1}}{k+1} + p_1} \right\}\right\} - O\left(\sqrt{\frac{\log T}{T}}\right).
\end{equation*}
We know that there are $(p_1, \dots, p_n)$ such that a $1/n$-aggressive strategy is at least $\lambda$-robust for some $\lambda > 1/2$ not depending on $T$ for all $T$ sufficiently large by \cref{thm:robustness_theorem}, so applying \cref{lem:crude_upper_bounds}, we can assume $0 < p_1 \leq 2e$, $p_2\leq 4e$, and $p_3\leq 12e$.
It follows that robustness of a $1/n$-aggressive strategy is not more than $O\left(\sqrt{\frac{\log T}{T}}\right)$ than the value of the following optimization problem.
\begin{align*}
    & \max_{\substack{0 < p_1\leq 2e\\0\leq p_2\leq 4e\\0\leq p_3\leq 12e\\p_4,p_5,\dots,p_n\in[0,\infty)}}\min_{k\in\{1,2,\dots,n-1\}}\max_{\gamma}\left(\frac{\gamma}{1+k} + (1-\gamma)\right)\min\left\{1, \frac{1}{\frac{p_{k+1}}{k+1}\cdot\gamma + p_1(1-\gamma)}\right\}\\
    & \text{s.t.} \quad 1\geq \gamma \geq \min\left\{1, \max\left\{\frac{1-1/n}{\left((1-1/n)\frac{p_k}{k} + \frac{p_{k+1}}{(k+1)n}\right)k},\frac{p_1}{\frac{\left((1-1/n)\frac{p_k}{k} + \frac{p_{k+1}}{(k+1)n}\right)k}{(1-1/n)} - \frac{p_{k+1}}{k+1} + p_1} \right\}\right\}
\end{align*}
By taking $n$ sufficiently large, the above problem's objective value can be made arbitrarily close to that of the below problem.
\begin{align}
    & \max_{\substack{0 < p_1\leq 2e\\0\leq p_2\leq 4e\\0\leq p_3\leq 12e\\p_4,p_5,\dots,p_n\in[0,\infty)}}\min_{k\in\{1,2,\dots,n-1\}}\max_{\gamma}\left(\frac{\gamma}{1+k} + (1-\gamma)\right)\min\left\{1, \frac{1}{\frac{p_{k+1}}{k+1}\cdot\gamma + p_1(1-\gamma)}\right\} \nonumber\\
    & \text{s.t.} \quad 1\geq  \gamma \geq \min\left\{1, \max\left\{\frac{1}{p_k},\frac{p_1}{p_k - \frac{p_{k+1}}{k+1} + p_1} \right\}\right\} \label{eq:upper_bound_robustness_optimization_problem_constraint}
\end{align}
Since the objective function can be written as the minimum of two linear-fractional functions with nonnegative denominator and thus quasiconvex functions in $\gamma$, the maximum over $\gamma$ occurs when either \eqref{eq:upper_bound_robustness_optimization_problem_constraint} is tight or if the two linear-fractional functions are equal, i.e.,
\begin{equation*}
    1 = \frac{1}{\frac{p_{k+1}}{k+1}\cdot\gamma + p_1(1-\gamma)}\iff \gamma = \frac{p_1}{\frac{p_{k+1}}{k+1}-p_1}.
\end{equation*}
Also, note that the $1\geq \gamma$ should not be tight at the maximizing $\gamma$, because then, then objective function is at most $\frac{1}{1+k}\leq 1/2$, and we know from the greater than $1/2$-robustness that the optimal objective is greater than $1/2$. Thus, only the $\gamma \geq \min\left\{1, \max\left\{\frac{1}{p_k},\frac{p_1}{p_k - \frac{p_{k+1}}{k+1} + p_1} \right\}\right\} $ should be tight at the maximizing $\gamma$. It follows from the previous arguments that we can add the constraint
\begin{equation*}
    \gamma\in\left\{\min\left\{1, \max\left\{\frac{1}{p_k},\frac{p_1}{p_k - \frac{p_{k+1}}{k+1} + p_1} \right\}\right\}, \frac{p_1}{\frac{p_{k+1}}{k+1}-p_1}\right\}.
\end{equation*}
without changing the objective value. Also, since we are upper bounding, we can assume $k\in\{1,2\}$, in which case the variables $p_4, p_5,\dots, p_n$ are unused so we can get rid of them. We then obtain the optimization program in the theorem statement.

\end{proof}

\section{Deferred Proofs from Section \ref{sec:equilibrium_mechanism}} \label{sec:app:equilibrium_mechanism_proofs}
In this section, we prove \cref{thm:equilibrium_mechanism}, which we restate here for completeness.

\equilibriumTheorem*

Throughout this section, we assume we use the parameters
\begin{equation*}
    \paymentConstant = \frac{n+1}{1+n(1-1/n)^{n+1}}
\end{equation*}
and
\begin{equation*}
    \epsilon = \sqrt{T\ln T}.
\end{equation*}

Let $\tau_i$ be the time at which agent $i$ runs out of budget and $T$ if the agent never runs out of budget:
\begin{equation*}
    \tau_i = \begin{cases}\min\left\{t:\sum_{s=1}^tP_i[s] \geq \frac Tn\right\} & \text{if $\sum_{s=1}^T P_i[s] \geq \alpha_iT$}\\ T & \text{otherwise}\end{cases}.
\end{equation*}
Define
\begin{equation*}
    \tau_i' = \begin{cases}\min\left\{t:\sum_{s=1}^t r_i[s] \leq \frac tn - \epsilon + 1\right\} & \text{if there exists a $t$ such that $\sum_{s=1}^t r_i[s] \leq \frac tn - \epsilon + 1$}\\T & \text{otherwise}\end{cases}.
\end{equation*}
We can see that with this definition of $\tau_i'$, if $t\leq \tau_i'$, then the minimum bidding constraint does not affect player $i$ at time $t$, in that it guarantees that $\sum_{s=1}^t r_i[s] > t/n - \epsilon$. Indeed, if $t \leq \tau_i'$, then
\begin{equation*}
    \sum_{s=1}^t r_i[s] \geq \sum_{s=1}^{t-1}r_i[s] > \frac{t-1}{n}-\epsilon + 1 \geq \frac{t}{n}-\epsilon.
\end{equation*}

Let $\mathcal H_t$ denote the history up to and including time $t$. Notice that the $\tau_i$ and $\tau_i'$ are stopping times with respect to the filtration $\mathcal H_t$.

The below lemma lower bounds $\tau_i'$ to show that the bidding minimum is not enforced with high probability as long as agents are using a $1/n$-aggressive strategy.

\begin{lemma} \label{lem:robustness_stopping_time_lower_bound}
If agent $i$ uses a $1/n$-aggressive strategy, regardless of the strategies of the other agents, $\tau_i \leq \tau_i'$ with probability at least $1-O(1/T^2)$. 
\end{lemma}
\begin{proof}
Assume agent $i$ uses a \oneOverNAggressiveStrategy. For any time $t\leq\min\{\tau_i,\tau_i'\}$, $\E[r_i[t]\mid\mathcal H_{t-1}] = 1/n$. It follows that $\sum_{s=1}^{\min\{t,\tau_i,\tau_i'\}}r_i[s] - \min\{t, \tau_i,\tau_i'\}/n$ is an $\mathcal H_t$-martingale, so by the Azuma-Hoeffding inequality,
\begin{equation*}
    \Pr\left(\sum_{t=1}^{\min\{\tau_i,\tau_i'\}} r_i[t] \leq \frac{\min\{\tau_i,\tau_i'\}}{n} - \epsilon + 1\right) \leq \exp\left(-\frac{2(\epsilon-1)^2}{T}\right) = O\left(\frac1{T^2}\right)
\end{equation*}
for $\epsilon = \sqrt{T\ln T}$. If $\tau_i'< T$ and $\tau_i'\leq \tau_i$, then the above event happens, so $\tau_i \leq \tau_i'$ with probability at least $1-O(1/T^2)$.
\end{proof}

The below lemma says that if all but one agent uses a $1/n$-aggressive strategy, then they will not run out of budget too quickly. This prevents a single agent from being able to run other agents out of budget.

\begin{lemma} \label{lem:equilibrium_stopping_time_lower_bound}
If all agents $j\neq i$ use a \oneOverNAggressiveStrategy, regardless of the strategy of agent $i$, $\min_{j\neq i}\tau_j \geq T - O\left(\sqrt{T\log T}\right)$ with probability at least $1-O(1/T^2)$.
\end{lemma}
\begin{proof}
Assume each agent $j\neq i$ uses a \oneOverNAggressiveStrategy{} and that agent $i$'s strategy is arbitrary. Let $\mathcal G_t$ be the $\sigma$-algebra generated by $\mathcal H_t$ and $r_i[t+1]$. Let $\tau = \min_{j\neq i}\tau_j$ and $\tau' = \min_{j\neq i}\tau_j'$. We must lower bound $\tau$. For any $t\leq \min\{\tau, \tau'\}$ and $j\neq i$,
\begin{equation*}
    \E[P_j[t]\mid\mathcal G_{t-1}] = \paymentConstant\E\left[\frac{r_j[t]}{1 + r_i[t] + r_j[t] + \sum_{k\neq i,j}r_k[t]}\,\middle|\,\mathcal G_{t-1}\right] = \frac{\paymentConstant}{n}\E\left[\frac{1}{2 + r_i[t] + Y}\right]
\end{equation*}
where $Y\sim\Binom(n-2,1/n)$ is independent of $r_i[t]$. The first equality follows directly from \cref{lem:expected_payment_of_agent}, and the second equality follows from the fact that each agent other than $i$ bids with probability $1/n$, independently of all other agents when conditioned on $\mathcal G_{t-1}$. It follows that
\begin{equation*}
    \sum_{s=1}^{\min\{t,\tau,\tau'\}}P_j[s] - \frac{\paymentConstant}{n}\sum_{s=1}^{\min\{t,\tau,\tau'\}}\E\left[\frac{1}{2 + r_i[s] + Y}\right]
\end{equation*}
is a $\mathcal G_t$-martingale. Notice that the $P_j[t]$ are bounded by $\paymentConstant$. By Azuma-Hoeffding,
\begin{equation}
\label{eq:following_agents_payments_are_not_too_high}
    \Pr\left(\sum_{t=1}^{\min\{\tau,\tau'\}}P_j[t] \geq \frac{\paymentConstant}{n}\sum_{s=1}^{\min\{\tau,\tau'\}}\E\left[\frac{1}{2 + r_i[s] + Y}\right] + \paymentConstant\sqrt{T\ln T}\right) \leq \exp\left(-\frac{2T\ln T}{\paymentConstant T}\right) \leq \frac{1}{T^2}.
\end{equation}
Also, note that from \cref{lem:robustness_stopping_time_lower_bound}, $\tau_j\leq \tau_j'$ with probability at least $1-O(1/T^2)$ for each $j\neq i$. If $\tau_j\leq \tau_j'$ for each $j\neq i$, then $\tau \leq \tau'$. 
Consider what happens on the event that $\tau \leq \tau'$ and that the event in \eqref{eq:following_agents_payments_are_not_too_high} does not happen for any $j\neq i$ which has probability at least $1-O(1/T^2)$ by the union bound.
In that case,
\begin{equation*}
    \sum_{s=1}^\tau P_j[s] \leq \frac{\paymentConstant}{n}\sum_{s=1}^\tau\E\left[\frac{1}{2 + r_i[s] + Y}\right] + \paymentConstant\sqrt{T\ln T}.
\end{equation*}
for every $j\neq i$. If $\tau=T$, then there is nothing to prove, so assume $\tau < T$. Then, $\sum_{s=1}^\tau P_j[s] \geq T/n$ for some $j$, so
\begin{equation} \label{eq:following_agents_budget_stopping_time}
    \frac{T}{n} \leq \frac{\paymentConstant}{n}\sum_{s=1}^\tau\E\left[\frac{1}{2 + r_i[s] + Y}\right] + \paymentConstant\sqrt{T\ln T}.
\end{equation}
Write
\begin{equation*}
    \sum_{s=1}^\tau\E\left[\frac{1}{2 + r_i[s] + Y}\right] = \E\left[\frac{1}{2  + Y}\right]\#\{s\leq\tau:r_i[s]=0\} + \E\left[\frac{1}{3 + Y}\right]\#\{s\leq\tau:r_i[s]=1\}.
\end{equation*}
By the minimum bidding rule of the mechanism, we see that $\#\{s\leq\tau:r_i[s]=1\} \geq \tau/n- \epsilon$. It follows that
\begin{equation*}
\begin{split}
    \sum_{s=1}^\tau \E\left[\frac{1}{2 + r_i[s] + Y}\right] & \leq \E\left[\frac{1}{2 + Y}\right]\left(\tau - \left(\frac\tau n-\epsilon\right)\right) + \E\left[\frac{1}{3  + Y}\right]\left(\frac\tau n-\epsilon\right)\\
    & = \left(\E\left[\frac{1}{2+Y}\right]\left(1-\frac1n\right) + \E\left[\frac{1}{3+Y}\right]\cdot\frac1n\right)\tau + O\left(\sqrt{T\log T}\right),
\end{split}
\end{equation*}
using the fact that $\epsilon = \sqrt{T\ln T}$ for second line. Use \cref{prop:expectation_of_functions_of_binom} to see that 
\begin{equation*}
    \E\left[\frac{1}{2+Y}\right]\left(1-\frac1n\right) + \E\left[\frac{1}{3+Y}\right]\cdot\frac1n = \frac{1+n(1-1/n)^{n+1}}{n+1} = \frac{1}{\paymentConstant}.
\end{equation*}
Then,
\begin{equation*}
    \sum_{s=1}^\tau \E\left[\frac{1}{2 + r_i[s] + Y}\right] \leq \frac{\tau}{\paymentConstant} + O\left(\sqrt{T\log T}\right),
\end{equation*}
and substituting into \eqref{eq:following_agents_budget_stopping_time}, we obtain
\begin{equation*}
    \frac{T}{n} \leq \frac{\tau}{n} + O\left(\sqrt{T\log T}\right).
\end{equation*}
Solving for $\tau$, we obtain the desired lower bound on $\tau$.
\end{proof}

The below lemma shows that the robustness guarantee from \cref{thm:general_robustness_claim} when substituting our new value of $\paymentConstant$ is always at least $5/(3e)$.
\begin{lemma} \label{lem:robustness_with_equilibrium_payment_constant}
With our choice of
\begin{equation*}
    \paymentConstant = \frac{n+1}{1+n(1-1/n)^{n+1}},
\end{equation*}
for any $n$,
\begin{equation*}
    \min\left\{1-\frac{3(1-1/n)}{3\paymentConstant - \paymentConstant/n}, \frac{5-1/n}{\paymentConstant(3-1/n)}\right\} \geq \frac{5}{3e}.
\end{equation*}
\end{lemma}
\begin{proof}
We have 
\begin{equation*}
\begin{split}
    &\min\left\{1-\frac{3(1-1/n)}{3\paymentConstant - \paymentConstant/n}, \frac{5-1/n}{\paymentConstant(3-1/n)}\right\}\\
    = &\min\left\{1-\frac{3(1-1/n)}{3-1/n}\cdot\frac{1+n(1-1/n)^{n+1}}{n+1},\frac{5-1/n}{3-1/n}\cdot\frac{1+n(1-1/n)^{n+1}}{n+1}\right\}.
\end{split}
\end{equation*}
The factors $\frac{5-1/n}{3-1/n}$ and $\frac{1+n(1-1/n)^{n+1}}{n+1}$ can be seen to be decreasing in $n$, so the second term in the minimum above is decreasing in $n$. This second term can easily be seen to have limit $5/(3e)$. We must now show that the first term in the minimum is always at least the second. We can check that first term is always at least $5/(3e)$ by checking it individually for $n\leq 4$, and for $n\geq 5$,
\begin{equation*}
    1-\frac{3(1-1/n)}{3-1/n}\cdot\frac{1+n(1-1/n)^{n+1}}{n+1} \geq 1 - \frac{1+n(1-1/n)^{n+1}}{n+1} \geq 0.614 > \frac{5}{3e},
\end{equation*}
using monotonicity to check the second-to-last inequality.
\end{proof}

The below proposition deals with the robustness of a \oneOverNAggressiveStrategy.  \cref{thm:general_robustness_claim} gives the robustness for \allPayMechanism, and \cref{lem:robustness_stopping_time_lower_bound} shows that \allPayMechanismWithBiddingMinimum and \allPayMechanism do not differ too much.
\begin{proposition} \label{prop:robustness_in_equilibrium_mechanism}
A \oneOverNAggressiveStrategy{} is $\lambda$-robust for some
\begin{equation*}
    \lambda \geq \frac{5}{3e} - O\left(\sqrt{\frac{\log T}{T}}\right).
\end{equation*}
\end{proposition}
\begin{proof}
Assume agent $i$ uses a \oneOverNAggressiveStrategy. Fix the strategies of players $j\neq i$. Couple \allPayMechanismWithBiddingMinimum{} and \allPayMechanism{} such that agent $i$ is using a \oneOverNAggressiveStrategy{} in both mechanisms, for $t\leq \tau_i'$, agent $i$ requests in \allPayMechanismWithBiddingMinimum{} if and only if she requests in \allPayMechanism, and agents $j\neq i$ request in \allPayMechanismWithBiddingMinimum{} if and only if they request in \allPayMechanism. This is possible to do because \allPayMechanismWithBiddingMinimum{} is just \allPayMechanism{} with a bidding minimum, so any strategy used by agents $j\neq i$ in \allPayMechanismWithBiddingMinimum{}  can be used in \allPayMechanism. Also, agent $i$ can use the same strategy for $t\leq\tau_i'$ because for $t\leq\tau_i'$, the bidding minimum does not affect agent $i$'s ability to request.

By \cref{lem:robustness_stopping_time_lower_bound}, $\tau_i\leq\tau_i'$ with probability at least $1-O(1/T^2)$. On this event, \allPayMechanism{} and \allPayMechanismWithBiddingMinimum{} are the same up to $\tau_i'$, and since the agent does not gain any utility after the time they run out of budget $\tau_i$ which is at most $\tau_i'$, she obtains the same utility in \allPayMechanism as in \allPayMechanismWithBiddingMinimum.

It follows from \cref{thm:general_robustness_claim} that a \oneOverNAggressiveStrategy{} is $\lambda_i$-robust in \allPayMechanismWithBiddingMinimum{} for
\begin{equation*}
\begin{split}
    \lambda_i & \geq \min\left\{1-\frac{3(1-1/n)}{3\paymentConstant - \paymentConstant/n}, \frac{5-1/n}{\paymentConstant(3-1/n)}\right\}- O\left(\sqrt{\frac{\log T}{T}}\right).
\end{split}
\end{equation*}
Note that our choice of $\bar b$ indeed is at least $2$ so we can use \cref{thm:robustness_theorem}. This is at least $5/(3e)$ by \cref{lem:robustness_with_equilibrium_payment_constant}.

\end{proof}

\begin{proposition} \label{prop:equilibrium_utility_claim}
If each agent plays a \oneOverNAggressiveStrategy, each agent obtains at least a $1-(1-1/n)^n - O\left(\sqrt{\frac{\log T}{T}}\right)$-fraction of their ideal utility:
\begin{equation*}
    \frac1T\sum_{t=1}^T W_i[t] \geq \left(1-\left(1-\frac1n\right)^n\right)v_i^* - O\left(\sqrt{\frac{\log T}{T}}\right)
\end{equation*}
with probability at least $1-O(1/T^2)$.
\end{proposition}
\begin{proof}
Let $\tau = \min_{j\neq i}\min\{\tau_j,\tau_j'\}$. By \cref{lem:equilibrium_stopping_time_lower_bound} and the union bound, there is an event of probability at least $1-O(1/T^2)$ on which $\tau \geq T - O\left(\sqrt{T\log T}\right)$. 

Fix an agent $i$.
For $t\leq \tau$,
\begin{align*}
    \E[W_i[t]\mid\mathcal H_{t-1}] & = \E\left[\frac{r_i[t]}{r_i[t] + \sum_{j\neq i}r_j[t]}\,\middle|\,\mathcal H_{t-1}\right] & \text{[Random allocation rule among bidding agents]}\\
    & = \frac1n\cdot\E_{X\sim\Binom(n-1,1/n)}\left[\frac1{1+X}\right] & \text{[Agents independently bid with probability $1/n$]}\\
    & = \frac1n\left(1-\left(1-\frac1n\right)^n\right). & \text{[\cref{prop:expectation_of_functions_of_binom}]}
\end{align*}
Then, $\sum_{s=1}^{\min\{t,\tau\}}W_i[s] - \frac{\min\{t,\tau\}}n\left(1-\left(1-\frac1n\right)^n\right)$ is an $\mathcal H_t$-martingale, so we can apply Azuma-Hoeffding to see that
\begin{equation*}
    \sum_{t=1}^\tau W_i[t] \geq \frac{\tau}{n}\left(1-\left(1-\frac1n\right)^n\right) - \sqrt{T\ln T}
\end{equation*}
with probability at least $1- 1/T^2$. Using the fact that $\tau \geq T - O\left(\sqrt{T\log T}\right)$ with probability at least $1-O(1/T^2)$, we obtain
\begin{equation*}
    \frac1T\sum_{t=1}^\tau W_i[t] \geq \frac{1}{n}\left(1-\left(1-\frac1n\right)^n\right) - O\left(\sqrt{\frac{\log T}{T}}\right)
\end{equation*}
with probability at least $1-O(1/T^2)$. Now we apply \cref{lem:bernoulli_reduction} to obtain the lemma statement.

\end{proof}

For the proof of the next lemma, it is useful to have the notion of $\beta$-ideal utility as defined in \cite{fikioris2025beyond}.

\begin{definition}[$\beta$-Ideal Utility] \label{def:beta_ideal_utility}
The \textit{$\beta$-ideal utility} $v_i^*(\beta)$ of agent $i$ is the value of the following maximization problem over measurable $\rho:[0,\infty)\to[0,1]$.
\begin{equation} \label{eq:beta_ideal_utility_maximization_problem}
    \max_{\rho}\E_{V_i\sim\mathcal F_i}[V_i\rho(V_i)] \quad\text{subject to}\quad \E_{V_i\sim\mathcal F_i}[\rho(V_i)]\leq\beta
\end{equation}
\end{definition}
This is the ideal utility of agent $i$ if they were to have fair share $\beta$. The ideal utility as defined in \cref{def:ideal_utility} of agent $i$ is just the $\alpha_i$-ideal utility. \citet{fikioris2025beyond} prove that the $\beta$-ideal utility is concave in $\beta$. This is intuitive: if $\beta_1 < \beta_2 < \beta_3$, then $\frac{v_i^*(\beta_2) - v_i^*(\beta_1)}{\beta_2-\beta_1} \geq \frac{v_i^*(\beta_3) - v_i^*(\beta_2)}{\beta_3 - \beta_2}$ since the gain in utility between the top $\beta_1$ and $\beta_2$ quantiles of the value distribution is greater than the gain in utility between the top $\beta_2$ and $\beta_3$ quantiles.
\begin{fact} \label{fact:beta_ideal_utility_concavity}
The function $\beta\mapsto v_i^*(\beta)$ is concave in $\beta$.
\end{fact}

The next lemma shows that agents cannot achieve much higher utility than guaranteed in \cref{prop:equilibrium_utility_claim} by deviating.
\begin{lemma} \label{lem:equilibrium_deviating_utility_upper_bound}
Fix an agent $i$. Assume all other agents $j\neq i$ are using a \oneOverNAggressiveStrategy. Then, regardless of the strategy of agent $i$,
\begin{equation*}
    \frac1T\sum_{t=1}^T \mathbb E[U_i[t]] \leq \left(1-\left(1-\frac1n\right)^n\right)v_i^* + O\left(\sqrt{\frac{\log T}{T}}\right).
\end{equation*}
\end{lemma}
\begin{proof}
Let $\tau = \min_{j\neq i}\min\{\tau_j,\tau_j'\}$. Let $\mathcal G_t$ be the $\sigma$-algebra generated by the history $\mathcal H_t$ and the next agent $i$ bid $r_i[t+1]$. At each time $t\leq\tau$,
\begin{align*}
    \E[P_i[t]\mid\mathcal G_{t-1}] & = \paymentConstant\E\left[\frac{r_i[t]}{1 + r_i[t] + \sum_{j\neq i}r_j[t]}\,\middle|\,\mathcal G_{t-1}\right] & \text{[by \cref{lem:expected_payment_of_agent}]}\\
     & = \paymentConstant r_i[t]\cdot \E_{X\sim\Binom(n-1,1/n)}\left[\frac{1}{2+X}\right] & \text{[by the other agents' strategies]}\\
     & = r_i[t]. & \text{[by \cref{prop:expectation_of_functions_of_binom}]}
\end{align*}
and
\begin{align*}
    \E[U_i[t]\mid\mathcal G_{t-1}] & = \E\left[V_i[t]\cdot\frac{r_i[t]}{r_i[t] + \sum_{j\neq i}r_j[t]}\,\middle|\,\mathcal G_{t-1}\right] & \text{[by the random allocation rule]}\\
     & = \E_{X\sim\Binom(n-1,1/n)}\left[\frac{1}{1+X}\right]\cdot \E[V_i[t]r_i[t]] & \text{[by the strategy of agents $j\neq i$]}\\
     & = \left(1-\left(1-\frac1n\right)^n\right)\E[V_i[t]r_i[t]]. & \text{[by \cref{prop:expectation_of_functions_of_binom}]}
\end{align*}
By Azuma-Hoeffding applied to the $\mathcal G_t$-martingale $\sum_{s=1}^{\min\{t,\tau\}}P_i[t] - \sum_{s=1}^{\min\{t,\tau\}} r_i[s]$ with increments bounded by $\paymentConstant$,
\begin{equation}
\label{eq:deviating_agent_payment}
    \sum_{t=1}^\tau P_i[t] \geq \sum_{t=1}^\tau r_i[t] - \paymentConstant\sqrt{T\ln T}
\end{equation}
with probability at least $1-1/T^2$ and to the $\mathcal G_t$-martingale $\sum_{s=1}^{\min\{t,\tau\}}U_i[t] - \min\{t,\tau\}$ with increments bounded by $\bar v$ where $\bar v$ bounds the value distribution $\mathcal F_i$,
\begin{equation}
\label{eq:deviating_agent_utility}
    \sum_{t=1}^\tau U_i[t] \leq \sum_{t=1}^\tau\left(1-\left(1-\frac1n\right)^n\right)\E[V_i[t]r_i[t]] + \bar v\sqrt{T\ln T}
\end{equation}
with probability at least $1-1/T^2$. Consider what happens on the event that \eqref{eq:deviating_agent_payment}, \eqref{eq:deviating_agent_utility}, and the event that $\tau \geq T - O\left(\sqrt{T\log T}\right)$, which has probability at least $1-O(1/T^2)$ by \cref{lem:robustness_stopping_time_lower_bound,lem:equilibrium_stopping_time_lower_bound}. Call this event $E$, which has probability at least $1-O(1/T^2)$ by the union bound.

Using \eqref{eq:deviating_agent_payment} and agent $i$'s budget constraint, we find that
\begin{equation*}
    \sum_{t=1}^\tau r_i[t] \leq \frac{T}{n} + \paymentConstant\sqrt{T\ln T}.
\end{equation*}
It follows from the fact that $\tau\geq T-O\left(\sqrt{T\log T}\right)$ that 
\begin{align}
    \sum_{t=1}^T r_i[t] \leq \frac Tn + O\left(\sqrt{T\log T}\right). \label{eq:deviating_agent_total_requests}
\end{align}
Also, by the fact that $\tau\geq T-O\left(\sqrt{T\log T}\right)$ and \eqref{eq:deviating_agent_utility},
\begin{align}
    \sum_{t=1}^T U_i[t] \leq \left(1-\left(1-\frac1n\right)^n\right)\sum_{t=1}^T\E[V_i[t]r_i[t]] + O\left(\sqrt{T\log T}\right). \label{eq:deviating_agent_total_utility}
\end{align}
Let $\rho_i[t]:[0,\infty)\to[0,1]$ be a measurable function such that $\rho_i[t](V_i[t]) = \E[r_i[t]\mid V_i[t]]$. Let
\begin{equation*}
    \rho_i(v_i) = \frac1T\sum_{t=1}^T \rho_i[t](v_i)
\end{equation*}

Taking expectations, using \eqref{eq:deviating_agent_total_requests} on $E$ and noting that $r_i[t]\leq 1$ on the complement event $E^c$, we obtain
\begin{equation*}
\begin{split}
    \E_{V_i\sim\mathcal F_i}[\rho_i(V_i)] = \frac1T\sum_{t=1}^T \E_{V_i\sim\mathcal F_i}[\rho_i[t](V_i)] = \frac1T\sum_{t=1}^T \E[r_i[t]] \leq \Pr(E^c) + \frac1T\sum_{t=1}^T \E[r_i[t]\pmb1_E]\\
     \leq \Pr(E^c) + \left(\frac1n + O\left(\sqrt{\frac{\log T}{T}}\right)\right)\Pr(E) \leq \frac1n + O\left(\sqrt{\frac{\log T}{T}}\right).
\end{split}
\end{equation*}
Thus, $\rho_i$ is a feasible solution to \eqref{eq:beta_ideal_utility_maximization_problem} with $\beta = \frac1n + O\left(\sqrt{\frac{\log T}{T}}\right)$. Then,
\begin{equation*}
    \E_{V_i\sim\mathcal F_i}[V_i\rho_i(V_i)] \leq v_i^*\left(\frac1n + O\left(\sqrt{\frac{\log T}{T}}\right)\right)\leq v_i^*\left(\frac1n\right) + O\left(\sqrt{\frac{\log T}{T}}\right)
\end{equation*}
where the inequality follows from the concavity of $\beta\mapsto v_i^*(\beta)$. Using the above and \eqref{eq:deviating_agent_total_utility}, we compute
\begin{equation*}
\begin{split}
    \frac1T\sum_{t=1}^T \E[U_i[t]] & \leq \frac1T\sum_{t=1}^T \E[U_i[t]\pmb1_E] + \frac1T\sum_{t=1}^T \E[U_i[t]\pmb1_{E^c}]\\
    & \leq \left(\frac1T\left(1-\left(1-\frac1n\right)^n\right)\sum_{t=1}^T\E[V_i[t]r_i[t]] + O\left(\sqrt{\frac{\log T}{T}}\right)\right)\Pr(E) + \E_{V_i\sim\mathcal F_i}[V_i\pmb1_{E^c}]\\
    & \leq \frac1T\left(1-\left(1-\frac1n\right)^n\right)\sum_{t=1}^T\E[V_i[t]r_i[t]] + O\left(\sqrt{\frac{\log T}{T}}\right)\\
    & = \left(1-\left(1-\frac1n\right)^n\right)\E_{V_i\sim\mathcal F_i}[V_i\rho_i(V_i)] + O\left(\sqrt{\frac{\log T}{T}}\right)\\
    & \leq v_i^*\left(\frac1n\right) + O\left(\sqrt{\frac{\log T}{T}}\right).
\end{split}
\end{equation*}
using the fact that the value distribution is bounded for the third inequality.

\end{proof}

% \dlcomment{I have found in my old notes that there is a way to show equilibrium by only assuming finite mean and not bounded value distribution but it's more tricky.}

We have essentially completed the proof of \cref{thm:equilibrium_mechanism}. \cref{prop:robustness_in_equilibrium_mechanism} gives the robustness claim. \cref{prop:equilibrium_utility_claim} gives the utility at equilibrium lower bound, and \cref{lem:equilibrium_deviating_utility_upper_bound} shows that each agent can gain at most an additive $O\left(\sqrt{\frac{\log T}{T}}\right)$ of utility from deviating.
\section{Deferred Proofs from Section \ref{sec:asymmetric_fair_shares}} \label{sec:app:asymmetric_fair_shares_proofs}

In this section, we prove \cref{thm:asymmetric_fair_shares}, which we restate below for convenience.

\asymmetricFairSharesTheorem*

The proof goes similarly to the proof of \cref{thm:equilibrium_mechanism} given in \cref{sec:app:equilibrium_mechanism_proofs}. We shall use hat notation to denote quantities in the simulated mechanism $\hat{\mathcal M}$; we let $\hat B_{(i,i')}[t]$, $\hat P_{(i,i')}[t]$, $\hat r_{(i,i')}[t]$, $\hat W_{(i,i')}[t]$ be the budgets, payments, request indicators, and win indicators of the simulated agent $(i,i')$ in $\hat{\mathcal M}$, respectively.

Throughout this section, we assume we use the parameters
\begin{equation*}
    \paymentConstant = \frac{m+1}{1+m(1-1/m)^{m+1}}
\end{equation*}
and
\begin{equation*}
    \epsilon = \sqrt{T\ln T}.
\end{equation*}

Let $\hat\tau_{(i,i')}$ be the time at which simulated agent $(i,i')$ runs out of budget in $\hat{\mathcal M}$, and $T$ if this never happens:
\begin{equation*}
    \tau_{(i,i')} = \begin{cases}\min\left\{t:\sum_{s=1}^t \hat P_{(i,i')}[s] \geq T/m\right\} & \text{if $\sum_{s=1}^T \hat P_{(i,i')}[s] \geq T/m$}\\T & \text{otherwise}\end{cases}.
\end{equation*}
Let $\tau_i = \min_{i'}\tau_{(i,i')}$.

Define
\begin{equation*}
    \tau_i' = \begin{cases}\min\left\{t:\sum_{s=1}^t r_i[s] \leq \left(\asymmetricRate{i}\right)t - \epsilon + 1\right\} & \text{if $\exists t$: $\sum_{s=1}^t r_i[s] \leq \left(\asymmetricRate{i}\right)t - \epsilon + 1$}\\T & \text{otherwise}\end{cases}.
\end{equation*}
We can see that with this definition of $\tau_i'$, if $t\leq \tau_i'$, then the minimum bidding constraint does not affect player $i$ at time $t$, in that it guarantees that $\sum_{s=1}^t r_i[s] > \left(\asymmetricRate{i}\right)t - \epsilon\cdot t$. Indeed, if $t \leq \tau_i'$, then
\begin{equation*}
    \sum_{s=1}^t r_i[s] \geq \sum_{s=1}^{t-1}r_i[s] > \left(\asymmetricRate{i}\right)(t-1)-\epsilon + 1 \geq \left(\asymmetricRate{i}\right)-\epsilon.
\end{equation*}

Let $\mathcal H_t$ denote the history up to and including time $t$. Notice that the $\tau_i$ and $\tau_i'$ are stopping times with respect to the filtration $\mathcal H_t$.

The below lemma uses standard probability concentration bounds to show that the number of times that an agent can bid after $\tau_i$, the time at which at least one of the simulated agents $(i,i')$ runs out of budget, is sublinear in $T$. The lemma will help us reason about agent $i$'s bids only for $t\leq \tau_i$.
\begin{lemma} \label{lem:not_many_requests_after_running_out_of_budget}
For any fixed agent $i$, regardless of any agent's strategy,
\begin{equation} \label{eq:not_many_requests_after_running_out_of_budget}
    \sum_{t=\tau_i+1}^T r_i[t] \leq O\left(\sqrt{T\log T}\right).
\end{equation}
\end{lemma}
\begin{proof}
Let $\mathcal G_t$ denote the $\sigma$-algebra generated by $\mathcal H_t$ and the requests in the next time, $r_j[t+1]$ for all $j\in[n]$. For any $t\leq \tau_i$ and for any $i'$ and $i''$, $\E[\hat P_{(i,i')}[t]\mid\mathcal G_{t-1}] = \E[\hat P_{(i,i'')}[t]\mid\mathcal G_{t-1}]$ by symmetry. By the Azuma-Hoeffding inequality applied to the $\mathcal G_t$-martingale $\sum_{s=1}^{\min\{t,\tau\}} \left(\hat P_{(i,i')}[s] - \hat P_{(i,i'')}[s]\right)$ with increments bounded by $\paymentConstant$, for any fixed $i'$ and $i''$,
\begin{equation} \label{eq:simulated_agent_payments_not_too_far}
    \left|\sum_{t=1}^{\tau_i} \hat P_{(i,i')}[t] - \sum_{t=1}^{\tau_i} \hat P_{(i,i'')}[t]\right|\leq \paymentConstant\sqrt{T\ln T}
\end{equation}
with probability at least $1-2/T^2$. For any fixed $i'$ and $t$, $\E[\hat P_{(i,i')}[t]\mid\mathcal G_{t-1}] \geq c_{i'}r_i[t]$ for some $c_{i'}$ not depending on $t$ or $T$, because if $r_i[t] = 1$, then the expected payment of the simulated agent $(i,i')$ is positive by the mechanism. By the Azuma-Hoeffding inequality applied to the $\mathcal G_t$-submartingale $\sum_{s=\tau_i+1}^t\hat P_{(i,i')}[s] - c_{i'}\sum_{s=\tau_i+1}^tr_i[s]$ with increments bounded by $\paymentConstant$,
\begin{equation} \label{eq:requests_when_one_simualted_agent_no_budget_small}
    \sum_{t=\tau_i+1}^T\hat P_{(i,i')}[t] \geq c_{i'}\sum_{t=\tau_i+1}^Tr_i[t] - \paymentConstant\sqrt{T\ln T}
\end{equation}
with probability at least $1-1/T^2$.

Consider what happens on the event that \eqref{eq:simulated_agent_payments_not_too_far} happens for all $i\neq i'$ and \eqref{eq:requests_when_one_simualted_agent_no_budget_small} happens for all $i'$, which has probability at least $1 - O(1/T^2)$ by the union bound. Let $\underline i' = \argmin_{i'}\tau_{(i,i')}$ and $\bar i' = \argmax_{i'} \tau_{(i,i')}$. By definition of $\tau_i$, $\sum_{t=1}^{\tau_i}\hat P_{(i,\underline i')}[t] \geq T/m$, and by \eqref{eq:simulated_agent_payments_not_too_far},
\begin{equation*}
    \sum_{t=1}^{\tau_i}\hat P_{(i,\bar i')}[t] \geq \sum_{t=1}^{\tau_i} \hat P_{(i,\underline i')}[t] - \paymentConstant\sqrt{T\ln T}\geq \frac Tm - \paymentConstant\sqrt{T\ln T}.
\end{equation*}
By the budget constraint enforcement, $\sum_{t=1}^T \hat P_{(i, \bar i')}[t] \leq T/m + \paymentConstant$, so combining the above with \eqref{eq:requests_when_one_simualted_agent_no_budget_small}, we have
\begin{equation*}
\begin{split}
    \sum_{t=\tau_i+1}^T r_i[t] & \leq \frac{1}{c_{\bar i'}}\left(\sum_{t=\tau_i+1}^T \hat P_{(i,\bar i')}[t] + \paymentConstant\sqrt{T\ln T}\right)\\
    & = \frac{1}{c_{\bar i'}}\left(\sum_{t=1}^T \hat P_{(i,\bar i')}[t] - \sum_{t=1}^{\tau_i}\hat P_{(i,\bar i')}[t] + \paymentConstant\sqrt{T\ln T}\right)\\
    & \leq \frac{1}{c_{\bar i'}}\left(\paymentConstant + 2\paymentConstant\sqrt{T\ln T}\right) = O\left(\sqrt{T\log T}\right).
\end{split}
\end{equation*}
\end{proof}

The below lemma is analogous to \cref{lem:robustness_stopping_time_lower_bound}.
\begin{lemma} \label{lem:asymmetric_robustness_stopping_time_lower_bound}
If agent $i$ uses a \asymmetricStrategy{i}, regardless of the strategies of the other agents, $\tau_i' = T$ with probability at least $1-O(1/T^2)$.
\end{lemma}
\begin{proof}
Assume agent $i$ uses a \asymmetricStrategy{i}. For $t\leq \tau_i'$, $\E[r_i[t]\mid \mathcal H_t] = \asymmetricRate{i}$, so by the Azuma-Hoeffding inequality applied to the martingale $\sum_{s=1}^{\min\{t,\tau_i'\}}r_i[s] - \left(\asymmetricRate{i}\right)\min\{t,\tau_i'\}$,
\begin{equation*}
    \Pr\left(\sum_{t=1}^{\tau_i'}r_i[t] \leq \left(\asymmetricRate{i}\right)\tau_i' - \epsilon + 1 \right) \leq \exp\left(-\frac{2(\epsilon-1)^2}{T}\right) \leq O\left(\frac{1}{T^2}\right).
\end{equation*}
If the above event does not happen, then by the definition of $\tau_i'$, it must be that $\tau_i'=T$, establishing the first claim of the lemma.  
\end{proof}

The below lemma shows that if agent $i$ is using a $\asymmetricRate{i}$ strategy, then the simulated agents are $(i,i')$ bidding i.i.d. with probability $1/m$ (subject to the minimum bidding constraint and budget constraint).
\begin{lemma} \label{lem:simulated_distribution}
Fix an agent $i$, and suppose
\begin{equation*}
    \E[r_i[t]\mid\mathcal H_{t-1}] = \asymmetricRate{i}.
\end{equation*}
Then, conditioned on $\mathcal H_{t-1}$, the $\hat r_{(i,i')}[t]$ are independent across $i'$ with expectation
\begin{equation*}
    \E[\hat r_{(i,i')}[t]\mid\mathcal H_{t-1}] = \frac1m.
\end{equation*}
\end{lemma}
\begin{proof}

Let $\hat Z_1, \dots, \hat Z_{k_i}$ be i.i.d. $\Bern(1/m)$ random variables. Let $Z = \hat Z_1 + \dots + \hat Z_{k_i}$. The distribution of $(\hat r_{(i,i')})_{i'=1}^{k_i}$ conditioned on $\mathcal H_{t-1}$ and $r_i[t]=1$ is $\mathcal D_{k_i,m}$, which is the distribution of $(\hat Z_{i'})_{i'=1}^{k_i}$ conditioned on $Z\geq 1$ by definition of $\mathcal D_{k_i,m}$.
% Thus, the distribution of $\hat r_{(i,i')}$ conditioned on $\mathcal H_{t-1}$ and $r_i[t]=1$ is the same as the distribution of $\hat Z_{i'}$ conditioned on $Z\geq 1$.
% For a fixed agent $(i,i')$,
% \begin{align*}
%     \E[\hat r_{(i,i')}[t]\mid\mathcal H_{t-1}] = \E[\hat r_{(i,i')}[t]r_i[t]\mid\mathcal H_{t-1}] = \E[\hat r_{(i,i')}[t]\mid\mathcal H_{t-1}, r_i[t]=1]\cdot\E[r_i[t]\mid \mathcal H_{t-1}]\\
%     = \E[\hat Z_{i'}\mid Z\geq 1]\cdot\left(\asymmetricRate{i}\right) = \frac{\E[\hat Z_i]}{\Pr(\hat Z_{i'}\geq 1)}\cdot \left(\asymmetricRate{i}\right) = \E[\hat Z_i] = \frac1m.
% \end{align*}

Consider any $\emptyset\subsetneq I'\subseteq\{1,2,\dots, k_i\}$. Then,
\begin{align*}
    \E\left[\prod_{i'\in I'}\hat r_{(i,i')}[t]\,\middle|\,\mathcal H_{t-1}\right]& = \E\left[\prod_{i'\in I'}\hat r_{(i,i')}[t]r_i[t]\,\middle|\,\mathcal H_{t-1}\right] & \text{[$r_i[t]=0\implies \hat r_{(i,i')}[t]=0$]}\\
    & = \E\left[\prod_{i'\in I'}\hat r_{(i,i')}[t]\,\middle|\,\mathcal H_{t-1}, r_i[t]=1\right]\E[r_i[t]\mid \mathcal H_{t-1}]\\
    & = \E\left[\prod_{i'\in I'}\hat Z_{i'}\,\middle|\,Z\geq 1\right]\E[r_i[t]\mid \mathcal H_{t-1}] & \text{[by the previous paragraph]}\\
    & = \frac{\E\left[\prod_{i'\in I'}\hat Z_{i'}\right]}{\Pr(Z\geq 1)}\cdot \E[r_i[t]\mid \mathcal H_{t-1}] & \text{[$Z=0\implies \hat Z_{i'}=0$]}\\
    & = \E\left[\prod_{i'\in I'}\hat Z_{i'}\right] = \left(\frac1m\right)^{|I'|}.
\end{align*}
To get the last line, we use the fact that $\Pr(Z\geq 1) = \asymmetricRate{i}$, which is equal to $\E[r_i[t]\mid\mathcal H_{t-1}]$ by assumption. This shows that the $\hat r_{(i,i')}[t]$ are i.i.d. $\Bern(1/m)$ random variables when conditioned on $\mathcal H_{t-1}$, as desired.
\end{proof}

The below lemma is analogous to \cref{lem:equilibrium_stopping_time_lower_bound}.
\begin{lemma} \label{lem:asymmetric_equilibrium_stopping_time_lower_bound}

If all agents $j\neq i$ use a \asymmetricStrategy{j}, regardless of the strategy of agent $i$, $\min_{j\neq i}\tau_j \geq T - O\left(\sqrt{T\log T}\right)$ with probability at least $1-O(1/T^2)$.
\end{lemma}
\begin{proof}
Assume agents $j\neq i$ use a \asymmetricStrategy{j}{} and that agent $i$'s strategy is arbitrary. Let $\mathcal G_t$ be the $\sigma$-algebra generated by $\mathcal H_t$ and $r_i[t+1]$. Let $\tau = \min_{j\neq i}\tau_i$ and $\tau'=\min_{j\neq i}\tau_j'$. We must lower bound $\tau$. For any $t\leq \min\{\tau_i,\tau,\tau'\}$, and for any $(j,j')\in\hat N$ with $j\neq i$,
\begin{equation} \label{eq:expected_payment_simulated_agent_equilibrium}
\begin{split}
    \E[\hat P_{(j,j')}\mid \mathcal G_{t-1}] & = \paymentConstant\E\left[\frac{\hat r_{(j,j')}[t]}{1  + \hat r_{(j,j')}[t]+ \sum_{i'=1}^{k_i}\hat r_{(i,i')}[t] + \sum_{j''\neq j'}\hat r_{(j,j'')}[t] + \sum_{k\neq i,j}^{k_k}\hat r_{(k,k')}[t]}\right]\\
    & = \frac{\paymentConstant}{m}\cdot\begin{cases}\E\left[\frac{1}{2 + \hat X + \hat Y}\right] & \text{if $r_i[t]=1$}\\\E\left[\frac{1}{2 + \hat Y}\right] & \text{if $r_i[t]=0$}\end{cases}.
\end{split}
\end{equation}
where $\hat X \overset{d}= \sum_{i'=1}^{i_k}\hat V_{(i,i')}$ where $(\hat V_{(i,i')})\sim\mathcal D_{k_i,m}$ and $\hat Y\sim\Binom(m-k_i-1,1/m)$ are independent. The first line follows directly from \cref{lem:expected_payment_of_agent}. We now argue the second line. The agent $(j,j')$ bids with probability $1/m$ by \cref{lem:simulated_distribution}. Also by \cref{lem:simulated_distribution}, $\hat Y$ is equal in distribution to $\sum_{j''\neq j'}\hat r_{(j,j'')}[t] + \sum_{k\neq i,j}^{k_k}\hat r_{(k,k')}[t]$ conditioned on $t\leq\min\{\tau,\tau'\}$. By the mechanism, $\hat X$ simulated agents $(i,i')$ bid if $r_i[t]=1$ and no simulated agents $(i,i')$ bid if $r_i[t]=0$. Combining the previous three sentences, we obtain the second equality in \eqref{eq:expected_payment_simulated_agent_equilibrium}. 

\cref{eq:expected_payment_simulated_agent_equilibrium} holds for $t\leq \min\{\tau_i,\tau,\tau'\}$. For any $t\leq \min\{\tau,\tau'\}$, we still have the bottom case of \eqref{eq:expected_payment_simulated_agent_equilibrium} by the same argument, in that we still have
\begin{equation} \label{eq:expected_payment_simulated_agent_equilibrium_specialized}
    \E[\hat P_{(j,j')}\mid \mathcal G_{t-1}] = \frac\paymentConstant m\E\left[\frac{1}{2 + \hat Y}\right] \quad \text{if $r_i[t]=0$}.
\end{equation}

Notice that
\begin{equation*}
    \sum_{s=1}^{\min\{t,\tau_i,\tau,\tau'\}}\hat P_{(j,j')}[s] - \sum_{s=1}^{\min\{t,\tau_i,\tau,\tau'\}}\E[\hat P_{(j,j')}[s]\mid \mathcal G_{t-1}]
\end{equation*}
is a $\mathcal G_t$-martingale with increments bounded by $\paymentConstant$, so by Azuma-Hoeffding,
\begin{equation*}
    \Pr\left(\sum_{s=1}^{\min\{\tau_i,\tau,\tau'\}} \hat P_{(j,j')}[s] \geq \sum_{s=1}^{\min\{\tau_i,\tau,\tau'\}}\E[\hat P_{(j,j')}[s]\mid\mathcal G_{s-1}] + \paymentConstant\sqrt{T\ln T}\right) \leq \frac{1}{T^2}.
\end{equation*}
From \cref{lem:asymmetric_robustness_stopping_time_lower_bound}, $\tau' = T$ with probability at least $1-O(1/T^2)$. Consider what happens on the event that $\tau'=T$, \eqref{eq:not_many_requests_after_running_out_of_budget} as in \cref{lem:not_many_requests_after_running_out_of_budget} happens for agent $i$, and the above display does not happen for any $(j,j')$ with $j\neq i$, which has probability at least $1-O(1/T^2)$ by the union bound. In that case,
\begin{equation*}
    \sum_{s=1}^{\min\{\tau_i,\tau\}} \hat P_{(j,j')}[s]\leq \sum_{s=1}^{\min\{\tau_i,\tau\}} \E[\hat P_{(j,j')}[s]\mid\mathcal G_{s-1}] + \paymentConstant\sqrt{T\ln T}
\end{equation*}
for every $(j,j')$ with $j\neq i$. If $\tau=T$, then there is nothing to prove, so assume $\tau < T$. Then, $\sum_{s=1}^\tau \hat P_{(j,j')}[s] \geq T/m$ for some $(j,j')$ with $j\neq i$, so
\begin{equation}
\label{eq:asymmetric_nondeviating_agents_stopping_time}
    \frac Tm \leq \sum_{s=1}^{\min\{\tau_i, \tau\}} \E[\hat P_{(j,j')}[s]\mid\mathcal G_{s-1}] + \paymentConstant\sqrt{T\ln T}.
\end{equation}
Write
\begin{equation*}
\begin{split}
    \sum_{s=1}^\tau \E[\hat P_{(j,j')}[s]\mid\mathcal G_{s-1}] & = \sum_{s=1}^{\min\{\tau_i,\tau\}}\E[\hat P_{(j,j')}[s]\mid \mathcal G_{s-1}]r_i[s] + \sum_{s=\tau_i+1}^\tau\E[\hat P_{(j,j')}\mid \mathcal G_{s-1}]r_i[s]\\
    & \quad + \sum_{s=1}^{\tau}\E[\hat P_{(j,j')}[s]\mid \mathcal G_{s-1}](1-r_i[s])\\
    & \leq \sum_{s=1}^{\min\{\tau_i,\tau\}}\E[\hat P_{(j,j')}\mid \mathcal G_{s-1}]r_i[s] + \sum_{s=1}^{\tau}\E[\hat P_{(j,j')}\mid \mathcal G_{s-1}](1-r_i[s]) + O\left(\sqrt{T\log T}\right)\\
    & \leq \frac\paymentConstant m\left(\E\left[\frac{1}{2 + \hat X + \hat Y}\right]\#\{s\leq \tau :r_i[s]=1\} + \E\left[\frac{1}{2 + \hat Y}\right]\#\{s\leq\tau: r_i[s]=0\}\right)\\
    & \quad + O\left(\sqrt{T\log T}\right)\\
\end{split}
\end{equation*}
where the first inequality follows from the fact that the payments are bounded and the fact that from \eqref{eq:not_many_requests_after_running_out_of_budget}, $\sum_{t=\tau_i+1}^{\tau}r_i[s]\leq O\left(\sqrt{T\log T}\right)$. The second inequality follows from \eqref{eq:expected_payment_simulated_agent_equilibrium} and \eqref{eq:expected_payment_simulated_agent_equilibrium_specialized}.

By the minimum bidding rule of the mechanism, $\#\{s\leq\tau:r_i[s]=1\}\geq \left(\asymmetricRate{i}\right)\tau - \epsilon$. It follows that
\begin{equation} \label{eq:asymmetric_payment_bound_equilibrium}
\begin{split}
    \sum_{s=1}^\tau \E[\hat P_{(j,j')}[s]\mid\mathcal G_{s-1}] & \leq \frac\paymentConstant m\left(\E\left[\frac{1}{2 + \hat X + \hat Y}\right]\left(\left(\asymmetricRate{i}\right)\tau - \epsilon\right)\right.\\
    & \quad \left.+ \E\left[\frac{1}{2 + \hat Y}\right]\left(\tau - \left(\left(\asymmetricRate{i}\right)\tau - \epsilon\right)\right) \right)\\
    & \leq \frac\paymentConstant m\left(\E\left[\frac{1}{2+\hat X+\hat Y}\right]\left(\asymmetricRate{i}\right) + \E\left[\frac{1}{2+\hat Y}\right]\left(1-\frac1m\right)^{k_i}\right)\tau\\
    & \quad + O\left(\sqrt{T\log T}\right).
\end{split}
\end{equation}
We shall now argue that
\begin{equation}
\label{eq:asymmetric_payment_at_equilibrium}
    \paymentConstant\left(\E\left[\frac{1}{2+\hat X+\hat Y}\right]\left(\asymmetricRate{i}\right) + \E\left[\frac{1}{2+\hat Y}\right]\left(1-\frac1m\right)^{k_i}\right) =1.
\end{equation}
Let $Z = \sum_{k=1}^{m-1}Z_k$ where each $Z_k\sim \Bern(1/m)$ are independent. Then, one can see that $\hat Y \overset{d}= \sum_{k=k_i+1}^{m-1} Z_k$, $\hat X+\hat Y$ is equal in distribution to $Z$ conditioned on $\sum_{k=1}^{k_i}Z_k\geq 1$, $\hat Y$ is equal in distribution to $Z$ conditioned on $\sum_{k=1}^{k_i}Z_k=0$, $\Pr(\sum_{k=1}^{k_i}Z_k\geq 1) =  1-(1-1/m)^{k_i}$, and $\Pr(\sum_{k=1}^{k_i}Z_k = 0) = (1-1/m)^{k_i}$. It follows that 
\begin{equation} \label{eq:asymmetric_probabilistic_payment_at_equilibrium}
\begin{split}
     \E\left[\frac{1}{2+Z}\right] = \E\left[\frac{1}{2+Z}\,\middle|\,\sum_{k=1}^{k_i}Z_k\geq 1\right]\Pr\left(\sum_{k=1}^{k_i}Z_k\geq 1\right) + \E\left[\frac{1}{2+Z}\,\middle|\,\sum_{k=1}^{k_i}Z_k = 0\right]\Pr\left(\sum_{k=1}^{k_i}Z_k = 0\right)\\
     = \E\left[\frac{1}{2+\hat X+\hat Y}\right]\left(\asymmetricRate{i}\right) + \E\left[\frac{1}{2+\hat Y}\right]\left(1-\frac1m\right)^{k_i}.
\end{split}
\end{equation}
By \cref{prop:expectation_of_functions_of_binom} on $Z\sim\Binom(m-1,1/m)$ and the value of $\paymentConstant$, $\E[1/(2+Z)]$ is exactly $1/\paymentConstant$, so \eqref{eq:asymmetric_payment_at_equilibrium} holds. 

Substituting \eqref{eq:asymmetric_payment_at_equilibrium} into \eqref{eq:asymmetric_payment_bound_equilibrium}, we obtain
\begin{equation*}
    \sum_{s=1}^\tau \E[\hat P_{(j,j')}[s]\mid\mathcal G_{s-1}] \leq \frac{\tau}{m} + O\left(\sqrt{T\log T}\right).
\end{equation*}
Then, substituting into \eqref{eq:asymmetric_nondeviating_agents_stopping_time}, and solving for $\tau$, we obtain that
\begin{equation*}
    \tau  \geq T - O\left(\sqrt{T\log T}\right).
\end{equation*}

\end{proof}

Recall that $v_i^*(\beta)$ denotes agent $i$'s $\beta$-ideal utility, as in \cref{def:beta_ideal_utility}.

The below lemma is analogous to \cref{lem:bernoulli_reduction} and shows that it suffices to argue about the number of wins of the simulated agents to obtain utility lower bounds.
\begin{lemma} \label{lem:asymmetric_bernoulli_reduction}
Fix an agent $i$ and the strategies of other agents $j\neq i$. Assume agent $i$ is using a \asymmetricStrategy{i}. Suppose
\begin{equation}
\label{eq:utility_of_simulated_agents_assumption}
    \Pr\left(\frac1T\sum_{t=1}^T \hat W_{(i,i')}[t] \geq \frac{\lambda_i}{m}\quad\forall i'\right) \geq 1-\delta.
\end{equation}
Then, with probability at least $1-\delta-O(1/T^2)$,
\begin{equation*}
    \frac1T\sum_{t=1}^T U_i[t] \geq \lambda_i\cdot\frac{k_i}{m}\cdot\frac{v_i^*\left(\asymmetricRate{i}\right)}{\asymmetricRate{i}} - O\left(\sqrt{\frac{\log T}{T}}\right).
\end{equation*}
\end{lemma}
\begin{proof}
By the \asymmetricStrategy{i}{} and the definition of $\beta$-ideal utility,
\begin{equation*}
    \E[V_i[t] \mid r_i[t]=1] = \frac{v_i^*\left(\asymmetricRate{i}\right)}{\asymmetricRate{i}}.
\end{equation*}
% where the second equality follows from the definition of the $\left(\asymmetricRate{i}\right)$-ideal utility and the strategy of agent $i$.
Conditioned on $r_i[t]$, the values $V_i[t]$ are independent across time and independent of the wins of the simulated agents $(\hat W_{(i,i')}[s])_{i',s}$. Conditioned on the $(\hat W_{(i,i')}[s])_{i',s}$, the utilities $U_i[t] = V_i[t]W_i[t] = V_i[t]\sum_{i'}\hat W_{(i,i')}[s]$ are i.i.d. with mean
\begin{equation*}
    \E[V_i[t]W_i[t] \mid (\hat W_{(i,i')}[s])_{i',s}] = \E[V_i[t]\mid r_i[t]=1]W_i[t] = \frac{v_i^*\left(\asymmetricRate{i}\right)}{\asymmetricRate{i}}\sum_{i'=1}^{k_i}\hat W_{(i,i')}[t]
\end{equation*}
by the strategy of agent $i$ and the definition of $\beta$-ideal utility (recall that agent $i$ is using a $\beta = \asymmetricRate{i}$ strategy, a strategy that requests with probability $\rho(V_i)$ where $\rho$ solves \eqref{eq:beta_ideal_utility_maximization_problem} as in the definition of $\beta$-ideal utility).
By Hoeffding's inequality, letting $\bar v$ upper bound the value distribution $\mathcal F_i$,
\begin{equation*}
    \Pr\left(\sum_{t=1}^T U_i[t] \leq \frac{v_i^*\left(\asymmetricRate{i}\right)}{\asymmetricRate{i}}\sum_{t=1}^T\sum_{i'=1}^{k_i}\hat W_{(i,i')}[t] - \bar v\sqrt{T\ln T}\right) \leq \frac{1}{T^2}.
\end{equation*}
Now, the result follows from the above, the assumption \eqref{eq:utility_of_simulated_agents_assumption}, and the union bound.
\end{proof}

We convert the guarantee in terms of the $\left(\asymmetricRate{i}\right)$-ideal utility in  \cref{lem:asymmetric_bernoulli_reduction} to a guarantee in terms of the $\alpha_i$-ideal utility in the below lemma.
\begin{lemma} \label{lem:asymmetric_concavity}
We have
\begin{equation*}
    \frac{k_i}{m}\cdot \frac{v_i^*\left(\asymmetricRate{i}\right)}{\asymmetricRate{i}} \geq v_i^*\left(\frac{k_i}{m}\right) = v_i^*(\alpha_i) = v_i^*.
\end{equation*}
\end{lemma}
\begin{proof}
This follows from the fact that $1-(1-1/m)^{k_i} \leq k_i/m = \alpha_i$ and the concavity of $\beta\mapsto v_i^*(\beta)$ (\cref{fact:beta_ideal_utility_concavity}).
\end{proof}

Now we give our robustness claim.
\begin{proposition} \label{prop:asymmetric_robustness}
A \asymmetricStrategy{i}{} is $\lambda_i$-robust for some
\begin{equation*}
    \lambda_i \geq \frac{5}{3e} - O\left(\sqrt{\frac{\log T}{T}}\right).
\end{equation*}
\end{proposition}
\begin{proof}
Assume agent $i$ uses a \asymmetricStrategy{i}. Then, each simulated agent $(i,i')$ is using a $1/m$-aggressive strategy in $\hat{\mathcal M}$ with values $\hat V_{(i,i')}$ if the minimum bidding constraint were not enforced in that the values $\hat V_{(i,i')}$ would be i.i.d. across time from the distribution $\hat{\mathcal F}_{(i,i')} = \Bern(1/m)$ by \cref{lem:simulated_distribution} and each simulated agent $(i,i')$ bids if and only if $\hat V_{(i,i')}=1$, which is a $1/m$-aggressive strategy in $\hat{\mathcal M}$.

Fix the strategies of players $j\neq i$. Couple $\hat{\mathcal M}$ and \allPayMechanism{} run with the simulated agents $\hat N$ such that each simulated agent $(i,i')$ is using a $1/m$-aggressive strategy in both mechanisms, for $t\leq \tau_i'$, agent $i$ requests in $\hat{\mathcal M}$ if and only if she requests in \allPayMechanism, and agents $(j,j')$ for $j\neq i$ request in $\hat{\mathcal M}$ if and only if they request in \allPayMechanism. This is possible to do because \allPayMechanismWithasymmetricFairShares{} runs the same as \allPayMechanism, so any strategy used by agents $j\neq i$ in $\hat{\mathcal M}$  can be used in \allPayMechanism. Also, agent $i$ can use the same strategy for $t\leq\tau_i'$ because for $t\leq\tau_i'$, the bidding minimum does not affect agent $i$'s ability to request. By \cref{lem:asymmetric_robustness_stopping_time_lower_bound}, $\tau_i'=T$ with probability at least $1-O(1/T^2)$. In this case, as argued before, since the minimum bidding rule is never enforced, each simulated agent $(i,i')$ is using a $1/m$-aggressive strategy in $\hat M$. It follows from \cref{thm:general_robustness_claim} that
\begin{equation*}
    \frac1T\sum_{t=1}^T\hat W_{(i,i')}[t] \geq \frac{\lambda_i}{m} - \left(\sqrt{\frac{\log T}{T}}\right).
\end{equation*}
for every $(i,i')$ for
\begin{equation*}
    \lambda_i \geq \min\left\{1 - \frac{3(1-1/m)}{3\paymentConstant - \paymentConstant/m}, \frac{5-1/m}{\paymentConstant(3-1/m)}\right\}  - \left(\sqrt{\frac{\log T}{T}}\right).
\end{equation*}
with probability at least $1-O(1/T^2)$.
By \cref{lem:asymmetric_bernoulli_reduction,lem:asymmetric_concavity}, a \asymmetricStrategy{i}{} is $\lambda_i$-robust in \allPayMechanismWithasymmetricFairShares. By \cref{lem:robustness_with_equilibrium_payment_constant}, $\lambda_i \geq 5/(3e)$.
\end{proof}

\begin{proposition} \label{prop:asymmetric_equilibrium_utility_claim}
If each agent $i$ plays a \asymmetricStrategy{i}, each agent is guaranteed
\begin{equation*}
    \frac1T\sum_{t=1}^T U_i[t] \geq \left(1-\left(1-\frac1m\right)^m\right)\frac{k_i}{m}\cdot\frac{v_i^*\left(\asymmetricRate{i}\right)}{\asymmetricRate{i}} - O\left(\sqrt{\frac{\log T}{T}}\right).
\end{equation*}

with probability at least $1-O(1/T^2)$.
\end{proposition}
\begin{proof}
Let $\tau = \min_{j\in[n]}\min\{\tau_j,\tau_j'\}$. By \cref{lem:robustness_stopping_time_lower_bound,lem:asymmetric_equilibrium_stopping_time_lower_bound} and the union bound, there is an event of probability at least $1-O(1/T^2)$ on which $\tau \geq T - O\left(\sqrt{T\log T}\right)$. 

Fix a simulated agent $(i,i')$.
For $t\leq \tau$,
\begin{align*}
    \E[\hat W_{(i,i')}[t]\mid\mathcal H_{t-1}] & = \E\left[\frac{\hat r_{(i,i')}[t]}{\hat r_{(i,i')}[t] + \sum_{\hat j\in \hat N\setminus\{(i,i')\}}r_{\hat j}[t]}\,\middle|\,\mathcal H_{t-1}\right] & \text{[Random allocation among bidding agents]}\\
    & = \frac1m\cdot\E_{X\sim\Binom(m-1,1/m)}\left[\frac1{1+X}\right] & \text{[\cref{lem:simulated_distribution}: agents bidding i.i.d. $\Bern(1/m)$]}\\
    & = \frac1m\left(1-\left(1-\frac1m\right)^m\right). & \text{[\cref{prop:expectation_of_functions_of_binom}]}
\end{align*}
Then, $\sum_{s=1}^{\min\{t,\tau\}}\hat U_{(i,i')}[s] - \frac{\min\{t,\tau\}}n\left(1-\left(1-\frac1m\right)^m\right)$ is an $\mathcal H_t$-martingale, so we can apply Azuma-Hoeffding to see that
\begin{equation*}
    \sum_{t=1}^\tau \hat W_{(i,i')}[t] \geq \frac{\tau}{m}\left(1-\left(1-\frac1m\right)^m\right) - \sqrt{T\ln T}
\end{equation*}
with probability at least $1- 1/T^2$. Using the fact that $\tau \geq T - O\left(\sqrt{T\log T}\right)$ with probability at least $1-O(1/T^2)$, we obtain
\begin{equation*}
    \frac1T\sum_{t=1}^\tau \hat W_{(i,i')}[t] \geq \frac{1}{m}\left(1-\left(1-\frac1m\right)^m\right) - O\left(\sqrt{\frac{\log T}{T}}\right)
\end{equation*}
with probability at least $1-O(1/T^2)$. This implies the result by \cref{lem:asymmetric_bernoulli_reduction}.

\end{proof}

\begin{lemma} \label{lem:asymmetric_equilibrium_deviating_utility_upper_bound}
Fix an agent $i$. Assume all other agents $j\neq i$ are using a \asymmetricStrategy{j}. Then, regardless of the strategy of agent $i$,
\begin{equation*}
    \frac1T\sum_{t=1}^T \mathbb E[U_i[t]] \leq \left(1-\left(1-\frac1m\right)^m\right)v_i^* + O\left(\sqrt{\frac{\log T}{T}}\right).
\end{equation*}
\end{lemma}
\begin{proof}
Let $\tau = \min_{j\neq i}\min\{\tau_j,\tau_j'\}$. Let $\mathcal G_t$ be the $\sigma$-algebra generated by the history $\mathcal H_t$ and the next agent $i$ bid $r_i[t+1]$. At each time $t\leq\tau$, by \cref{lem:expected_payment_of_agent},
\begin{align*}
    \E\left[\sum_{i'=1}^{k_i}\hat P_{(i,i')}[t]\,\middle|\,\mathcal G_{t-1}\right] & = \paymentConstant\E\left[\frac{\sum_{i'=1}^{k_i}\hat r_{(i,i')}[t]}{1 + \sum_{i'=1}^{k_i}\hat r_{(i,i')}[t] + \sum_{j\neq i}\sum_{j'=1}^{k_j}\hat r_{(j,j')}[t]}\,\middle|\,\mathcal G_{t-1}\right]\\
    & = \paymentConstant\E\left[\frac{\hat X}{1+\hat X + \hat Y}\right]r_i[t]%\\
    % & = \frac{k_i}{m}\cdot \frac{1}{\asymmetricRate{i}}r_i[t].
\end{align*}
where $\hat X \overset{d}= \sum_{i'=1}^{i_k}\hat V_{(i,i')}$ where $(\hat V_{(i,i')})\sim\mathcal D_{k_i,m}$ and $\hat Y\sim\Binom(m-k_i,1/m)$ are independent. The first line follows from \cref{lem:expected_payment_of_agent} and the second line follows from the mechanism and \cref{lem:simulated_distribution} (so that simulated agents $(j,j')$ for $j\neq i$ are bidding i.i.d. $\Bern(1/m)$).

Notice that $\hat X$ is the distribution of a $\Binom(k_i,1/m)$ random variable conditioned on being nonzero. Using \cref{prop:ratio_of_binomials} and substituting $\paymentConstant$, letting $X\sim \Binom(k_i,1/m)$, the above display is equal to
\begin{equation*}
\begin{split}
    \paymentConstant\cdot\E\left[\frac{X}{1+X+\hat Y}\,\middle|\,X\geq 1\right]r_i[t] = \paymentConstant\cdot\frac{\E\left[\frac{X}{1+X+\hat Y}\right]}{\Pr(X\geq 1)} = \paymentConstant\cdot\frac{\frac {k_i}m\cdot\frac{1+m(1-1/m)^{m+1}}{m+1}}{\asymmetricRate{i}}\cdot r_i[t]\\
     = \frac{k_i}{m}\cdot\frac{r_i[t]}{\asymmetricRate{i}}.
\end{split}
\end{equation*}
Using similar reasoning,
\begin{align*}
    \E[U_i[t]\mid\mathcal G_{t-1}] & = \E\left[V_i[t]\cdot \sum_{i'=1}^{k_i}\hat W_{(i,i')}[t]\right]\\
    & = \E[V_i[t]\mid r_i[t]=1]r_i[t]\sum_{i'=1}^{k_i}\E\left[\frac{\hat r_{(i,i')}[t]}{\hat r_{(i,i')}[t] + \sum_{\hat j\neq (i,i')}r_{\hat j}[t]}\,\middle|\,\mathcal G_{t-1}, r_i[t]=1\right]\\
    & =  \E[V_i[t]\mid r_i[t]=1]\E\left[\frac{\hat X}{\hat X + \hat Y}\,\middle|\,\hat X\geq 1\right]r_i[t]\\
    & = \frac{\frac{k_i}{m}\left(1-\left(1-\frac1m\right)^m\right)}{\asymmetricRate{i}}\cdot\E[V_i[t]r_i[t]]
\end{align*}
where we use \cref{prop:ratio_of_binomials} for the last equality.

By Azuma-Hoeffding applied to the $\mathcal G_t$-martingale $\sum_{s=1}^{\min\{t,\tau\}}\sum_{i'=1}^{k_i}\hat P_{(i,i')}[t] - \frac{k_i}{m}\cdot r_i[t]\cdot\min\{t,\tau\}$ with increments bounded by $\paymentConstant$,
\begin{equation}
\label{eq:asymmetric_deviating_agent_payment}
    \sum_{t=1}^\tau \sum_{i'=1}^{k_i} \hat P_{(i,i')}[t] \geq \frac{k_i}{m}\sum_{t=1}^\tau r_i[t] - \paymentConstant\sqrt{T\ln T}
\end{equation}
with probability at least $1-1/T^2$ and to the $\mathcal G_t$-martingale $\sum_{s=1}^{\min\{t,\tau\}}U_i[t] - \min\{t,\tau\}$ with increments bounded by $\bar v$ where $\bar v$ bounds the value distribution $\mathcal F_i$,
\begin{equation}
\label{eq:asymmetric_deviating_agent_utility}
    \sum_{t=1}^\tau U_i[t] \leq \sum_{t=1}^\tau\frac{\frac{k_i}{m}\left(1-\left(1-\frac1m\right)^m\right)}{\asymmetricRate{i}}\E[V_i[t]r_i[t]] + \bar v\sqrt{T\ln T}
\end{equation}
with probability at least $1-1/T^2$. Consider what happens on the event that \eqref{eq:asymmetric_deviating_agent_payment}, \eqref{eq:asymmetric_deviating_agent_utility}, and the event that $\tau \geq T - O\left(\sqrt{T\log T}\right)$, which has probability at least $1-O(1/T^2)$ by \cref{lem:asymmetric_robustness_stopping_time_lower_bound,lem:asymmetric_equilibrium_stopping_time_lower_bound}. Call this event $E$, which has probability at least $1-O(1/T^2)$ by the union bound.

Using agents' $(i,i')$ budget constraints that $\sum_{t=1}^\tau \hat P_{(i,i')} \leq T/m + \paymentConstant$, so
\begin{equation*}
    \sum_{t=1}^\tau \sum_{i'=1}^{k_i}\hat P_{(i,i')}[t] \leq \frac{k_i}{m}\cdot T + k_i\paymentConstant.
\end{equation*}
Substituting the above into \eqref{eq:asymmetric_deviating_agent_payment} and rearranging,
\begin{equation*}
    \sum_{t=1}^\tau r_i[t] \leq \left(\asymmetricRate{i}\right)T + O\left(\sqrt{T\log T}\right).
\end{equation*}
It follows from the fact that $\tau\geq T-O\left(\sqrt{T\log T}\right)$ that 
\begin{align}
    \sum_{t=1}^T r_i[t] \leq \left(\asymmetricRate{i}\right)T + O\left(\sqrt{T\log T}\right) \label{eq:asymmetric_deviating_agent_total_requests}
\end{align}
and also using the fact that $\tau\geq T-O\left(\sqrt{T\log T}\right)$ in \eqref{eq:asymmetric_deviating_agent_utility}, we obtain
\begin{align}
    \sum_{t=1}^T U_i[t] \leq \frac{\frac{k_i}{m}\left(1-\left(1-\frac1m\right)^m\right)}{\asymmetricRate{i}}\sum_{t=1}^T\E[V_i[t]r_i[t]] + O\left(\sqrt{T\log T}\right). \label{eq:asymmetric_deviating_agent_total_utility}
\end{align}
Let $\rho_i[t]:[0,\infty)\to[0,1]$ be a measurable function such that $\rho_i[t](V_i[t]) = \E[r_i[t]\mid V_i[t]]$. Let
\begin{equation*}
    \rho_i(v_i) = \frac1T\sum_{t=1}^T \rho_i[t](v_i)
\end{equation*}

Taking expectations, using \eqref{eq:asymmetric_deviating_agent_total_requests} on $E$ and noting that $r_i[t]\leq 1$ on the complement event $E^c$, we obtain
\begin{equation*}
\begin{split}
    \E_{V_i\sim\mathcal F_i}[\rho_i(V_i)] = \frac1T\sum_{t=1}^T \E_{V_i\sim\mathcal F_i}[\rho_i[t](V_i)] = \frac1T\sum_{t=1}^T \E[r_i[t]] \leq \Pr(E^c) + \frac1T\sum_{t=1}^T \E[r_i[t]\pmb1_E]\\
        \leq \Pr(E^c) + \left(\asymmetricRate{i} + O\left(\sqrt{\frac{\log T}{T}}\right)\right)\Pr(E) \leq \asymmetricRate{i}+ O\left(\sqrt{\frac{\log T}{T}}\right).
\end{split}
\end{equation*}
Thus, $\rho_i$ is a feasible solution to \eqref{eq:beta_ideal_utility_maximization_problem} with $\beta = \asymmetricRate{i} + O\left(\sqrt{\frac{\log T}{T}}\right)$. Then,
\begin{equation*}
    \E_{V_i\sim\mathcal F_i}[V_i\rho_i(V_i)] \leq v_i^*\left(\asymmetricRate{i} + O\left(\sqrt{\frac{\log T}{T}}\right)\right)\leq v_i^*\left(\asymmetricRate{i}\right) + O\left(\sqrt{\frac{\log T}{T}}\right)
\end{equation*}
where the inequality follows from the concavity of $\beta\mapsto v_i^*(\beta)$ (\cref{fact:beta_ideal_utility_concavity}). Using the above and \eqref{eq:deviating_agent_total_requests}, we compute
\begin{equation*}
\begin{split}
    \frac1T\sum_{t=1}^T \E[U_i[t]] & \leq \frac1T\sum_{t=1}^T \E[U_i[t]\pmb1_E] + \frac1T\sum_{t=1}^T \E[U_i[t]\pmb1_{E^c}]\\
    & \leq \left(\frac1T\cdot\frac{\frac{k_i}{m}\left(1-\left(1-\frac1m\right)^m\right)}{\asymmetricRate{i}}\sum_{t=1}^T\E[V_i[t]r_i[t]] + O\left(\sqrt{\frac{\log T}{T}}\right)\right)\Pr(E) + \E_{V_i\sim\mathcal F_i}[V_i\pmb1_{E^c}]\\
    & \leq \frac1T\cdot \frac{\frac{k_i}{m}\left(1-\left(1-\frac1m\right)^m\right)}{\asymmetricRate{i}}\sum_{t=1}^T\E[V_i[t]r_i[t]] + O\left(\sqrt{\frac{\log T}{T}}\right)\\
    & = \frac{\frac{k_i}{m}\left(1-\left(1-\frac1m\right)^m\right)}{\asymmetricRate{i}}\E_{V_i\sim\mathcal F_i}[V_i\rho_i(V_i)] + O\left(\sqrt{\frac{\log T}{T}}\right)\\
    & \leq \frac{\frac{k_i}{m}\left(1-\left(1-\frac1m\right)^m\right)}{\asymmetricRate{i}}\cdot v_i^*\left(\asymmetricRate{i}\right) + O\left(\sqrt{\frac{\log T}{T}}\right).
\end{split}
\end{equation*}
using the fact that the value distribution is bounded for the third inequality.

\end{proof}

Now we have essentially proved all of \cref{thm:asymmetric_fair_shares}. The robustness follows from \cref{prop:asymmetric_robustness}. The fact that each player $i$ playing a \asymmetricStrategy{i} is an approximate equilibrium follows from \cref{prop:asymmetric_equilibrium_utility_claim,lem:asymmetric_equilibrium_deviating_utility_upper_bound}. The fraction of ideal utility guaranteed at equilibrium guarantee follows from \cref{prop:asymmetric_equilibrium_utility_claim,lem:asymmetric_concavity}.
\dledit{\section{Approximating the Fair Shares in \texorpdfstring{\allPayMechanismWithasymmetricFairShares}{Asymmetric Fair Share Mechanism}} \label{sec:approximate_fair_shares}
In this section, we show that if we use fair shares $\alpha_i'$ that are close to the true fair shares $\alpha_i$ in \allPayMechanismWithasymmetricFairShares, we obtain approximate utility guarantees. This is useful because the computational resources used by \allPayMechanismWithasymmetricFairShares scales with the least common denominator of the fair shares being used, so if the $\alpha_i$ are irrational, or if the common denominator of the $\alpha_i$ is large, we can instead use fair shares $\alpha_i'$ close to the true fair shares $\alpha_i$ with a small common denominator.

Recall the definition of $\beta$-ideal utility (\cref{def:beta_ideal_utility}): agent $i$'s $\beta$-ideal utility is their ideal utility if they had fair share $\beta$. Recall that the $\beta$-ideal utility is concave in $\beta$ (\cref{fact:beta_ideal_utility_concavity}). We are interested in proving bounds with respect to the $\alpha_i$-ideal utility, what we were calling just ``the ideal utility'' previously. If we instead use fair shares $\alpha_i'$, our guarantees (e.g., \cref{thm:asymmetric_fair_shares}) would be in terms of the $\alpha_i'$-ideal utility. To obtain guarantees in terms of the $\alpha_i$-ideal utility, we lower bound the $\alpha_i'$-ideal utility in terms of the $\alpha_i$-ideal utility in the below lemma.
\begin{lemma}
    Let $\epsilon > 0$. If $\alpha_i' \geq \alpha_i - \delta$ for $\delta \leq \alpha_i\epsilon$, then $v_i^*(\alpha_i') \geq (1-\epsilon)v_i^*(\alpha_i)$
\end{lemma}
\begin{proof}
    Assume $\alpha_i' < \alpha_i$; otherwise, there is nothing to prove by the monotonicity of $\beta\mapsto v_i^*(\beta)$. By the concavity of $\beta\mapsto v_i^*(\beta)$ (\cref{fact:beta_ideal_utility_concavity}) and the fact that $v_i^*(0)=0$, 
    \[
    v_i^*(\alpha_i') \geq \frac{\alpha_i'}{\alpha_i}\cdot v_i^*(\alpha_i) \geq \frac{\alpha_i - \delta}{\alpha_i}\cdot v_i^*(\alpha_i) \geq (1-\epsilon)v_i^*(\alpha_i)
    \]
    for $\delta \leq \alpha_i\epsilon$.
\end{proof}

Using the above lemma, by using fair shares $(\alpha_i')$ such that $\alpha_i'$ differs from $\alpha_i$ by at most $\alpha_i\epsilon$ for each $i$, if agent $i$ obtains at least a $\lambda$-fraction of their $\alpha_i'$-ideal utility, they obtain at least a $(1-\epsilon)\lambda$ fraction of their $\alpha_i$-ideal utility. Since every real number is at distance at most $1/(2m')$ from a rational number with denominator at most $m'$, to obtain a $(1-\epsilon)$-approximation of the guarantees in \cref{thm:asymmetric_fair_shares}, it suffices to use approximate fair shares $\alpha_i'$ with denominators at least $1/(2\min_i\alpha_i\epsilon)$.}

\end{document}